\documentclass[11pt]{article}
\usepackage[letterpaper,margin=1in]{geometry}
\usepackage{amssymb}

\usepackage[utf8]{inputenc} % allow utf-8 input
\usepackage[T1]{fontenc}    % use 8-bit T1 fonts
\usepackage{url}            % simple URL typesetting
\usepackage{booktabs}       % professional-quality tables
\usepackage{amsfonts}       % blackboard math symbols
\usepackage{nicefrac}       % compact symbols for 1/2, etc.
\usepackage{microtype}      % microtypography
\usepackage{xcolor}         % colors

\usepackage{amsmath}
\usepackage{mathtools}
\usepackage{amsthm}

\usepackage{makecell}

\usepackage[right, outer, inline, final]{showlabels} %draft--> final for submission
\usepackage{nicematrix}

\usepackage{pict2e,picture,graphicx}
\usepackage{tikz}
\usetikzlibrary{external}
\usepackage{float}
\usetikzlibrary{shapes,backgrounds,patterns,calc,arrows,arrows.meta}
\usepackage{chemfig}

\usepackage{subcaption}
\usepackage{cases}
\usepackage{thm-restate}
\usepackage{enumitem}
\usepackage{wrapfig}
\usepackage{yfonts}
\usepackage{multirow}

\usepackage{xfrac}
\usepackage{xparse}

\usepackage{bm}
\usepackage{bbm}

\usepackage{hyperref}
\usepackage[capitalize, noabbrev]{cleveref}

\hypersetup{
	colorlinks = true,
	linkcolor = royalBlue,
	citecolor = darkGreen,
	linktocpage = true,
	urlcolor = mgreen
}

\usepackage{comment}

\usepackage{colortbl}
\usepackage{cellspace}

\usepackage{xspace}
\usepackage{algpseudocode}
\usepackage{algorithm}

\usepackage[textsize=tiny]{todonotes}
\usepackage[short, nocomma]{optidef} % optimization problems

\catcode`\_=11
\definearrow5{-x>}{%
    \CF_arrowshiftnodes{#3}%
    \CF_expafter{\draw[}\CF_arrowcurrentstyle](\CF_arrowstartnode)--(\CF_arrowendnode)%
        coordinate[midway,shift=(\CF_ifempty{#4}{225}{#4}+\CF_arrowcurrentangle:\CF_ifempty{#5}{5pt}{#5})](line@start)%
        coordinate[midway,shift=(\CF_ifempty{#4}{45}{#4}+\CF_arrowcurrentangle:\CF_ifempty{#5}{5pt}{#5})](line@end)%
        coordinate[midway,shift=(\CF_ifempty{#4}{135}{#4}+\CF_arrowcurrentangle:\CF_ifempty{#5}{5pt}{#5})](line@start@i)%
        coordinate[midway,shift=(\CF_ifempty{#4}{315}{#4}+\CF_arrowcurrentangle:\CF_ifempty{#5}{5pt}{#5})](line@end@i);
    \draw(line@start)--(line@end);%
    \draw(line@start@i)--(line@end@i);%
    \CF_arrowdisplaylabel{#1}{0.5}+\CF_arrowstartnode{#2}{0.5}-\CF_arrowendnode
}
\catcode`\_=8

\theoremstyle{plain}
\newtheorem{theorem}{Theorem}[section]
\newtheorem{proposition}[theorem]{Proposition}

\newtheorem{lemma}[theorem]{Lemma}

\theoremstyle{definition}
\newtheorem{definition}[theorem]{Definition}

\theoremstyle{remark}

\definecolor{niceRed}{RGB}{190,38,38}
\definecolor{Red2}{RGB}{219, 50, 54}
\definecolor{mgreen}{RGB}{160, 200, 140}
\definecolor{blueGrotto}{RGB}{5,157,192}
\definecolor{limeGreen}{HTML}{81B622}
\definecolor{myellow}{rgb}{0.88,0.61,0.14}
\definecolor{darkGreen}{HTML}{2E8B57}
\definecolor{navyBlueP}{HTML}{03468F}
\definecolor{Sepia}{HTML}{7F462C}
\definecolor{red2}{HTML}{1F462C}
\definecolor{orange2}{HTML}{FF8000}
\definecolor{mgray}{HTML}{ABB3B8}
\definecolor{lgray}{HTML}{E5E8E9}
\definecolor{myPurple}{RGB}{175,0,124}
\definecolor{mypurple2}{rgb}{0.8,0.62,1}
\definecolor{royalBlue}{HTML}{057DCD}
\definecolor{mpink}{HTML}{FC6C85}
\definecolor{lblue}{RGB}{74,144,226}
\definecolor{peagreen}{RGB}{152,193,39}
\definecolor{typ_navy}{HTML}{001f3f}
\definecolor{typ_blue}{HTML}{0074d9}
\definecolor{typ_aqua}{HTML}{7fdbff}
\definecolor{typ_teal}{HTML}{39cccc}
\definecolor{typ_eastern}{HTML}{239dad}
\definecolor{typ_purple}{HTML}{b10dc9}
\definecolor{typ_fuchsia}{HTML}{f012be}
\definecolor{typ_maroon}{HTML}{85144b}
\definecolor{typ_red}{HTML}{ff4136}
\definecolor{typ_orange}{HTML}{ff851b}
\definecolor{typ_yellow}{HTML}{ffdc00}
\definecolor{typ_olive}{HTML}{3d9970}
\definecolor{typ_green}{HTML}{2ecc40}
\definecolor{typ_lime}{HTML}{01ff70}
\definecolor{newgreen}{HTML}{83c702}
\definecolor{newpurp}{RGB}{97,96,121}
\usepackage[style=alphabetic,natbib=true,maxcitenames=2, maxbibnames=10,backend=bibtex]{biblatex}
\definecolor{mygreen}{rgb}{0.0, 0.5, 0.0}
\definecolor{myorange}{rgb}{0.55, 0.62, 1}
\newcommand{\nb}[3]{{\colorbox{#2}{\bfseries\sffamily\scriptsize\textcolor{white}{#1}}}{\textcolor{#2}{\sf\small\textsf{#3}}}}

\newcommand{\mat}[1]{\nb{Matteo}{blue}{#1}}

\newcommand{\bollo}[1]{\nb{Bollo}{red}{#1}}

\DeclareMathOperator*{\argmax}{arg\,max}

\newcommand{\typeSet}{\Theta}
\newcommand{\revenue}{U^{\textnormal{P}}}
\newcommand{\profit}{U^{\textnormal{P}}} % proft and revenue are the same, we changed name from revenue to profit

\newcommand{\firstType}{\theta^{\textnormal{F}}} % \firstType_i is the first type of agent i according to <
\newcommand{\lastType}{\theta^{\textnormal{L}}} % \lastType_i is the last type of agent i according to <

\newcommand{\OPT}{\textnormal{OPT}} % optimal revenue
\newcommand{\OPTC}{\textnormal{OPT}_{\textnormal{C}}} % optimal "clayrvoyant" revenue.
\newcommand{\PIC}{\textnormal{PIC}} % \OPTC / \OPT

\newcommand{\SUP}{\textnormal{SUP}} % supremum of the revenue in the multi-task problem

\newcommand{\Splus}{\mathcal{S}^+} % set of pairs (\bar t, \bar \theta) such tha the LP has solution and revenue > 0 trict (S stands for solutions)
\newcommand{\Sminus}{\mathcal{S}^-} % set of the other pairs

\newcommand{\tSet}{\mathcal{T}} % set of possible multi-task allocation vectors

\allowdisplaybreaks

\title{Matching Tasks to Agents:\\ Mechanism Design for Multi-Agent Contracts}

\author{
	Matteo Bollini \and
	Matteo Castiglioni \and
	Alberto Marchesi  \and
	Politecnico di Milano, Milan, Italy\\
	\texttt{name.surname@polimi.it}
}

\date{}

\begin{document}    
    \maketitle
    
    \begin{abstract}
We study task assignment under adverse selection and moral hazard. In the single-task setting, agents privately know their types, which determine their success probabilities and effort costs. Based on the agents' reported types, the principal selects an agent to perform the task and offers a success-contingent payment. Although optimal Bayesian contract design is generally computationally intractable and Myerson's characterization does not directly extend to settings with hidden actions and multidimensional private information, we show that the single-task problem admits a simple Myerson-style characterization. Our key insight is that incentive constraints depend only on the ratio between effort cost and success probability. We characterize an optimal mechanism and show that it can be computed in polynomial time. We also establish a tight $\Theta(1/(NK))$ worst-case bound on the ratio between the optimal profit and the full-information benchmark, where $N$ is the number of agents and $K$ is the number of types per agent. We then extend our analysis to multiple tasks, showing that an optimal mechanism may not exist. Nevertheless, for a constant number of agents, we provide a polynomial-time approximation scheme for the supremum profit.
\end{abstract}

    \noindent\textbf{AI Usage Statement.}
    % The authors used generative AI tools to assist with writing, proofreading, and literature review. All technical claims and mathematical proofs are the authors' own work and have been independently verified by the authors.
The authors used generative AI tools to assist with writing, proofreading, and literature review. All technical claims and mathematical proofs were developed and independently verified by the authors.
    \pagenumbering{gobble} 
    \clearpage
    \tableofcontents
    \renewcommand{\thmtformatoptarg}[1]{~#1}
    %\renewcommand{\listtheoremname}{List of problems}
    %\listoftheorems[ignoreall,show={problem},swapnumber]
    \newpage
    \pagenumbering{arabic}
    
    \section{Introduction}

Contract theory studies how a principal can incentivize agents to exert costly, unobservable effort through payments contingent on observable outcomes. In the classical model with a single agent and known action costs and outcome distributions, an optimal contract can be computed by linear programming~\cite{dutting2019simple}. Subsequent work has extended this model in several directions to capture more realistic environments. Two main directions are particularly relevant to our work. The first introduces private information, adding adverse selection to moral hazard: an agent can misrepresent its type before choosing an action~\cite{castiglioni2022bayesian,guruganesh2021contracts,alon2021contracts}. The second introduces multiple agents, requiring the principal to make allocation and coordination decisions about how to distribute work and incentives among them~\cite{dutting2026multi,dutting2024combinatorial,ezra_et_al:LIPIcs.ITCS.2024.44}. Understanding when these features admit simple, computationally tractable mechanisms is a central challenge for algorithmic contract design.

Even with a single agent, designing optimal contracts in Bayesian settings is particularly challenging. Indeed, optimizing contracts or menus of deterministic contracts is computationally hard in general~\cite{castiglioni2022bayesian,guruganesh2021contracts}. Allowing menus of randomized contracts can restore efficient approximation of the principal’s optimal utility~\cite{castiglioni2023designing}, albeit at the cost of greater mechanism complexity. Even when private information is one-dimensional, combining hidden types with hidden actions requires implementability conditions that go beyond the usual monotonicity condition of auction theory~\cite{alon2021contracts}. \cite{cacciamani2024multi} extend the study of Bayesian settings to multi-agent environments and provide algorithms for approximately optimal randomized contracts. Their formulation explicitly represents joint type and action profiles, whose number can grow exponentially with the number of agents. Moreover, the resulting mechanisms are obtained through general optimization procedures rather than a simple characterization of the optimal allocation and payment rules. These limitations motivate the following question:
\begin{quote}
    Which structured environments allow the principal to handle private information and multiple agents through a simple optimal mechanism?
\end{quote}

In this paper, we focus on the problem of assigning tasks to agents. A firm commissioning a project, for example, may face contractors who differ in their ability to complete it and in the cost of doing so. These characteristics may be privately known, while the effort exerted after the contract is awarded remains unobservable. A successful mechanism must therefore address two distinct incentive problems: it must elicit the information needed to select an agent and induce that agent to work. When several tasks must be assigned, these decisions are further coupled by the requirement that each agent can undertake at most one task.

We study a model of task assignment with binary effort levels and outcomes, which can be seen as a special case of \cite{alon2025multi}, in which only one agent can be assigned to each task. There are $N$ agents whose types are drawn independently from known finite distributions. For a single task, an agent's type $\theta_i$ specifies a success probability $p_{\theta_i}$ and an effort cost $c_{\theta_i}$. The principal commits to a mechanism that maps reported types to an assignment and a payment upon successful completion. A selected agent then decides whether to exert effort. Exerting effort incurs the stated cost and completes the task with the stated probability; exerting no effort produces no reward. The principal receives a unit reward for completion and seeks to maximize expected reward minus payments. In the multi-task extension, each type specifies a success probability and an effort cost for every task, and feasible assignments are matchings between agents and tasks. We seek Bayesian incentive-compatible and individually rational mechanisms which maximize the principal's utility.

\subsection{Our Contributions}

Our results reveal a marked contrast between the single-task and multi-task settings. For a single task, we characterize a simple optimal mechanism and show how to compute it efficiently for an arbitrary number of agents. With multiple tasks, the structure underlying this characterization no longer applies, and an optimal mechanism may not even exist. We therefore take an algorithmic approach and study approximation to the supremum profit when the number of agents is constant.

\paragraph{Single task: a tractable characterization} 
Our main result is a characterization of an optimal mechanism for a single task and an arbitrary number of agents, together with an efficient algorithm for computing it. The first challenge is that agents' types are two-dimensional: both effort costs and success probabilities are private information. Our key observation is that, from the perspective of an agent's incentives, these two parameters can be summarized by a single quantity: the ratio between the cost and the success probability, which we call the \emph{critical value}.
% \[
%     v_{\theta_i} := \frac{c_{\theta_i}}{p_{\theta_i}},
% \]
% which is the smallest success payment that makes effort individually rational. Dividing an agent's expected utility by its positive success probability reveals that its reporting incentives depend on its type only through this ratio.
The critical value is the success payment that makes effort individually rational. 
By rewriting the incentive-compatibility constraints in terms of the critical values, we recover a structure similar to that of procurement and reverse auctions.

This observation allows us to recover a monotonicity-based characterization. We show that every incentive-compatible mechanism has interim allocation probabilities that are nonincreasing in the critical value. Conversely, such monotonicity supports an incentive-compatible payment rule. 
The main challenge is that the agents can decide whether to exert effort or not \emph{after} receiving the contract, which introduces nonlinearity in the incentive-compatibility constraints.
We show that, without loss of profit, the success payment to a selected agent can depend only on that agent's report.
This simplifies the incentive constraints and allows us to characterize the minimum payments associated with a monotone allocation rule.

The principal's objective nevertheless depends separately on success probabilities. We incorporate these probabilities into virtual utilities that express the principal's expected profit as expected virtual surplus. Then, adapting Myerson's ironing approach~\cite{myerson1981optimal}, we obtain an optimal deterministic allocation: assign the task to an agent with the largest positive ironed virtual utility, using a fixed rule to break ties, and leave it unassigned if no such agent exists. The mechanism can be computed and evaluated in time polynomial in the explicit descriptions of the agents' type distributions, without enumerating all joint type profiles.

Interestingly, the same allocation admits a dominant-strategy incentive-compatible implementation with the same expected profit. Thus, in the single-task setting, neither randomized allocations nor the weaker Bayesian reporting requirement improves the principal's optimal profit. The simultaneous presence of hidden types, hidden effort, and multiple agents therefore remains compatible with a simple optimal mechanism.

\paragraph{Profit relative to full information} 
Beyond designing optimal mechanisms, we investigate how much profit the principal loses because agents' types are private. We compare the optimal incentive-compatible profit with a clairvoyant benchmark that observes agents' types but must still incentivize effort. With $N$ agents and $K$ types per agent, we establish
\[
    \mathrm{OPT} \geq \frac{\mathrm{OPT}^{C}}{NK},
\]
where $\mathrm{OPT}^{C}$ is the clairvoyant profit. This guarantee can already be achieved by a particularly simple mechanism: choose one agent and one success payment in advance, and assign the task to that agent whenever its reported critical value does not exceed the payment. Furthermore, this ratio is optimal in the worst case, up to constants.

\paragraph{Multiple tasks: nonexistence of an optimal mechanism and approximation} 
With multiple tasks, the problem becomes more challenging. Each type now induces a vector of task-specific critical values, and different tasks can induce different orderings of an agent's types. Consequently, the scalar ordering that underlies our single-task characterization does not extend directly and simple mechanisms are unlikely to exist.  We therefore take an algorithmic approach. Since an explicit representation of a mechanism over all joint type profiles can require space exponential in the number of agents, we focus on the setting in which this number is constant.

We first show that an optimal mechanism may not exist. We exhibit an instance with two agents, two tasks, and two types per agent in which no profit-maximizing mechanism exists. The principal can improve profit by making certain assignments increasingly rare while increasing their associated success payments. This preserves the incentives needed for truthful reporting while reducing the cost of those assignments. As the probabilities of these assignments tend to zero, the profit approaches a supremum that no feasible mechanism attains.

We therefore turn to mechanisms that approximate this supremum. For a constant number of agents, we develop a linear-programming approach that computes, for any $\epsilon\in(0,1)$, a Bayesian incentive-compatible and individually rational mechanism with profit at least $(1-\epsilon)\mathrm{SUP}$ in time polynomial in the input size and $\log(1/\epsilon)$, where $\mathrm{SUP}$ denotes the supremum profit.
The main challenge is to establish an upper bound $B$ on the payments employed by such a mechanism, as in general payments tend to infinity as $\epsilon$ decreases. Given this upper bound, we show that the problem of computing an optimal mechanism with bounded payments reduces to a linear program.

% \mat{AI generated to do.-------------}
% The formulation uses assignment probabilities and probability-weighted success payments. The main challenge is to ensure that positive weighted payments correspond to assignments made with positive probability, so that they can be implemented by finite success payments. We address this challenge through a perturbation argument that yields near-optimal mechanisms with payments of polynomially bounded encoding length. This allows us to impose a finite payment bound and recover a mechanism from a polynomial-size linear program. Finally, randomization over assignments suffices: conditional on the reported types and the realized assignment, success payments can be deterministic.
% \mat{-----------------------}

\subsection{Related Work}

The economic foundations of contracting under hidden information and hidden actions are developed extensively in the literature \cite{bolton2004contract, carroll2015robustness, bajari2001incentives,gottlieb2015simple}. An important theme is understanding when simple contracts perform well. \cite{carroll2015robustness} provides a robustness-based justification for linear contracts, while \cite{dutting2019simple} study their performance relative to optimal contracts from an approximation perspective. The survey by \cite{dutting2024survey} provides a broad overview of the computational questions that arise in contract design.

A closely related line of work studies Bayesian contract design with a  single agent, where the principal is uncertain about the agent's characteristics. \cite{alon2021contracts} consider a single-dimensional model in which the agent's private type determines its cost per unit of effort, and characterize which allocation rules can be implemented. \cite{alon2023bayesian} further investigate the performance of linear contracts in this setting. More general models allow private information to affect both action costs and outcome distributions~\cite{guruganesh2021contracts,castiglioni2022bayesian}. These models provide negative computational results and raise questions about the value of offering menus rather than a single contract. \cite{guruganesh2023power} quantify the gains from menus and randomization, while \cite{castiglioni2023designing} show that menus of randomized contracts permit efficient approximation of the supremum profit despite the hardness of optimizing deterministic menus.
Recent work clarifies the relationship between single-dimensional and multidimensional private information. ~\cite{castiglioni2025reduction} give an almost approximation-preserving polynomial-time reduction from multi-parameter to single-parameter Bayesian contract design, showing that one-dimensional private information alone does not ensure tractability. \cite{bernasconi2025singledimensional} provide an additive polynomial-time
approximation scheme for single-dimensional contract design and establish learning guarantees under suitable assumptions. Our single-task result relies on a different structural feature. With binary effort and success-contingent payments, reporting incentives depend only on the ratio between cost and success probability, even though the principal's objective depends on both parameters. This permits an exact characterization for multiple competing agents.

A related line of research studies how a principal can incentivize several agents whose actions jointly determine its reward. Here, the main focus is on succinctly represented settings without Bayesian uncertainty.
Another relevant strand studies combinatorial contracts, in which the principal's reward depends on a set of actions or on the joint effort of several agents. \cite{dutting2026multi} provide approximation algorithms for multi-agent contracts with submodular and XOS rewards, together with impossibility results for more general subadditive rewards. \cite{ezra_et_al:LIPIcs.ITCS.2024.44} establish further approximation hardness results, including the absence of a polynomial-time approximation scheme for the multi-agent submodular setting. \cite{duetting2024multiagent} extend this analysis to teams whose members can each take multiple actions. Our model has a different source of combinatorial structure: the principal chooses a matching between agents and tasks, and its reward is additive across completed tasks. The main difficulty is to combine this assignment decision with truthful reporting and effort incentives.

Particularly close to our work are models combining multiple agents with private information or individually observed outcomes. \cite{CastiglioniEtAl23} study contracts in which each agent generates an individual outcome and the principal's reward depends jointly on these outcomes, including Bayesian settings with privately informed agents. \cite{cacciamani2024multi} consider general multi-agent contracts beyond binary actions and obtain algorithms for approximately optimal randomized mechanisms, also in the presence of private types. 

We focus instead on exclusive task assignments and binary effort.  The closest work in this direction is \cite{alon2025multi}, which studies the allocation of agents to multiple projects, with each agent assigned to at most one project. Their model allows several agents to contribute to the same project and assumes that costs and success functions are known to the principal. We instead restrict each task to a single agent, simplifying the allocation structure, but introduce private information about both effort costs and success probabilities. This requires mechanisms that jointly incentivize truthful reporting and effort.

    \section{Preliminaries}
In the following, we formalize the models we study in this paper.
Specifically, we first describe the single-task problem, where a principal aims to assign a single task to one of $N$ agents, and then we generalize this model to the multi-task problem, where multiple tasks have to be allocated.
\subsection{Single-Task Setting}
\label{sec:pre_single_task}
A single-task problem consists of one principal, $N$ agents, and a single task which w.l.o.g. provides a reward $r=1$ to the principal upon completion.
Each agent $i \in [N]$ has a finite set $\typeSet_i$ of $K \ge 2$ possible types.
We denote by $\typeSet = \typeSet_1 \times \typeSet_2 \times \dots \times \typeSet_N$ the set of all possible type profiles.
Given a type profile $\theta \in \typeSet$, we denote by $\theta_i$ the type of agent $i \in [N]$ and by $\theta_{-i}$ the type profile of the other agents.
At the beginning of the interaction, the principal commits to a mechanism, consisting of an allocation rule $x$ and payment scheme $\alpha$ that we will describe in detail later.
Then, a type profile $\theta \in \typeSet$ is drawn from a distribution $\mu$; we assume types to be drawn independently across agents, with $\mu_{\theta_i}>0$ being the probability of drawing type $\theta_i \in \typeSet_i$ for agent $i \in [N]$.
Each agent $i$ observes its own type $\theta_i$ and strategically reports a type $\hat\theta_i \in \typeSet_i$, possibly different from $\theta_i$, to the principal.
Given the reported types $\hat\theta \in \typeSet$, the principal selects (according to the allocation rule) at most a single agent $i \in [N]$ to assign the task to, and commits to a payment (according to the payment scheme) upon task completion.
The agent $i \in [N]$ that receives the task can either exert effort or not. 
In the first case, the agent incurs cost $c_{\theta_i} \in [0,1]$ and the task is completed with probability $p_{\theta_i} \in (0,1]$.\footnote{We assume w.l.o.g. that $p_{\theta_i} > 0$. If $p_{\theta_i}=0$, the agent would not misreport as they cannot complete the task to receive any payment, and the mechanism would not assign the task to agent $i$ since this provides profit zero.} %\mat{not convincing. Agent of type $0$ might missreport. I'm fine with the conclusion but not the "proof".}\bollo{If $p_{\theta_i}>0$, the agent would not misreport as they cannot complete the task to receive any payment, and the mechanism would not assign the task to agent $i$ since this provides profit zero. (?)}}
Notice that $c_{\theta_i}$ and $p_{\theta_i}$ refer to the agent's true type $\theta_i$, and not to the reported type $\theta'_i$.
In contrast, if the agent does not exert effort, the principal's reward is zero.

We now formalize the mechanism employed by the principal.
A mechanism $(x,\alpha)$ consists of an allocation rule $x : \typeSet \rightarrow \Delta_{\mathcal{N}}$ and a payment scheme (also called \emph{contract}) $\alpha : \typeSet \rightarrow \mathbb{R}^N_+$.
The allocation rule $x$ maps each type report $\theta \in \typeSet$ to a distribution $x(\theta)$ over the set $\mathcal{N} \coloneqq \{1,\ldots,N, \varnothing\}$; the principal gives the task to an agent $i \sim x(\theta)$, with the convention that no agent receives the task if $\varnothing$ is sampled.
For ease of notation, we denote by $x_{i,\theta}$ the probability that agent $i$ is selected given the type report $\theta$. %, and with $\allocSet$ the set of all possible allocation rules.
The payment scheme specifies the monetary transfer $\alpha_{i,\theta}$ due to agent $i \in [N]$ when such an agent receives and completes the task.
% For the ease of notation, we will denote $\alpha_i(\theta)$ with $\alpha_{i,\theta}$.
Let us remark that if agent $i'$ is not selected to complete the task, then they do not receive any payment. % \mat{why do you need two notations. Define only the pedix one if you use only that, otherwise say that you use both notations.}

%,\emph{i.e.} only the contract $\alpha_{i,\theta}$ referring to the reported type profile $\theta$ and the selected agent $i$ is actually used.

In this paper, we focus on Bayesian incentive-compatible (BIC) and individually rational (IR) mechanisms. 
% \mat{inventato tu valid? non vorrei che qualche economista si incazza}
Informally, a BIC and IR mechanism ensures that every agent is incentivized to both report their true type and exert effort if selected.
Formally, such a mechanism $(x,\alpha)$ satisfies the following constraints:\footnote{For every $r \in \mathbb{R}$, we let $[r]_+ \coloneqq \max\{r,0\}$.}
% \begin{subequations}
% \begin{align}
%     x_{i,\theta} (\alpha_{i,\theta} p_{\theta_i} -c_{\theta_i}) \ge 0 \quad \forall i \in [N], \theta \in \typeSet \label{eq:constraint_ir_single_alpha}\\
%     \sum_{\theta_{-i} \in \typeSet_{-i}} \mu_{\theta_{-i}} x_{i,\theta_i,\theta_{-i}} (\alpha_{i,\theta_i,\theta_{-i}} p_{\theta_i} -c_{\theta_i}) 
%     \ge \sum_{\theta_{-i} \in \typeSet_{-i}} \mu_{\theta_{-i}} x_{i,\theta'_i, \theta_{-i}} [\alpha_{i,\theta'_i, \theta_{-i}} p_{\theta_i} -c_{\theta_i}]_+ \quad \forall i \in [N], \theta_i \in \typeSet_i, \theta'_i \in \typeSet_{i} \label{eq:constraint_ic_single_alpha}
% \end{align}
% \end{subequations}
\begin{align}
    &x_{i,\theta} (\alpha_{i,\theta} p_{\theta_i} -c_{\theta_i})
    \ge 0
    &&\quad \forall i \in [N], \theta \in \typeSet
    \label{eq:constraint_ir_single_alpha}\\
    &\sum_{\theta_{-i} \in \typeSet_{-i}}
        \mu_{\theta_{-i}} x_{i,\theta_i,\theta_{-i}}
        (\alpha_{i,\theta_i,\theta_{-i}} p_{\theta_i} -c_{\theta_i}) \ge \nonumber\\
        &\hspace{3.5cm}
        \sum_{\theta_{-i} \in \typeSet_{-i}}
        \mu_{\theta_{-i}} x_{i,\hat\theta_i,\theta_{-i}}
        [\alpha_{i,\hat\theta_i,\theta_{-i}} p_{\theta_i} -c_{\theta_i}]_+
        &&\quad \forall i \in [N],\,
        \theta_i \in \typeSet_i,\,
        \hat\theta_i \in \typeSet_i
    \label{eq:constraint_ic_single_alpha}
\end{align}

Equation~\eqref{eq:constraint_ir_single_alpha} enforces individual rationality: whenever an agent receives the task, exerting effort provides non-negative expected utility.
% Notice that here the expectation is with respect to the success probability $p_{\theta_i}$ only: the agent decides wether to exert effort after receiving the task and the corresponding contract $\alpha_{\theta_i,\theta_{-i}}$.
Notice that the agent decides whether to exert effort or not only after having received the task and the corresponding contract $\alpha_{i,\theta_i,\theta_{-i}}$.
Equation~\eqref{eq:constraint_ic_single_alpha} enforces Bayesian incentive compatibility: the expected utility (over the types of the other agents) of reporting the true type $\theta_i$ and exerting effort must be at least as large as the utility from
misreporting $\hat\theta_i$ and exerting effort when convenient.
Notice that every term in the summation on the left-hand side of Equation~\eqref{eq:constraint_ic_single_alpha} is non-negative thanks to the IR constraint.
By contrast, on the right-hand side we enforce each term to be at least zero; this allows the agent to decide whether to exert effort after receiving the task and the contract, as type $\theta_i$ is not necessarily incentivized to exert effort by all contracts $\alpha_{i,\hat\theta_i,\theta_{-i}}$. %\mat{Dovremmo iniziare ad accennare che BIC=DSIC alla fine?}
We will also briefly consider a stronger notion of incentive compatibility, namely Dominant-Strategy Incentive Compatibility (DSIC), which imposes that reporting the true type $\theta_i$ is optimal for every possible type report of the other agents, rather than in expectation over them. As we will show in Section~\ref{sec:single_task}, the optimal BIC and DSIC profits coincide.

The goal of the principal is to maximize their own expected utility, henceforth called profit for short, which is defined as the expected collected reward minus the expected payment:
\begin{equation}
    \profit(x,\alpha) = \sum_{\theta \in \typeSet} \mu_\theta \sum_{i \in [N]} x_{i,\theta}p_{\theta_i}(1-\alpha_{i,\theta}).
\end{equation}
For ease of notation, we also let $\profit(x,\alpha|\theta) \coloneqq \sum_{j \in [N]}x_{j,\theta}p_{\theta_j}(1-\alpha_{j,\theta})$ for every $\theta \in \typeSet$.
We say that an optimal mechanism is one that maximizes the profit among all the BIC and IR ones.
Our goal is to design an optimal mechanism $(x,\alpha)$ and a corresponding algorithm to compute, for any given $\theta \in \typeSet$, the entries $(x(\theta),\alpha(\theta))$ of the mechanism in time polynomial in the instance size. 

\subsection{Multi-Task Setting}
In the multi-task problem, we need to assign $N$ agents to $M$ tasks, where w.l.o.g. we assume that each completed task provides a reward $r=1$ to the principal.\footnote{If each task $m \in [M]$ has a different reward $r_m \in [0,1]$, we can equivalently set $r_m=1$ by redefining the success probabilities as $p'_{\theta_i,m} \coloneqq p_{\theta_i,m}r_m$.} %\mat{cosi ovvio perchè è WLOG?}
Each agent can perform at most one task, and each task can be assigned to at most one agent.
As in the single-task setting, each agent $i \in [N]$ has a private type $\theta_i \in \typeSet_i$ independently drawn from distribution $\mu$.
When the agent exerts effort to complete the assigned task $t_i \in [M]$, it incurs a cost $c_{\theta_i,t_i} \in [0,1]$ and completes the task with probability $p_{\theta_i,t_i} \in [0,1]$.
The principal devises a (possibly randomized) mechanism that determines, given the type report $\theta \in \typeSet$, an allocation of the tasks and the contract for each agent.
Notably, the principal need not assign tasks independently across
agents, \emph{i.e.} the probabilities that agent $i$ receives a task $t_i$ and agent $j$ receives a task $t_j$ need not to be independent.
Therefore, to formally define a multi-task mechanism we introduce the set of feasible allocations:
\begin{equation*}
    \tSet \coloneqq \{t \in \{1,2,\dots,M,\varnothing\}^N \mid \forall i \neq j \in [N] \;\; t_i=t_j=\varnothing \;\vee\; t_i \neq t_j \}.
\end{equation*}
An allocation $t \in \tSet$ specifies the task $t_i$ assigned to each agent $i \in [N]$, with the convention that when $t_i = \varnothing$ the agent does not perform any task.
With an abuse of notation, we let $p_{\theta_i,\varnothing} = c_{\theta_i,\varnothing} \coloneqq 0$ for every possible agent's type $\theta_i \in \typeSet_i$.

A (multi-task) mechanism $(y,\sigma)$ consists of an allocation rule $y : \Theta \rightarrow \Delta_{\tSet}$ and a payment scheme $\sigma :  \Theta \times \tSet \rightarrow \mathbb{R}_+^N$.\footnote{Randomizing over payments does not improve the profit; see Proposition~\ref{prop:multi_task_no_rand_contr} in Appendix~\ref{app:no_rand}.}
The allocation rule specifies, for every type profile $\theta \in \typeSet$, the distribution over the allocations $y(\theta)$ to be employed when the agents report $\theta$; we will denote with $y(t|\theta)$ the probability assigned to allocation $t$.
The payment scheme specifies the contracts $\sigma_i(\theta,t)$ to be adopted when the agents jointly report $\theta$ and the allocation $t \sim y(\theta)$ is drawn; with an abuse of notation we will sometimes denote $\sigma_i(\theta,t)$ with $\sigma_i(t|\theta)$.
Agent $i$ will thus receive payment $\sigma_i(t|\theta)$ if the corresponding task $t_i$ is completed.\footnote{If $t_i = \varnothing$, the agent is not assigned any task and receives no payment, regardless of the value of $\sigma_i(t|\theta)$.} %; this is modeled by defining $p_{\theta_i,\varnothing}=0$.}
The profit extracted by the mechanism $(y,\sigma)$ is:
\begin{equation*}
    \profit(y,\sigma) \coloneqq \sum_{\theta \in \typeSet} \mu_\theta \sum_{t \in \tSet} y(t|\theta) \sum_{i \in [N]} p_{\theta_i,t_i}(1-\sigma_i(t|\theta)).
\end{equation*}
We say that the mechanism $(y,\sigma)$ is Bayesian Incentive-Compatible (BIC) and individually rational (IR) when the following constraints are satisfied:
% \begin{align}
%     y(t|\theta) (\sigma_i(t|\theta) p_{\theta_i,t_i} -c_{\theta_i,t_i}) \ge 0 \quad \forall i \in [N], \theta \in \typeSet, t \in \tSet \label{eq:constraint_ir_multi_task_sigma}\\
%     \sum_{\theta_{-i} \in \typeSet_{-i}} \mu_{\theta_{-i}}\sum_{t \in \tSet} y(t|\theta_i,\theta_{-i})(\sigma_i(t|\theta_i,\theta_{-i}) p_{\theta_i,t_i} -c_{\theta_i,t_i}) \ge \sum_{\theta_{-i} \in \typeSet_{-i}} \mu_{\theta_{-i}}\sum_{t \in \tSet} y(t|\hat\theta_i,\theta_{-i})[\sigma_i(t|\hat\theta_i,\theta_{-i}) p_{\theta_i,t_i} -c_{\theta_i,t_i}]_+ \quad \forall i \in [N], \theta_i \in \typeSet
%     \label{eq:contraint_ic_multi_task_sigma}
% \end{align}
\begin{align}
    &y(t|\theta) (\sigma_i(t|\theta) p_{\theta_i,t_i} -c_{\theta_i,t_i}) \ge 0 \quad &&\forall i \in [N], \theta \in \typeSet, t \in \tSet \label{eq:constraint_ir_multi_task_sigma}\\
    &\sum_{\theta_{-i} \in \typeSet_{-i}} \mu_{\theta_{-i}}\sum_{t \in \tSet} y(t|\theta_i,\theta_{-i})(\sigma_i(t|\theta_i,\theta_{-i}) p_{\theta_i,t_i} -c_{\theta_i,t_i}) \ge \nonumber\\
    &\hspace{2.5cm}\sum_{\theta_{-i} \in \typeSet_{-i}} \mu_{\theta_{-i}}\sum_{t \in \tSet} y(t|\hat\theta_i,\theta_{-i})[\sigma_i(t|\hat\theta_i,\theta_{-i}) p_{\theta_i,t_i} -c_{\theta_i,t_i}]_+ \quad &&\forall i \in [N], \theta_i,\hat\theta_i \in \typeSet_i
    \label{eq:contraint_ic_multi_task_sigma}
\end{align}

Notice that a single-task mechanism $(x,\alpha)$, as described in Section~\ref{sec:pre_single_task}, is a simplified representation of a corresponding multi-task mechanism $(y,\sigma)$ with $M=1$.
Indeed, when $M=1$ every allocation $t \in \mathcal{T}$ is uniquely identified by a single agent $i \in \mathcal{N} = [N] \cup \{\varnothing\}$ that receives the task ($i=\varnothing$ indicates that no agent receives the task), and thus a distribution $y(\theta) \in \Delta_{\tSet}$ corresponds to a distribution $x(\theta) \in \Delta_{\mathcal{N}}$.
Furthermore, for every type report $\theta \in \typeSet$ and agent $i \in [N]$, only the contract $\sigma_i(t^i|\theta)$ of the allocation $t^i$ that assigns the task to $i$ is relevant, and it is equal to $\alpha_{i,\theta}$.

    \section{An Optimal Mechanism for a Single Task}
\label{sec:single_task}
At first glance, this is a multi-parameter mechanism-design problem. Indeed, an agent $i$ has two parameters: a cost $c_{\theta_i}$ and a success probability $p_{\theta_i}$.
However, we show that it is possible to recover some classical properties of single-parameter problems, which allow us to design a mechanism by means of a Myerson-style approach.
Specifically, we show that a mechanism is IC and IR only if the allocation rule is monotone in a new parameter, the \emph{critical value}, and that any monotone allocation rule corresponds to a valid mechanism with appropriate payments.
Then, we show how to design virtual utilities accounting for both the critical values and the success probabilities, resulting in a Myerson-like optimization problem that can be solved efficiently by standard techniques.

% Without loss of generality, we assume that each success probability $p_{\theta_i}$ is strictly greater than zero.
% It is easy to see that if $p_{\theta_i} = 0$ for some $\theta_i \in \typeSet_i$, then the optimal mechanism does not assign the task to agent $i$ when it reports $\theta_i$. 
% This solution also disincentivizes agent $i$ from misreporting $\theta_i$ (since it would yeld utility zero) when it is of type $\hat\theta_i \neq \theta_i$, and from misreporting $\hat\theta_i$ when it is of type $\theta_i$ (sine type $\theta_i$ can attain only utility zero no matter the contract or the allocation probability).

\subsection{Ordering the types}
Given a mechanism $(x,\alpha)$, define $\beta_{i,\theta} \coloneqq \alpha_{i,\theta} x_{i,\theta}$ as the expected contract of agent $i \in [N]$ given type report $\theta \in \Theta$. 
It is easy to see that a mechanism can be equivalently represented in terms of $(x,\alpha)$ or $(x,\beta)$.
Specifically, we will say that $(x,\beta)$ is a valid $\beta$-representation if $\beta_{i,\theta}$ is set to zero whenever $x_{i,\theta}=0$. 
% \mat{vuoi dire che è un meccanismo che rappresenta un meccanismo vero? Non ha senso usare il $\beta$ nel nome se ho capito bene. forse valid $\beta$-representation o qualcosa del genere}\bollo{L'idea è che formalmente la tupla $(x,\beta)$ non è un mecchanismo, ma una rappresentazione alternativa ($\beta$-representation) di un mecchanismo. Poi con un abuso dico che è un mecchanismo per evitare di statement lunghi}. 
This allows us to recover the contracts by computing $\alpha_{i,\theta} = \beta_{i,\theta} / x_{i,\theta}$, with the convention that $0/0 \coloneqq 0$.
In the following we will work with the $\beta$-representation of mechanisms, and with an abuse of notation we will simply call the tuple $(x,\beta)$ a mechanism.
Finally, we let:
\begin{equation}
    x_{i,\theta_i} \coloneqq \sum_{\theta_{-i}} \mu_{\theta_{-i}} x_{i,\theta_i,\theta_{-i}} 
    \quad\text{and}\quad 
    \beta_{i,\theta_i} \coloneqq \sum_{\theta_{-i}} \mu_{\theta_{-i}} \beta_{i,\theta_i,\theta_{-i}}
\end{equation}
be the ex-interim allocation probability and ex-interim expected contract of agent $i \in [N]$ when it reports $\theta_i \in \Theta_i$. % \bollo{ex-inteirm contract non è esatto}
Using this notation, the profit can be expressed as:
\begin{equation}
    \label{eq:profit_in_beta_and_exinterim}
    \profit(x,\beta) = \sum_{\theta \in \typeSet} \mu_\theta \sum_{i \in [N]} p_{\theta_i}(x_{i,\theta}-\beta_{i,\theta}) = \sum_{i \in [N]} \sum_{\theta_i \in \typeSet_i} \mu_{\theta_i} p_{\theta_i}(x_{i,\theta_i} - \beta_{i,\theta_i}).
\end{equation}
Furthermore, we can rewrite the IR and BIC constraints as follows:
\begin{align*}
    \beta_{i,\theta} p_{\theta_i} -x_{i,\theta}c_{\theta_i} &\ge 0 \quad &&\forall i \in [N], \theta \in \typeSet \\
    \beta_{i,\theta_i} p_{\theta_i} -x_{i,\theta_i}c_{\theta_i} &\ge \sum_{\theta_{-i} \in \typeSet_{-i}} \mu_{\theta_{-i}} [\beta_{i,\theta'_i, \theta_{-i}} p_{\theta_i} -x_{i,\theta'_i, \theta_{-i}}c_{\theta_i}]_+ \quad &&\forall i \in [N], \theta_i \in \typeSet_i, \theta'_i \in \typeSet_i.
\end{align*}
Notice that we can rewrite the left-hand side of Equation~\ref{eq:constraint_ic_single_alpha} by means of ex-interim allocation probabilities and contracts, since it represents the expected utility of agent $i$ with respect to the types of the other agents.
However, the right-hand side cannot be written in terms of the ex-interim values, as the agent selects their action only after receiving the task and the linear contract.
Similarly, the ex-interim values cannot be used to write the IR constraints.

We are now ready to introduce the concept of critical value.
\begin{definition}
    \label{def:critical_value}
    The critical value of agent $i \in [N]$ with type $\theta_i \in \typeSet_i$ is $v_{\theta_i} \coloneqq \frac{c_{\theta_i}}{p_{\theta_i}}$.
\end{definition}
The critical value $v_{\theta_i}$ corresponds to the minimum payment necessary to satisfy the corresponding IR constraint, \emph{i.e.} an agent of type $\theta_i$ gains non-negative utility by exerting effort if the payment is at least $v_{\theta_i}$. 
This value allows us to simplify the structure of the IR and IC constraints.
By dividing both sides by $p_{\theta_i}$, we rewrite them in terms of the critical values of the agents as follows:
\begin{align}
    \beta_{i,\theta} -x_{i,\theta}v_{\theta_i} &\ge 0 \quad \forall i \in [N], \theta \in \typeSet \label{eq:ir_with_critical_value}\\
    \beta_{i,\theta_i} -x_{i,\theta_i}v_{\theta_i} &\ge \sum_{\theta_{-i} \in \typeSet_{-i}} \mu_{\theta_{-i}} [\beta_{i,\theta'_i, \theta_{-i}} -x_{i,\theta'_i, \theta_{-i}}v_{\theta_i}]_+ \quad \forall i \in [N], \theta_i \in \typeSet_i, \theta'_i \in \typeSet_i. \label{eq:ic_with_critical_value}
\end{align}
Let us remark that the principal's profit depends instead on the success probabilities. As a result, the problem remains multi-parameter.
The crucial observation is that the IR and IC constraints are single-parameter: they depend on the critical values, but not on any other parameter (success probability or costs) associated with the types. 
This allows us to recover monotonicity properties.
Indeed, Equation~\eqref{eq:ic_with_critical_value} is similar to the IC constraint in procurement problems (see, e.g.,~\cite{laffont1991auction}). In fact, it can be seen as a stricter version of such a constraint.
Therefore, it can be employed to prove that the allocation rule $x$ must satisfy a monotonicity constraint, as formalized in the following.\footnote{The proof can be found in Appendix~\ref{app:single_task}.}
\begin{restatable}{lemma}{MonotoneNecessaryLemma}
    \label{lem:monotone_necessary}
    Let $(x,\beta)$ be an IR and BIC mechanism. 
    Then, for every agent $i$ and types $\theta_i$, $\hat \theta_i\in \Theta_i$ %type profiles $\theta_{-i},\hat\theta_i  \in \typeSet_{-i}$
     it holds that $x_{i,\theta_i} \ge x_{i,\hat \theta_i}$ if $v_{\theta_i} < v_{\hat\theta_i}$.
\end{restatable}

Thanks to Lemma~\ref{lem:monotone_necessary}, we can order the types of agent $i \in [N]$ in ascending $v_{\theta_i}$.
The ex-interim allocation probability $x_{i,\theta_i}$ is then monotone non-increasing in this order.
However, this order is partial: if two types share the same critical value, then Lemma~\ref{lem:monotone_necessary} does not specify an order between them.
Nonetheless, it is natural to expect that the type with the largest success probability $p_{\theta_i}$ should also have the largest allocation probability, as it will contribute more to the principal's profit.
Formally, we argue that there always exists an optimal mechanism such that $x_{i,\theta_i} \ge x_{i,\theta'_i}$ whenever $v_{\theta_i} < v_{\theta'_i}$ or the two critical values are the same and $p_{\theta_i} > p_{\theta'_i}$.

In order to do so, we show that it is possible to focus on a subclass of mechanisms which admits a simpler formulation of the IC constraints.
We formally define these mechanisms as follows:\footnote{We adopt the convention that $0/0 \coloneqq 0$.} %\bollo{trovare un nome sensato}
\begin{definition}
    \label{def:regular_emchanism}
    A mechanism $(x,\beta)$ is regular if $\beta_{i,\theta_i,\theta_{-i}} = \frac{\beta_{i,\theta_i}}{x_{i,\theta_i}}x_{i,\theta_i,\theta_{-i}}$ for every agent $i \in [N]$ and type report $(\theta_i,\theta_{-i}) \in \typeSet$.
\end{definition}
Let us observe that a regular mechanism is fully specified by the (ex-post) allocation rule $x$ and the ex-interim payments $\beta_{i,\theta_i}$, with $i \in [N]$ and $\theta_i \in \typeSet_i$. 
So long as $\beta_{i,\theta_i}=0$ whenever $x_{i,\theta_i}=0$, we can recover the expected contracts $\beta_{i,\theta_i,\theta_{-i}}$ from the ex-interim values. 
% By definition, the expected contracts are $\beta_{i,\theta_i,\theta_{-i}} = \frac{\beta_{i,\theta_i}}{x_{i,\theta_i}}x_{i,\theta_i,\theta_{-i}}$.
Furthermore, if we look at the $(x,\alpha)$ representation of a regular mechanism, we have $\alpha_{i,\theta_i,\theta_{-i}} = \frac{\beta_{i,\theta_i}}{x_{i,\theta_i}}$, \emph{i.e.} the contract received by agent $i$ depends only on its reported type $\theta_i$, while the types of the others influence only the probability of receiving the task.
Hence, the IR constraint for a regular mechanism can be written in terms of ex-interim allocation probabilities and contracts, namely $\frac{\beta_{i,\theta_i}}{x_{i,\theta_i}} \ge v_{\theta_i}$ for every type $\theta_i \in \typeSet_i$ such that $x_{i,\theta_i} > 0$.
Moreover, if the agent misreports $\theta'_i \neq \theta_i$ and receives the task, they receive the same contract $\frac{\beta_{i,\theta'_i}}{x_{i,\theta'_i}}$, regardless of the types of the others. The IC constraint (Equation~\eqref{eq:ic_with_critical_value}) can thus be formulated without the max operator.
As a result, we can rewrite the constraints in terms of the ex-interim probabilities and contracts only. Formally, the following result holds.
\begin{restatable}{lemma}{RegularICIRLemma}
    \label{lem:regular_ic_and_ir}
    % A regular mechanism $(x,\beta)$ is IR and BIC if and only if:
    % \begin{align*}
    %     \beta_{i,\theta_i} p_{\theta_i} -x_{i,\theta_i}c_{\theta_i} &\ge 0 \quad &&\forall i \in [N], \theta_i \in \typeSet_i \\
    %     \beta_{i,\theta_i} p_{\theta_i} -x_{i,\theta_i}c_{\theta_i} &\ge \beta_{i,\theta'_i} p_{\theta_i} -x_{i,\theta'_i}c_{\theta_i} \quad &&\forall i \in [N], \theta_i,\theta'_i \in \typeSet_i,
    % \end{align*}
    % or, equivalently:
    % \begin{align*}
    %     \beta_{i,\theta_i} -x_{i,\theta_i}v_{\theta_i} &\ge 0 \quad &&\forall i \in [N], \theta_i \in \typeSet_i \\
    %     \beta_{i,\theta_i} -x_{i,\theta_i}v_{\theta_i} &\ge \beta_{i,\theta'_i}-x_{i,\theta'_i}v_{\theta_i} \quad &&\forall i \in [N], \theta_i,\theta'_i \in \typeSet_i.
    % \end{align*} 
    A regular mechanism $(x,\beta)$ is IR and BIC if and only if:
    \begin{align}
        \beta_{i,\theta_i} -x_{i,\theta_i}v_{\theta_i} &\ge 0 \quad &&\forall i \in [N], \theta_i \in \typeSet_i \label{eq:regular_ir}\\
        \beta_{i,\theta_i} -x_{i,\theta_i}v_{\theta_i} &\ge \beta_{i,\theta'_i}-x_{i,\theta'_i}v_{\theta_i} \quad &&\forall i \in [N],  \theta_i,\theta'_i \in \typeSet_i. \label{eq:regular_ic}
    \end{align} 
\end{restatable}
\begin{proof}
    We first focus on the IR constraints.
    Fix an agent $i \in [N]$ and its type $\theta_i \in \typeSet_i$.
    On the one hand, if the IR constraint is satisfied, then it also holds that $\beta_{i,\theta_i} p_{\theta_i} -x_{i,\theta_i}c_{\theta_i} \ge 0$, by the linearity of expectation.
    % On the other hand, for every $\theta_{-i} \in \typeSet_{-i}$, if $\beta_{i,\theta_i} p_{\theta_i} -x_{i,\theta_i}c_{\theta_i} \ge 0$, either $x_{i,\theta_i,\theta_{-i}} = \beta_{i,\theta_i,\theta_{-i}} = 0$ and the IR constraint is satisfied, or:
    On the other hand, if $\beta_{i,\theta_i} p_{\theta_i} -x_{i,\theta_i}c_{\theta_i} \ge 0$, then for every $\theta_{-i} \in \typeSet_{-i}$:
    \begin{align*}
        \beta_{i,\theta_i,\theta_{-i}} p_{\theta_i} -x_{i,\theta_i,\theta_{-i}}c_{\theta_i} &= \frac{\beta_{i,\theta_i}}{x_{i,\theta_i}}x_{i,\theta_i,\theta_{-i}} p_{\theta_i} -x_{i,\theta_i,\theta_{-i}}c_{\theta_i} = \frac{x_{i,\theta_i,\theta_{-i}}}{x_{i,\theta_i}}(\beta_{i,\theta_i} p_{\theta_i} -x_{i,\theta_i} c_{\theta_i}) \ge 0,
    \end{align*}
    where the first equality holds by Definition~\ref{def:regular_emchanism}.
    As a result, the IR constraints are satisfied if and only if Equation~\ref{eq:regular_ir} is verified.
    
    Consider now the BIC constraints. 
    Fix an agent $i \in [N]$ and two types $\theta_i,\theta'_i \in \typeSet_i$.
    Let $\alpha_{i,\theta_i,\theta_{-i}} = \alpha_{i,\theta_i} \coloneqq \frac{\beta_{i,\theta_i}}{x_{i,\theta_i}}$ for every $\theta_{-i} \in \typeSet_{-i}$.
    From the definition of regular mechanism:
    \begin{align*}
        \sum_{\theta_{-i} \in \typeSet_{-i}} \mu_{\theta_{-i}} [\beta_{i,\theta'_i, \theta_{-i}} p_{\theta_i} -x_{i,\theta'_i, \theta_{-i}}c_{\theta_i}]_+ &= \sum_{\theta_{-i} \in \typeSet_{-i}} \mu_{\theta_{-i}} [\alpha_{i,\theta'_i}x_{i,\theta'_i,\theta_{-i}} p_{\theta_i} -x_{i,\theta'_i, \theta_{-i}}c_{\theta_i}]_+ \\
        &= \sum_{\theta_{-i} \in \typeSet_{-i}} \mu_{\theta_{-i}} x_{i,\theta'_i,\theta_{-i}}[\alpha_{i,\theta'_i} p_{\theta_i} -c_{\theta_i}]_+ \\
        &= \left[\sum_{\theta_{-i} \in \typeSet_{-i}} \mu_{\theta_{-i}} x_{i,\theta'_i,\theta_{-i}}(\alpha_{i,\theta'_i} p_{\theta_i} -c_{\theta_i}) \right]_+ \\
        &= [\beta_{i,\theta'_i} p_{\theta_i} -x_{i,\theta'_i}c_{\theta_i}]_+.
    \end{align*}
    Therefore, the BIC constraints can be rewritten as:
    \begin{align*}
        \beta_{i,\theta_i} p_{\theta_i} -x_{i,\theta_i}c_{\theta_i} &\ge [\beta_{i,\theta'_i} p_{\theta_i} -x_{i,\theta'_i}c_{\theta_i}]_+ \quad \forall i \in [N], \theta_i,\theta'_i \in \typeSet_i.
    \end{align*}
    Noticing that the IR constraints imply that $\beta_{i,\theta_i} p_{\theta_i} -x_{i,\theta_i}c_{\theta_i} \ge 0$, the max operator can be removed from the inequality above.
    Dividing both sides by $p_{\theta_i}$ concludes the proof.
    % Therefore, the IR and IC constraints can be rewritten as:
    % \begin{align*}
    %     \beta_{i,\theta} p_{\theta_i} -x_{i,\theta}c_{\theta_i} &\ge 0 \quad \forall i \in [N], \theta \in \typeSet \\
    %     \beta_{i,\theta_i} p_{\theta_i} -x_{i,\theta_i}c_{\theta_i} &\ge [\beta_{i,\theta'_i} p_{\theta_i} -x_{i,\theta'_i}c_{\theta_i}]_+ \quad \forall i \in [N], \theta_i,\theta'_i \in \typeSet_i.
    % \end{align*}
    % Noticing that the IR constraints implies: $\beta_{i,\theta_i} p_{\theta_i} -x_{i,\theta_i}c_{\theta_i} \ge 0$, the IC constraint can be rewritten as:
    % \begin{equation*}
    %     \beta_{i,\theta_i} p_{\theta_i} -x_{i,\theta_i}c_{\theta_i} \ge \beta_{i,\theta'_i} p_{\theta_i} -x_{i,\theta'_i}c_{\theta_i} \quad \forall i \in [N], \theta_i,\theta'_i \in \typeSet_i,
    % \end{equation*}
    % since the left hand is non-negative.
    % Dividing both sides of the inequalities above by $p_{\theta_i}$ concludes the proof.
\end{proof}

We can now show that we can focus on regular mechanisms without loss of generality, as they are optimal.
\begin{restatable}{lemma}{RegualarOptimal}
    \label{lem:regular_optimal}
    There always exists an optimal BIC and IR regular mechanism.
\end{restatable}
\begin{proof}
    Let $(x,\beta)$ be an optimal valid mechanism, and let $(x,\beta')$ be the regular mechanism with the same allocation rule $x$ and the payment scheme such that $\beta'_{i,\theta_i} = \beta_{i,\theta_i}$ for every agent $i \in [N]$ and type $\theta_i \in \typeSet_i$.
    Notice that this is well defined, as it is sufficient to set $\beta'_{i,\theta_i,\theta_{-i}} = \frac{\beta_{i,\theta_i}}{x_{i,\theta_i}}x_{i,\theta_i,\theta_{-i}}$, with the convention that $0/0 \coloneqq 0$.
    By Equation~\eqref{eq:profit_in_beta_and_exinterim}, it follows that $\profit(x,\beta) = \profit(x,\beta')$.
    We now show that it is IR and BIC by means of Lemma~\ref{lem:regular_ic_and_ir}.
    Fix an agent $i \in [N]$ of type $\theta_i \in \typeSet_i$.
    Since $(x,\beta)$ satisfies the IR constraint for every type report $(\theta_i,\theta_{-i}) \in \typeSet$, it must be the case that $\beta_{i,\theta_i} p_{\theta_i} -x_{i,\theta_i}c_{\theta_i} \ge 0$.
    As $(x,\beta')$ has the same ex-interim allocation probabilities and contracts, it must be IR by Lemma~\ref{lem:regular_ic_and_ir}.
    % We have that:
    % \begin{align*}
    %     \beta'_{,\theta_i,\theta_{-i}} = \frac{\beta_{i,\theta_i}}{x_{i,\theta_i}}x_{i,\theta_i,\theta_{-i}} \ge \frac{x_{i,\theta_i} v_{\theta_i}}{x_{i,\theta_i}}x_{i,\theta_i,\theta_{-i}} = x_{i,\theta_i} v_{\theta_i},
    % \end{align*}
    % for every type report $\theta_{-i} \in \typeSet_{-i}$, where the inequality holds because otherwise $(x,\alpha)$ would not be IR. As a result.
    Now consider the IC constraints, which, by Lemma~\ref{lem:regular_ic_and_ir} and the fact that $\beta'_{i,\theta_i} = \beta_{i,\theta_i}$, can be expressed as:
    \begin{equation*}
        \beta_{i,\theta_i} p_{\theta_i} -x_{i,\theta_i}c_{\theta_i} \ge \beta_{i,\theta'_i} p_{\theta_i} -x_{i,\theta'_i}c_{\theta_i}
    \end{equation*}
    for every other type $\theta'_i \in\typeSet_i$.
    This is verified since:
    \begin{align*}
        \beta_{i,\theta_i} p_{\theta_i} -x_{i,\theta_i}c_{\theta_i} &\ge \sum_{\theta_{-i} \in \typeSet_{-i}} \mu_{\theta_{-i}} [\beta_{i,\theta'_i, \theta_{-i}} p_{\theta_i} -x_{i,\theta'_i, \theta_{-i}}c_{\theta_i}]_+ \\
        &\ge \sum_{\theta_{-i} \in \typeSet_{-i}} \mu_{\theta_{-i}} \beta_{i,\theta'_i, \theta_{-i}} p_{\theta_i} -x_{i,\theta'_i, \theta_{-i}}c_{\theta_i} \\
        &= \beta_{i,\theta'_i} p_{\theta_i} -x_{i,\theta'_i}c_{\theta_i},
    \end{align*}
    where the first inequality holds by the BIC constraints satisfied by $(x,\beta)$.
\end{proof}

Finally, we show that the allocation rule of the optimal regular mechanism is monotone over a suitably defined total ordering over the types.
\begin{restatable}{lemma}{TotalOrderMonotonicityLemma}
    \label{lem:total_monotonicity}
    There always exists an optimal regular mechanism $(x,\beta)$ such that, for every agent $i$ and two types $\theta_i, \hat\theta_i \in \typeSet_i$, it holds that $x_{i,\theta_i} \ge x_{i,\hat\theta_i}$ whenever $v_{\theta_i} < v_{\hat\theta_i}$ or $v_{\theta_i} = v_{\hat\theta_i}$ and $p_{\theta_i} > p_{\hat\theta_i}$.
    % notice that it holds for basically every optimum. The only degenerate cases where it holds some optimum does not use this tie breaking requires some type with c=p=0
\end{restatable}
\begin{proof}
    According to Lemma~\ref{lem:monotone_necessary}, the allocation rule of a BIC mechanism is monotone in ascending $v_{\theta_i}$.
    Now let $(x,\alpha)$ be a regular profit-maximizing BIC and IR mechanism (which exists by Lemma~\ref{lem:regular_optimal}), and let $(x,\beta)$ be its $\beta$-representation.
    Let $\tilde\theta_i, \hat \theta_i \in \typeSet_i$ be two types of agent $i \in [N]$ such that $v_{\tilde\theta_i} = v_{\hat\theta_i} \eqqcolon v$ and $p_{\tilde\theta_i} > p_{\hat\theta_i}$.
    If $v \ge 1$, then one can build a mechanism with weakly larger profit, by setting $x_{i,\theta_i} =0$ for types $\theta_i$ with critical value $v_{\theta_i} \ge 1$, and reducing the payments for the other types.
    % Don't make me prove this formally, it becomes a mess: you need to use the fact that the payments that we compute later are optimal for any monotone allocatio rule in ascending order of critical value and \emph{any} tie-breaking rule (sicne we do not use the tie-breakign rule there). So take a mechanism that uses types with $v\ge1$, compute its optimal payments, then change the allocation rule, recompute payments, and the new payments have decreased. Note: giving the task with v>1 incur a loss of profit because you pay more than one.
    Therefore, we consider the case where $v < 1$.
    Suppose, by contradiction, that $x_{i,\tilde\theta_i} < x_{i,\hat\theta_i}$.
    We now create a new mechanism $(x',\beta')$ that is BIC and IR, and $\profit(x',\beta') > \profit(x,\alpha)$, reaching a contradiction.

    As a first step, let us analyze the mechanism $(x,\beta)$ (under the  contradictory hypothesis $x_{i,\tilde\theta_i} < x_{i,\hat\theta_i}$) in order to recover two properties that we will need later. 
    Let $\Delta_{\tilde\theta_i} \coloneqq x_{i,\tilde\theta_i} -\beta_{i,\tilde\theta_i}$ and $\Delta_{\hat\theta_i} \coloneqq x_{i,\hat\theta_i} -\beta_{i,\hat\theta_i}$.
    Thanks to the IC constraints (specifically Lemma~\ref{lem:regular_ic_and_ir}) between the types $\tilde\theta_i$ and $\hat\theta_i$ that share the same critical value, we have:
    \begin{equation*}
        \beta_{i,\tilde\theta_i} - x_{i,\tilde\theta_i} v \ge \beta_{i,\hat\theta_i} -x_{i,\hat\theta_i} v \ge \beta_{i,\tilde\theta_i} -x_{i,\tilde\theta_i} v,
    \end{equation*}
    where the first inequality disincentivizes misreporting $\hat\theta_i$ in place of $\tilde\theta_i$, and the second the other way around.
    Therefore:
    \begin{equation}
        \label{eq:proof_same_utility_at_same_v}
        \beta_{i,\tilde\theta_i} -x_{i,\tilde\theta_i} v = \beta_{i,\hat\theta_i} -x_{i,\hat\theta_i} v,
    \end{equation}
    which gives: $\beta_{i,\hat\theta_i} = \beta_{i,\tilde\theta_i} +(x_{i,\hat\theta_i}-x_{i,\tilde\theta_i}) v$ and $\beta_{i,\tilde\theta_i} = \beta_{i,\hat\theta_i} +(x_{i,\tilde\theta_i} -x_{i,\hat\theta_i})v$.
    Hence:
    \begin{equation}
    \label{eq:proof_delta_hat_ge_delta_tilde}
        % \Delta_{\hat\theta} = x_{i,\hat\theta_i} -\beta_{i,\hat\theta_i} = x_{i,\hat\theta_i}-\beta_{i,\tilde\theta_i} +(x_{i,\hat\theta_i}-x_{i,\tilde\theta_i}) v \ge x_{i,\hat\theta_i}-\beta_{i,\tilde\theta_i} > x_{i,\tilde\theta_i}-\beta_{i,\tilde\theta_i},
        \Delta_{\hat\theta} -\Delta_{\tilde\theta} = x_{i,\hat\theta_i} -\beta_{i,\tilde\theta_i} -(x_{i,\hat\theta_i}-x_{i,\tilde\theta_i}) v  -x_{i,\tilde\theta_i} +\beta_{i,\tilde\theta_i} = (x_{i,\hat\theta_i}-x_{i,\tilde\theta_i})(1-v) > 0.
    \end{equation}
    where the inequality is strict as $x_{i,\hat\theta_i} > x_{i,\tilde\theta_i}$ and $(1-v) > 0$.

    Now we build a new regular mechanism $(x',\beta')$.
    Intuitively, we keep the same allocation rule when agent $i$ does not report $\tilde\theta_i$ or $\hat\theta_i$, namely $x'_{\theta_i,\cdot}=x_{\theta_i,\cdot}$ for every $\theta_i \notin \{\tilde\theta_i,\hat\theta_i\}$.
    For the remaining two types of agent $i$, we build $x'_{\tilde\theta_i,\cdot}=x'_{\hat\theta_i,\cdot}$ as a suitable convex combination of the two original allocation rules. 
    % Basicaly, use the linear combination of the tow lines. Then change the contracts of agent i to get a regualr mehcanism while keeping the same ex-interim probabilities. Notice also that doing this does not change the contracts of the other agents $j\ neq i$.
    Specifically, we let:
    \begin{align*}
        \tilde a \coloneqq \frac{\mu_{\tilde\theta_i}}{\mu_{\tilde\theta_i} +\mu_{\hat\theta_i}},
        \quad
        \hat a \coloneqq \frac{\mu_{\hat\theta_i}}{\mu_{\tilde\theta_i} +\mu_{\hat\theta_i}}
    \end{align*}
    and then we complete the definition of the allocation rule with:
    \begin{align*}
        x'_{j,\tilde\theta_i,\theta_{-i}} &\coloneqq x'_{j,\hat\theta_i,\theta_{-i}} \coloneqq \tilde a x_{j,\tilde\theta_i,\theta_{-i}} +\hat a x_{j,\hat\theta_i,\theta_{-i}},
    \end{align*}
    for every agent $j \in [N]$ (including agent $i$) and type report $\theta_{-i} \in \typeSet_{-i}$.
    Notice that $\sum_j x_{j,\theta} \le 1$ for every $\theta \in \typeSet$, ensuring that $x'$ is an allocation rule.
    Furthermore, the agents $j \neq i$ see their ex-interim allocation probabilities unchanged.
    As for the contracts, we keep the same ex-interim contracts $\beta'_{j,\theta_j} = \beta_{j,\theta_j}$ for every agent $j \neq i$ and type $\theta_j \in \Theta_j$.
    For agent $i$, we define the following ex-interim contracts:
    \begin{align*}
        \beta'_{i,\tilde\theta_i} &\coloneqq \beta'_{i,\hat\theta_i} \coloneqq \tilde a \beta_{i,\tilde\theta_i} +\hat a \beta_{i,\hat\theta_i}, \\
        \beta'_{i,\theta_i} &\coloneqq \beta_{i,\theta_i} &&\forall \theta_i \in \typeSet_i \setminus \{\tilde\theta_i,\hat\theta_i\}.
    \end{align*}
    Notice that the interim expected contracts are zero whenever the interim allocation probability is zero, and thus one can compute the expected contract $\beta'_{i,\theta_i,\theta_{-i}} = \frac{\beta'_{i,\theta_i}}{x'_{i,\theta_i}}x'_{i,\theta_i,\theta_{-i}}$ for every $(\theta_i,\theta_{-i}) \in \typeSet$. %.  \mat{check: aggiungere ' dappertutto in quest'ultima uguaglianza?}

    We now show that $(x',\beta')$ is BIC and IR.
    Since the mechanism is regular, we can employ Lemma~\ref{lem:regular_ic_and_ir} and analyze the constraints by means of the ex-interim values only.
    % To prove this, decompose (by definition) $x'_{j,\theta_j}$ and $\beta'_{j,\theta_j}$ into the weighted sum over the $\theta_{-j}$.  For every $\theta_{-j}$ whose $i$-component is not $\tilde\theta_i,\hat\theta_i$, you have the same contribution as in the originalmechanism. For every $\theta_{-j}$ whose $i$-component is $\tilde\theta_i$, there is another $\theta'_{-j}$ whose $i$-th component is $\hat\theta_i$. Sum the two components and you recover the same two compoents of the original mechanism.
    For every $j \neq i$, it is easy to verify that $x'_{j,\theta_j} = x_{j,\theta_j}$ and $\beta'_{j,\theta_j} = \beta_{j,\theta_j}$ for every $\theta_j \in \typeSet_j$; thus, the BIC and IR constraints for these agents are satisfied. 
    We now analyze agent $i$.
    When agent $i$ does not have type $\tilde \theta_i$ or $\hat \theta_i$, the IR constraints are obviously satisfied, as the ex-interim values are the same as in the original mechanism $(x,\beta)$.
    When instead the type is $\theta_i \in \{\tilde\theta_i, \hat\theta_i\}$, we have:
    \begin{align*}
        \beta'_{i,\tilde\theta_i} -x'_{i,\tilde\theta_i} v_{\tilde\theta_i}  &= \tilde a \beta_{i,\tilde\theta_i} +\hat a \beta_{i,\hat\theta_i} -(\tilde a x_{i,\tilde\theta_i} +\hat a x_{i,\hat\theta_i})v_{\theta_i} \\
        &= \tilde a (\beta_{i,\tilde\theta_i} -x_{i,\tilde\theta_i} v_{\theta_i}) +\hat a (\beta_{i,\hat\theta_i} -x_{i,\hat\theta_i} v_{\theta_i}) \ge 0,
    \end{align*}
    where the inequality holds since $(x,\beta)$ is IR. 
    Therefore, the IR constraints are satisfied.

    % \bollo{Careful: qui non è regolare, quindi c'è un max di mezzo. Esempio: se nella riga $\tilde\theta$ abbiamo $x_i=[1,0]$ e $\beta_i=[1/2,0]$, mentre sull'altra riga $x_i=[0,1]$ e $\beta_i=[0,1/4]$, una combinazione delle due da $x'=[1,1]$ e $\beta'=[\text{due cose diverse}]$, quindi nienete regolarità}
    % \bollo{Mybe I should make an extra step. "redistribute the contracts" of $\beta'$ to get $(x^\circ,\beta^\circ)$ regular, with $x^\circ=x'$ and same ex-interim contracts. This things has the same profit, so it's better than the  $(x,\beta)$, and now we can prove that it is BIC and IR by using the ex-interim values and it works.}
    To complete the proof that $(x',\beta')$ is BIC, we show that the BIC constraints for agent $i$ are also satisfied.
    Observe that, when facing $(x',\beta')$, reporting $\tilde\theta_i$ or $\hat\theta_i$ provides the same expected utility to agent $i$, as the mechanism treats the two types the same way.
    Now, suppose that agent $i$ has true type $\theta_i \in \typeSet_i \setminus \{\tilde\theta_i,\hat\theta_i\}$.
    Since the original mechanism is BIC, the agent is not incentivized to misreport $\theta'_i \notin \{\tilde\theta_i,\hat\theta_i\}$.
    Furthermore, it is not incentivized to misreport $\tilde\theta_i$ (or equivalently $\hat\theta_i$) since:
    \begin{align*}
        \beta'_{i,\tilde\theta_i} -x'_{i,\tilde\theta_i} v_{\theta_i} &= \tilde a (\beta_{i,\tilde\theta_i} -x_{i,\tilde\theta_i} v_{\theta_i}) +\hat a(\beta_{i,\hat\theta_i} -x_{i,\hat\theta_i} v_{\theta_i})  \\
        &\le \tilde a (\beta_{i,\theta_i} -x_{i,\theta_i} v_{\theta_i}) +\hat a (\beta_{i,\theta_i} -x_{i,\theta_i} v_{\theta_i}) \\
        &= \beta_{i,\theta_i} -x_{i,\theta_i} v_{\theta_i} = \beta'_{i,\theta_i} -x'_{i,\theta_i} v_{\theta_i},
    \end{align*}
    where the inequality holds thanks to the BIC constraint of mechanism $(x,\beta)$, and the second to last equality holds since $\tilde a +\hat a = 1$.
    Suppose instead that agent $i$ is of type $\theta_i \in \{\tilde\theta_i,\hat\theta_i\}$.
    Then, it is not incentivized to misreport $\theta'_i \notin \{\tilde\theta_i,\hat\theta_i\}$ as:
    \begin{align*}
        \beta'_{i,\tilde\theta_i} -x'_{i,\tilde\theta_i} v &= \beta'_{i,\hat\theta_i} -x'_{i,\hat\theta_i} v \\
        &= \tilde a (\beta_{i,\tilde\theta_i} -x_{i,\tilde\theta_i} v) +\hat a (\beta_{i,\hat\theta_i} -x_{i,\hat\theta_i} v) \\
        &= \tilde a (\beta_{i,\tilde\theta_i} -x_{i,\tilde\theta_i} v) +\hat a (\beta_{i,\tilde\theta_i} -x_{i,\tilde\theta_i} v) \\
        &= \beta_{i,\tilde\theta_i} -x_{i,\tilde\theta_i} v \\
        &\ge \beta_{i,\theta_i} -x_{i,\theta_i} v = \beta'_{i,\theta_i} -x'_{i,\theta_i} v,
    \end{align*}
    where the first two equalities  and the last one hold by the definition of $(x',\beta')$, the third thanks to Equation~\eqref{eq:proof_same_utility_at_same_v}, the fourth since $\tilde a +\hat a = 1$, and the inequality holds by the BIC constraint of mechanism $(x,\beta)$.
    As a result, we conclude that $(x',\beta')$ is BIC and IR by Lemma~\ref{lem:regular_ic_and_ir}.

    Finally, we compare the profit of the two mechanisms.
    Decomposing the profit by Equation~\eqref{eq:profit_in_beta_and_exinterim}, we observe that the only difference between the two profits consists in the contribution of agent $i$ when it is of type $\tilde\theta_i$ or $\hat\theta_i$, which we denote as $r,r'$ for the two mechanisms $(x,\beta)$ and $(x',\beta')$, respectively.
    Let $\Delta_{\tilde\theta_i} \coloneqq x_{i,\tilde\theta_i} -\beta_{i,\tilde\theta_i}$ and $\Delta_{\hat\theta_i} \coloneqq x_{i,\hat\theta_i} -\beta_{i,\hat\theta_i}$. We have:
    \begin{align*}
        r' &= \mu_{\tilde\theta_i} p_{\tilde\theta_i} (x'_{i,\tilde\theta_i} -\beta'_{i,\tilde\theta_i}) +\mu_{\hat\theta_i} p_{\hat\theta_i} (x'_{i,\hat\theta_i} -\beta'_{i,\hat\theta_i}) \\
        &= \mu_{\tilde\theta_i} p_{\tilde\theta_i} \tilde a \Delta_{\tilde\theta_i} +\mu_{\tilde\theta_i} p_{\tilde\theta_i} \hat a \Delta_{\hat\theta_i} +\mu_{\hat\theta_i} p_{\hat\theta_i} \tilde a \Delta_{\tilde\theta_i} +\mu_{\hat\theta_i} p_{\hat\theta_i} \hat a\Delta_{\hat\theta_i} \\
        &=\tilde a\Delta_{\tilde\theta_i} (\mu_{\tilde\theta_i} p_{\tilde\theta_i} +\mu_{\hat\theta_i} p_{\hat\theta_i}) +\hat a \Delta_{\hat\theta_i}(\mu_{\tilde\theta_i} p_{\tilde\theta_i} +\mu_{\hat\theta_i} p_{\hat\theta_i}) \\
        &=\frac{\mu_{\tilde\theta_i}\Delta_{\tilde\theta_i} (\mu_{\tilde\theta_i} p_{\tilde\theta_i} +\mu_{\hat\theta_i} p_{\hat\theta_i})}{\mu_{\tilde\theta_i} +\mu_{\hat\theta_i}} +\frac{\mu_{\hat\theta_i}\Delta_{\hat\theta_i}(\mu_{\tilde\theta_i} p_{\tilde\theta_i} +\mu_{\hat\theta_i} p_{\hat\theta_i})}{\mu_{\tilde\theta_i}+\mu_{\hat\theta_i}} \\
        &=\frac{\mu^2_{\tilde\theta_i} p_{\tilde\theta_i} \Delta_{\tilde\theta_i} +\mu_{\tilde\theta_i}\mu_{\hat\theta_i} p_{\hat\theta_i} \Delta_{\tilde\theta_i} +\mu_{\tilde\theta_i}\mu_{\hat\theta_i} p_{\tilde\theta_i} \Delta_{\hat\theta_i} +\mu^2_{\hat\theta_i} p_{\hat\theta_i} \Delta_{\hat\theta_i}}{\mu_{\tilde\theta_i} +\mu_{\hat\theta_i}} \\
        &=\frac{\mu^2_{\tilde\theta_i} p_{\tilde\theta_i} \Delta_{\tilde\theta_i} +\mu_{\tilde\theta_i}\mu_{\hat\theta_i}( p_{\hat\theta_i} \Delta_{\tilde\theta_i} +p_{\tilde\theta_i} \Delta_{\hat\theta_i}) +\mu^2_{\hat\theta_i} p_{\hat\theta_i} \Delta_{\hat\theta_i}}{\mu_{\tilde\theta_i} +\mu_{\hat\theta_i}}
    \end{align*}
    Consequently:
    \begin{align*}
        r' -r &= \frac{\mu^2_{\tilde\theta_i} p_{\tilde\theta_i} \Delta_{\tilde\theta_i} +\mu_{\tilde\theta_i}\mu_{\hat\theta_i}( p_{\hat\theta_i} \Delta_{\tilde\theta_i} +p_{\tilde\theta_i} \Delta_{\hat\theta_i}) +\mu^2_{\hat\theta_i} p_{\hat\theta_i} \Delta_{\hat\theta_i}}{\mu_{\tilde\theta_i} +\mu_{\hat\theta_i}} -\mu_{\tilde\theta_i}p_{\tilde\theta_i}\Delta_{\tilde\theta_i} -\mu_{\hat\theta_i}p_{\hat\theta_i}\Delta_{\hat\theta_i} \\
        &\hspace{-1cm}=\frac{\mu^2_{\tilde\theta_i} p_{\tilde\theta_i} \Delta_{\tilde\theta_i} \hspace{-0.1cm}+\hspace{-0.1cm}\mu_{\tilde\theta_i}\mu_{\hat\theta_i}( p_{\hat\theta_i} \Delta_{\tilde\theta_i} \hspace{-0.1cm}+\hspace{-0.1cm}p_{\tilde\theta_i} \Delta_{\hat\theta_i}) \hspace{-0.1cm}+\hspace{-0.1cm}\mu^2_{\hat\theta_i} p_{\hat\theta_i} \Delta_{\hat\theta_i} - \mu^2_{\tilde\theta_i} p_{\tilde\theta_i} \Delta_{\tilde\theta_i} \hspace{-0.1cm}-\hspace{-0.1cm}\mu_{\tilde\theta_i}\mu_{\hat\theta_i} (p_{\tilde\theta_i} \Delta_{\tilde\theta_i} + p_{\hat\theta_i} \Delta_{\hat\theta_i}) \hspace{-0.1cm}-\hspace{-0.1cm}\mu^2_{\hat\theta_i} p_{\hat\theta_i} \Delta_{\hat\theta_i}}{\mu_{\tilde\theta_i} +\mu_{\hat\theta_i}}\\
        &=\frac{\mu_{\tilde\theta}\mu_{\hat\theta}}{\mu_{\tilde\theta}+\mu_{\hat\theta}}(p_{\hat\theta_i} \Delta_{\tilde\theta_i} +p_{\tilde\theta_i} \Delta_{\hat\theta_i} -p_{\tilde\theta_i} \Delta_{\tilde\theta_i} - p_{\hat\theta_i} \Delta_{\hat\theta_i}) \\
        &=\frac{\mu_{\tilde\theta}\mu_{\hat\theta}}{\mu_{\tilde\theta}+\mu_{\hat\theta}}(p_{\tilde\theta_i}-p_{\hat\theta_i})(\Delta_{\hat\theta_i}-\Delta_{\tilde\theta_i}) > 0,
    \end{align*}
    % \bollo{Propongo di spostare qualche conto in appendix e quì dire "per linearità/convessita"}
    where the last inequality holds because every type appears with non-zero probability, $p_{\tilde\theta_i} > p_{\hat\theta_i}$ by assumption, and $\Delta_{\hat\theta_i} > \Delta_{\tilde\theta_i}$ (Equation~\eqref{eq:proof_delta_hat_ge_delta_tilde}).
    As a result, we have built a BIC and IR mechanism $(x',\beta')$ such that $\profit(x',\beta') > \profit(x,\beta)$.
    Hence, $(x,\beta)$ is not optimal, reaching a contradiction.
\end{proof}

Lemma~\ref{lem:total_monotonicity} allows us to impose a total order over the types of every agent $i \in [N]$.
Specifically, for every agent $i$, we define a total order $<$ over $\typeSet_i$ by ordering types according to increasing critical value $v_{\theta_i}$.
Ties between types with the same critical value are broken by placing the type with the higher success probability $p_{\theta_i}$ first.
Furthermore, we denote by $\firstType_i$ and $\lastType_i$ the first and the last type according to this order, respectively.
For every $\theta_i \in \Theta_i$ and every $k \in \mathbb{Z}$ for which it is well defined, we let $\theta_i + k$ denote the type at distance $k$ from $\theta_i$ in the order $<$.
We will say that an allocation rule $x$ is monotone when $x_{i,\theta_i} \ge x_{i,\theta_i} +1$ for every $\theta_i < \lastType_i$. 

\subsection{Computing an optimal mechanism}
By Lemma~\ref{lem:regular_optimal} and Lemma~\ref{lem:total_monotonicity}, restricting to monotone regular mechanisms is optimal.
We now follow the classical procedure to build a Myerson mechanism~\cite{myerson1981optimal}.
Specifically, we show that given any monotone allocation rule $x$, one can build a payment scheme $\beta$ such that $(x,\beta)$ is BIC and IR.
We then express the profit of such a mechanism as a function of the allocation rule $x$ only, and show that it can be decomposed into the sum of suitable virtual utilities. 

We now build, given any monotone allocation rule $x$, a payment scheme $\beta$ such that $(x,\beta)$ is regular, BIC and IR.
Furthermore, the mechanism maximizes the profit among all the BIC mechanisms with such an allocation rule.
Since the BIC constraint (Equation~\eqref{eq:regular_ic}) resembles that of reverse auctions, we recover a similar payment scheme.
Specifically, we design a payment scheme such that, for every agent $i\in [N]$, the BIC constraint between any two consecutive types $\theta_i < \lastType_i$ and $\theta_i+1$ is binding.
This gives us a recursive payment structure that also satisfies the remaining BIC constraints.
We complete the definition of $\beta$ by ensuring individual rationality, which results in the following ex-interim payments:
\begin{equation}
    \label{eq:beta_expected_profit_max}
    \beta_{i,\theta_i} = \begin{cases}
        \beta_{i,\theta_i+1} +(x_{i,\theta_i}-x_{i,\theta_i+1})v_{\theta_i}, & \theta_i < \lastType_i \\
        x_{i,\lastType_i} v_{\lastType_i}, & \theta_i = \lastType_i.
    \end{cases}
\end{equation}
Observe that since $x$ is monotone, we have $\beta_{i,\theta_i}=0$ whenever $x_{i,\theta_i}=0$.
Therefore, the mechanism $(x,\beta)$ is well-defined, as we can recover the ex-post payment following Definition~\ref{def:regular_emchanism}.
We can now prove the following result, whose proof can be found in Appendix~\ref{app:single_task}.
\begin{restatable}{lemma}{BetaForSingleTask}
    \label{lem:beta_given_x}
    Let $(x,\beta)$ be a regular mechanism such that $x$ is monotone and $\beta$ satisfies Equation~\eqref{eq:beta_expected_profit_max}. Then the mechanism is BIC, IR, and maximizes the profit among all the BIC and IR mechanisms with the same allocation rule.
\end{restatable}
We now show that, when $x$ is monotone and $\beta$ satisfies Equation~\eqref{eq:beta_expected_profit_max}, the profit of the mechanism $(x,\beta)$ can be computed by means of appropriate virtual utilities.
Such virtual utilities must account for the multidimensionality of the problem, \emph{i.e.} they make use of both $p_{\theta_i}$ and $v_{\theta_i}$ for every $i \in [N]$ and $\theta_i \in \typeSet_i$.

\begin{restatable}{lemma}{VirtualUtilitiesLemma}
    \label{lem:virtual_utilities}
    % \bollo{Rivedere la notazione: perchè pedice $i$ a $\varphi$?}
    Let $x$ be a monotone allocation rule and $\beta$ a payment scheme satisfying Equation~\eqref{eq:beta_expected_profit_max}.
    Then:
    \begin{equation*}
        \profit(x,\beta) = \sum_{i \in [N]}\sum_{\theta_i \in \typeSet_i} \mu_{\theta_i}  x_{i,\theta_i} \varphi_i(\theta_i) =\sum_{\theta \in \typeSet} \mu_\theta \sum_{i \in [N]} x_{i,\theta} \varphi_i(\theta_i)
    \end{equation*}
    where the virtual utility $\varphi_i(\theta_i)$ is defined as:
    \begin{equation*}
        \varphi_i(\theta_i) \coloneqq p_{\theta_i} (1-v_{\theta_i}) -(v_{\theta_i} -v_{\theta_i-1})\sum_{\hat\theta_i< \theta_i} \frac{\mu_{\hat\theta_i}}{\mu_{\theta_i}} p_{\hat\theta_i}
    \end{equation*}
    for every agent $i \in [N]$ and type $\theta_i \in \typeSet_i$, with the convention that
    the second term is zero for $\theta_i=\firstType_i$. %\mat{AI si lamenta che $\theta_{i-1}$ potrebbe non esistere. gestire il caso}
\end{restatable}
\begin{proof}
    Fix an agent $i \in [N]$. We show that:
    \begin{equation}
        \label{eq:proof_equality_for_virtual_utility}
        \profit_i(x,\beta) \coloneqq \sum_{\theta_i \in \typeSet_i} \mu_{\theta_i} p_{\theta_i}(x_{i,\theta_i}-\beta_{i,\theta_i}) = \sum_{\theta_i \in \typeSet_i} \mu_{\theta_i}  x_{i,\theta_i} \varphi_i(\theta_i).
    \end{equation}
    By unraveling the recursion in Equation~\eqref{eq:beta_expected_profit_max}, it is easy to prove that:
    \begin{equation}
        \label{eq:proof_beta_unraveled}
        \beta_{i,\theta_i} = \sum_{\hat\theta_i \ge \theta_i} x_{i,\hat\theta_i} v_{\hat\theta_i} - \sum_{\hat\theta_i > \theta_i} x_{i,\hat\theta_i} v_{\hat\theta_i-1} = x_{i,\theta_i}v_{\theta_i} +\sum_{\hat\theta_i > \theta_i} x_{i,\hat\theta_i} (v_{\hat\theta_i}-v_{\hat\theta_i-1}).
    \end{equation}
    We now rewrite the profit provided by agent $i$ as follows:
    \begin{align*}
        \profit_i(x,\beta) &=\sum_{\theta_i \in \typeSet_i} \mu_{\theta_i} p_{\theta_i}(x_{i,\theta_i}-\beta_{i,\theta_i}) \\
        &= \mu_{\lastType_i} x_{i,\lastType_i}p_{\lastType_i} (1 -v_{\lastType_i}) + \sum_{\theta_i<\lastType_i} \mu_{\theta_i} p_{\theta_i}(x_{i,\theta_i}-\beta_{i,\theta_i+1} -(x_{i,\theta_i}-x_{i,\theta_i+1})v_{\theta_i}) \\
        &=\mu_{\lastType_i} x_{i,\lastType_i}p_{\lastType_i} (1 -v_{\lastType_i}) +\sum_{\theta_i < \lastType_i} \mu_{\theta_i} x_{i,\theta_i}p_{\theta_i}(1-v_{\theta_i}) +\sum_{\theta_i < \lastType_i}\mu_{\theta_i} p_{\theta_i} (x_{i,\theta_i+1}v_{\theta_i} -\beta_{i,\theta_i+1})  \\
        &=\sum_{\theta_i \in \Theta_i} \mu_{\theta_i} x_{i,\theta_i}p_{\theta_i}(1-v_{\theta_i}) +\sum_{\theta_i < \lastType_i}\mu_{\theta_i} p_{\theta_i} (x_{i,\theta_i+1}v_{\theta_i} -\beta_{i,\theta_i+1}) \\
        &=\sum_{\theta_i \in \Theta_i} \mu_{\theta_i} x_{i,\theta_i}p_{\theta_i}(1-v_{\theta_i}) -\sum_{\theta_i < \lastType_i}\mu_{\theta_i} p_{\theta_i} (\beta_{i,\theta_i+1}-x_{i,\theta_i+1}v_{\theta_i}),
    \end{align*}
    where we applied Equation~\eqref{eq:beta_expected_profit_max} and simple algebra.
    By employing Equation~\eqref{eq:proof_beta_unraveled}, we can rewrite the last summation above as:
    \begin{align*}
        \sum_{\theta_i < \lastType_i}\mu_{\theta_i} p_{\theta_i} (\beta_{i,\theta_i+1}-x_{i,\theta_i+1}v_{\theta_i}) &= \sum_{\theta_i < \lastType_i}\mu_{\theta_i} p_{\theta_i} \left(x_{i,\theta_i+1}v_{\theta_i+1} +\sum_{\hat\theta_i > \theta_i+1} x_{i,\hat\theta_i} (v_{\hat\theta_i}-v_{\hat\theta_i-1}) -x_{i,\theta_i+1}v_{\theta_i} \right) \\
        &= \sum_{\theta_i < \lastType_i}\mu_{\theta_i} p_{\theta_i}\sum_{\hat\theta_i > \theta_i} x_{i,\hat\theta_i} (v_{\hat\theta_i}-v_{\hat\theta_i-1}) \\
        &=\sum_{\theta_i = \firstType_i}^{\lastType_i-1}\sum_{\hat\theta_i = \theta_i+1}^{\lastType_i} \mu_{\theta_i} p_{\theta_i} x_{i,\hat\theta_i} (v_{\hat\theta_i}-v_{\hat\theta_i-1}) \\
        &=\sum_{\theta_i = \firstType_i+1}^{\lastType_i} \sum_{\hat\theta_i = \firstType_i}^{\theta_i-1} \mu_{\hat\theta_i} p_{\hat\theta_i} x_{i,\theta_i} (v_{\theta_i} -v_{\theta_i-1}) \\
        &= \sum_{\theta_i > \firstType_i} x_{i,\theta_i} (v_{\theta_i} -v_{\theta_i-1})\sum_{\hat\theta_i < \theta_i}\mu_{\hat\theta_i} p_{\hat\theta_i}.
    \end{align*}
    where the second to last equality simply changes the order of the double summation.
    As a result:
    \begin{equation*}
        \profit_i(x,\beta) = \sum_{\theta_i \in \Theta_i} \mu_{\theta_i} x_{i,\theta_i}p_{\theta_i}(1-v_{\theta_i}) -\sum_{\theta_i > \firstType_i} x_{i,\theta_i} (v_{\theta_i} -v_{\theta_i-1})\sum_{\hat\theta_i < \theta_i}\mu_{\hat\theta_i} p_{\hat\theta_i} = \sum_{\theta_i \in \typeSet_i} \mu_{\theta_i}  x_{i,\theta_i} \varphi_i(\theta_i),
    \end{equation*}
    where the second equality holds by the definition of $\varphi_i$.
    We can thus conclude that Equation~\eqref{eq:proof_equality_for_virtual_utility} is satisfied.
    Finally, we observe that:
    \begin{align*}
        \profit(x,\beta) &= \sum_{i \in [N]} \profit_i(x,\beta) \\
        &=\sum_{i \in [N]}\sum_{\theta_i \in \typeSet_i} \mu_{\theta_i}  x_{i,\theta_i} \varphi_i(\theta_i) \\
        &=\sum_{i \in [N]}\sum_{\theta_i \in \typeSet_i} \mu_{\theta_i} \sum_{\theta_{-i} \in \typeSet_{-i}} \mu_{\theta_{-i}} x_{i,\theta_i,\theta_{-i}} \varphi_i(\theta_i) \\
        &= \sum_{i \in [N]} \sum_{\theta \in \typeSet} \mu_\theta x_{i,\theta} \varphi_i(\theta_i),
    \end{align*}
    concluding the proof.
\end{proof}

Thanks to Lemma~\ref{lem:regular_optimal}, Lemma~\ref{lem:total_monotonicity}, Lemma~\ref{lem:beta_given_x}, and Lemma~\ref{lem:virtual_utilities}, the allocation rule of an optimal mechanism is the solution to the following optimization problem:
\begin{equation}\label{prog:alloc_virtual_utilites}
\begin{aligned}
    \max_{x \in \mathcal{X}} &\sum_{i \in [N]} \sum_{\theta_i \in \Theta_i} \mu_{\theta_i} x_{i,\theta_i} \varphi_i(\theta_i) \\
    \text{s.t.} \;&x_{i,\theta_i} \ge x_{i,\theta_i+1} &\forall i \in [N], \theta_i < \lastType_i,
\end{aligned}
\end{equation}
where $\mathcal{X}$ is the set of all possible allocation rules.
The value of the solution corresponds to the profit of the optimal mechanism, and its contracts are computed 
according to Equation~\eqref{eq:beta_expected_profit_max}. 

Notably, Program~\ref{prog:alloc_virtual_utilites} exhibits the same structure of the profit-maximizing problem in classical single-item auctions~\cite{myerson1981optimal}.
We can thus solve it by means of the same technique.
Suppose that the virtual utilities are monotone, namely $\varphi_i(\firstType_i) \ge \dots \ge \varphi_i(\lastType_i)$ for every agent $i \in [N]$.
Then, given type report $\theta \in \typeSet$, the optimal allocation rule assigns the task to the agent with the largest virtual utility $\varphi_i(\theta_i)$, or leaves it unassigned if no agent has positive virtual utility.
If the virtual utilities are not monotone, we build suitable ironed virtual utilities $\bar\varphi_i(\cdot)$ that are monotone themselves (see Appendix~\ref{app:ironing}).
The optimal allocation rule is then computed using the ironed virtual utilities in place of the original ones.
Ties must be broken consistently, \emph{e.g.}, in favor of the agent with the smallest index, which gives the following optimal allocation rule:
\begin{equation}
    \label{eq:optimal_allocation_ironing}
    x^\star_{i,\theta} = \begin{cases}
        1 & i= \min\{j \in [N] : \bar\varphi_j(\theta_j) > 0 \;\wedge\; \bar\varphi_j(\theta_j) \ge \bar\varphi_z(\theta_z) \; \forall z \in [N]\}, \\
        0 & \text{otherwise},
    \end{cases}
    \quad \forall i \in [N], \theta \in \typeSet.
\end{equation}
% \mat{The prose uses positive virtual utilities, while (18) permits zero. Either convention can be optimal, but state the chosen convention consistently.}
Therefore, we have the following result.
\begin{restatable}{theorem}{OptimalRegularComputation}
    \label{thm:single_task}
    There exists an optimal, regular, BIC and IR mechanism $(x^\star,\beta^\star)$ such that $x^\star$ satisfies Equation~\eqref{eq:optimal_allocation_ironing} and $\beta^\star$ satisfies Equation~\eqref{eq:beta_expected_profit_max}.
    % Given any type profile $\theta \in \typeSet$, the entries $x^\star(\theta),\beta^\star(\theta)$ can be computed in time polynomial in the instance size.
\end{restatable}
%
% \bollo{Dal punto di vista computazionale, une meccanismo regolare non è male: uno può precalcolare le virtual utility ed i contratti (che diependono solo dal tipo dei singolo agente, perchè regolare), poi dato un type report basta calcolare il vincitore e farelookup su una tabella coi pagamenti. Però formalmente un algoritmo che dato il type report deve calcolarsi tutto deve rifare tuttigli step di preprocess, quindi diventa più lento di uno che rifà tutto usando i pagamenti DSIC. Ma se escludiamo il preprocessing, i pagamenti online di un meccanismo regolare sono più veloci da calcolare.}
%
% \mat{In generale, a parte osservare che è efficiente, qui non mi concentrerei sulla computazione ma sul fatto che abbiamo "close form" o qualcosa del genere. Non so qual è il termine corretto, ma fin qui stiamo facendo econ alla fine. Anche il titolo della sezione è probabilmente sbagliato. Parliamone}
%
Theorem~\ref{thm:single_task} provides a closed-form solution to the problem of computing an optimal mechanism in the single-task setting.
We remark that this solution can be computed efficiently.
Specifically, one can recompute the (ironed) virtual utilities and the contracts (recall that since $(x^\star,\beta^\star)$ is regular, for every agent $i \in [N]$ and type $\theta_i \in \typeSet_i$ one has to compute a single contract $\beta_{i,\theta_i}/x_{i,\theta_i}$) in polynomial time.
Then, given a type report $\theta \in \typeSet$, we can directly compute the agent that receives the task by Equation~\eqref{eq:optimal_allocation_ironing}. 
However, we observe that $(x^\star,\beta^\star)$ need not be DSIC, and in general the optimal DSIC mechanism is not regular.

We complement Theorem~\ref{thm:single_task} by characterizing the optimal DSIC mechanism.
Specifically, we show that by changing the payments, we can build a DSIC mechanism $(x^\star,\beta^\circ)$ that attains the same profit as the optimal BIC mechanism.
Since DSIC is stronger than BIC, this proves that $(x^\star,\beta^\circ)$ is optimal for both notions of incentive compatibility, and that relaxing the DSIC constraints to BIC does not improve the profit.

In the regular BIC mechanism $(x^\star,\beta^\star)$ the payments depend only on the type of the agent receiving the task.
By contrast, the payments $\alpha^\circ$ of the optimal DSIC mechanism depend on the types of all agents \emph{except} the one executing the task, and are defined as:
\begin{equation}
    \label{eq:dsic_payments}
    \alpha^\circ_{i,\theta_i, \theta_{-i}} \coloneqq \bar \alpha^\circ_{i,\theta_{-i}} \coloneqq \max_{\substack{\theta'_i \in \typeSet_i, \\ x^\star_{i,\theta'_i,\theta_{-i}}=1}} v_{\theta'_i} \quad \forall i \in [N], (\theta_i,\theta_{-i}) \in \typeSet,
\end{equation}
with the convention that the maximum over an empty set is zero (\emph{i.e.}, $\alpha^\circ_{i,\theta_{-i}}=0$ when agent $i$ does not receive the task regardless of the agent's type).
Notice that therefore the expected contracts are:
\begin{equation*}
    \beta^\circ_{i,\theta_i, \theta_{-i}} = \begin{cases}
        % \displaystyle\max_{\substack{\theta'_i \in \typeSet_i, \\ x^\star_{i,\theta'_i,\theta_{-i}}=1}} v_{\theta'_i}, & x^\star_{i,\theta_i,\theta_{-i}}=1, \\
        \bar \alpha^\circ_{i,\theta_{-i}}, & x^\star_{i,\theta_i,\theta_{-i}}=1, \\ 
        0, & x^\star_{i,\theta_i,\theta_{-i}}=0.
    \end{cases}
\end{equation*}

Formally, we show the following.
\begin{restatable}{proposition}{OptiamalDSIC}
    Let $(x^\star,\beta^\star)$ be the optimal BIC mechanism computed according to Equation~\eqref{eq:optimal_allocation_ironing} and Equation~\eqref{eq:beta_expected_profit_max}, and let $(x^\star,\beta^\circ)$ be the mechanism whose payments are given by Equation~\ref{eq:dsic_payments}.
    Then $(x^\star,\beta^\circ)$ is IR and DSIC and $\profit(x^\star,\beta^\star) = \profit(x^\star,\beta^\circ)$.
\end{restatable}
\begin{proof}
    The mechanism $(x^\star,\beta^\circ)$ is IR, as by construction whenever an agent receives the task, the corresponding payment is at least than the critical value.
    Now we show that $(x^\star,\beta^\circ)$ is DSIC.
    Fix an agent $i \in [N]$, its true type $\theta_i \in \typeSet_i$, another type $\hat\theta_i \neq \theta_i$, and the type profile of the other agents $\theta_{-i} \in \typeSet_{-i}$.
    Since the contract $\bar \alpha^\circ_{i,\theta_{-i}}$ received by the agent does not depend on her reported type, the agent $i$ may be incentivized to misreport $\hat\theta_i$ only when doing so allows her to receive the task, while reporting $\theta_i$ does not. 
    However, in this case $v_{\theta_i} \ge \bar \alpha^\circ_{i,\theta_{-i}}$ by Equation~\eqref{eq:dsic_payments}.
    Therefore, by misreporting and executing the task, the agent's utility is negative or zero, while opting out it is zero.
    We conclude that the agent is not incentivized to misreport.

    Finally, we show that $\profit(x^\star,\beta^\star) = \profit(x^\star,\beta^\circ)$, by proving that for every agent $i \in [N]$ and type $\theta_i \in \typeSet_i$, we have $\beta^\star_{i,\theta_i} = \beta^\circ_{i,\theta_i}$.
    Recall that $\beta^\star_{i,\theta_i}$ is defined by Equation~\eqref{eq:beta_expected_profit_max}, and let $\tilde\theta_i$ be the last type such that $x^\star_{i,\theta_i}>0$.
    Notice that $\beta^\star_{i,\theta_i} = \beta^\circ_{i,\theta_i} = 0$ for every $\theta_i > \tilde\theta_i$, while we need to prove that $\beta^\star_{i,\theta_i} = \beta^\circ_{i,\theta_i}$ for $\theta_i \le \tilde\theta_i$.
    We prove this by induction, starting from $\tilde\theta_i$.
    Since agent $i$ does not get the task for every $\theta_i > \tilde\theta_i$, it must be the case that whenever $x^\star_{i,\tilde\theta_i,\theta_{-i}}=1$, the linear contract $\alpha^\circ_{i,\tilde\theta_i,\theta_{-i}}$ is $v_{\tilde\theta_i}$, and $\beta^\circ_{i,\tilde\theta_i,\theta_{-i}}=x^\star_{i,\tilde\theta_i,\theta_{-i}}v_{\tilde\theta_i}$.
    Moreover, when $x^\star_{i,\tilde\theta_i,\theta_{-i}}=0$, we have $0=\beta^\circ_{i,\tilde\theta_i,\theta_{-i}}=x^\star_{i,\tilde\theta_i,\theta_{-i}}v_{\tilde\theta_i}$.
    Overall, we conclude that:
    \begin{equation*}
        \beta^\circ_{i,\tilde\theta_i} = \sum_{\theta_{-i} \in \typeSet_{-i}} \mu_{\theta_{-i}} x^\star_{i,\tilde\theta_i,\theta_{-i}}v_{\tilde\theta_i} = x_{i,\tilde\theta_i} v_{\tilde\theta_i} = \beta^\star_{i,\tilde\theta_i},
    \end{equation*}
    where the last equality holds by Equation~\eqref{eq:beta_expected_profit_max} since $x_{i,\theta_i} = 0$ for every $\theta_i > \tilde\theta_i$, proving the base step.
    We now reason by induction for $\theta_i < \tilde\theta_i$, and prove the following chain of equalities:
    \begin{align*}
        \beta^\star_{i,\theta_i} = \beta^\star_{i,\theta_i+1} +(x^\star_{i,\theta_i} -x^\star_{i,\theta_i+1}) v_{\theta_i} = \beta^\circ_{i,\theta_i+1} +(x^\star_{i,\theta_i} -x^\star_{i,\theta_i+1}) v_{\theta_i} = \beta^\circ_{i,\theta_i}.
    \end{align*}
    The first equality holds by Equation~\eqref{eq:beta_expected_profit_max}, the second by the inductive hypothesis, and for the last one we argue that:
    \begin{equation}        
        \label{eq:proof_beta_dsic_inductive}
        \beta^\circ_{i,\theta_i+1,\theta_{-i}} +(x^\star_{i,\theta_i,\theta_{-i}} -x^\star_{i,\theta_i+1,\theta_{-i}}) v_{\theta_i} = \beta^\circ_{i,\theta_i,\theta_{-i}}.
    \end{equation}
    for every $\theta_{-i} \in \typeSet_{-i}$.
    Indeed, when $x^\star_{i,\theta_i,\theta_{-i}}=0$, we also have $x^\star_{i,\theta_i+1,\theta_{-i}}=0$ and Equation~\eqref{eq:proof_beta_dsic_inductive} is satisfied.
    When $x^\star_{i,\theta_i,\theta_{-i}}=1$ and $x^\star_{i,\theta_i+1,\theta_{-i}}=0$, we have $\alpha^\circ_{i,\theta_i,\theta_{-i}} = v_{\theta_i}$ and $\beta^\circ_{i,\theta_i+1,\theta_{-i}}=x^\star_{i,\theta_i+1,\theta_{-i}}=0$, thus Equation~\eqref{eq:proof_beta_dsic_inductive} is again satisfied.
    Finally, when $x^\star_{i,\theta_i,\theta_{-i}}=x^\star_{i,\theta_i+1,\theta_{-i}}=1$, Equation~\eqref{eq:proof_beta_dsic_inductive} imposes $\beta^\circ_{i,\theta_i+1,\theta_{-i}} = \beta^\circ_{i,\theta_i,\theta_{-i}}$, which clearly holds because, under the mechanism $(x^\star,\beta^\circ)$, agent $i$ receives the task with the same contract given both type reports.
    Therefore, Equation~\eqref{eq:proof_beta_dsic_inductive} always holds.
    We conclude that $\beta^\star_{i,\theta_i} = \beta^\circ_{i,\theta_i}$, and by induction this holds for every $\theta_i$.
    Therefore, by Equation~\eqref{eq:profit_in_beta_and_exinterim} the mechanisms $(x^\star,\beta^\star)$ and $(x^\star,\beta^\circ)$ collect the same profit.
\end{proof}

    \section{Price of Incentive Compatibility}
In this section we study the price of incentive compatibility in the single-task problem, \emph{i.e.}, the profit loss due to the need to elicit the true types from the strategic agents, formally:
% \bollo{TODO: check if this is really standard in the literature, I just asked ai. And decide if we prefer the incentive-compatibility ratio which is the inverse}
% Formally, this price is defined as:
\begin{equation*}
    \PIC \coloneqq \frac{\OPTC}{\OPT},
\end{equation*}
where $\OPT$ is the profit of the optimal BIC and IR mechanism, and $\OPTC$ that of an optimal ‘‘clairvoyant'' mechanism that sees the true types of the agents, rather than receiving the (strategically disclosed) report from them.
Such a clairvoyant mechanism simply picks the best agent $i$ and gives them the task with payment $v_{\theta_i}$, disregarding the incentive compatibility constraints.
% In other words, the clairvoyant mechanism is one that can see the true types of the agent, rather than receiving the (strategically disclosed) report from them.
% where $\OPT$ is the profit of the optimal BIC and IR mechanism, and $\OPTC$ that of an optimal ‘‘clairvoyant'' mechanism that does not need to satisfy the IC constraints.
% In other words, the clairvoyant mechanism is one that can see the true types of the agent, rather than receiving the (strategically disclosed) report from them.
% Thus, whenever the optimal clairvoyant (henceforth called simply clairvoyant) gives the task to agent $i$ of type $\theta_i$, it employs the contract $\alpha_{i,\theta_i} = v_{\theta_i}$, that is the minimum linear contract required to satisfy IR. \mat{In generale, potrebbe essere che in certi punti sovraspieghi e formalizzi le cose. Qui ad esempio dire che prende sempre l agente con il tipo migliore chiarisce più che la parte formale forse. Anche relazione di prende il migliore con IC non vale è formale ma non intuitiva}
% Given type report $\theta \in \typeSet$, the mechanism gives the task to the agent with the largest $p_{\theta_i}(1-v_{\theta_i})$, or to no one if every agent has a critical value larger than one.
Its profit can thus be computed as follows:
\begin{equation*}
    \OPTC \coloneqq \sum_{\theta \in \typeSet} \mu_\theta \max_{i \in [N]} [p_{\theta_i}(1-v_{\theta_i})]_+ = \sum_{\theta \in \typeSet} \mu_\theta \max_{i \in [N]} [p_{\theta_i} -c_{\theta_i}]_+.
\end{equation*}

Our main result in this section is that $\PIC = \Theta\left(NK\right)$ in the worst case, namely, our mechanism loses a factor $1/NK$ of utility with respect to the clairvoyant.
We prove this in two steps.
We first observe that it is always possible to attain at least a fraction $\frac{1}{NK}$ of the clairvoyant profit, \emph{i.e.}, $\text{PIC} \le NK$.
Intuitively, this can be done by selecting an agent $ i \in [N]$ and a type $\bar \theta_i \in \typeSet_i$, and always assigning the task to agent $i$ with a fixed contract $\bar\alpha = v_{\bar\theta_i}$. 
The agent will actually exert effort only when its type $\theta_i \in \typeSet$ satisfies $v_{\theta_i} \le v_{\bar\theta_i}$.
This mechanism attains profit of at least $\mu_{\bar\theta_i}p_{\bar\theta_i}(1-v_{\bar\theta_i})$.
As the clairvoyant can extract at most $\mu_{\theta_i}p_{\theta_i}(1-v_{\theta_i})$ from each agent $i \in [N]$ and type $\theta_i \in \typeSet_i$, a suitable choice of $i$ and $\bar\theta_i$ proves the upper bound on the price of incentive compatibility.
To provide the formal argument, we slightly change the mechanism to satisfy the IR constraints, and we do not assign the task when agent $i$ reports a type $\theta_i$ with $v_{\theta_i} > \bar \alpha$.
\begin{restatable}{proposition}{UtilityBoundOneLemma}
    \label{lem:utility_price_upper}
    In every instance where $\OPTC > 0$, it holds that:
    \begin{equation*}
        % \OPT \ge \frac{\OPTC}{NK}.
        \PIC = \frac{\OPTC}{\OPT} \le NK.
    \end{equation*}
\end{restatable}
\begin{proof}
    Let:
    \begin{equation*}
        (i, \bar\theta_i) \in \argmax_{j \in [N], \theta_j \in \typeSet_j} \mu_{\theta_j}p_{\theta_j}(1-v_{\theta_j}).
    \end{equation*}
    Let $(x,\beta)$ be the regular mechanism defined as:
    \begin{align*}
        x_{j,\theta} \coloneqq \begin{cases}
            1 & j=i, v_{\theta_i} \le v_{\bar\theta_i} \\
            0 & \text{otherwise},
        \end{cases} 
        \quad \beta_{j,\theta_j} = \begin{cases}
            v_{\bar\theta_i} & j=i, v_{\theta_j} \le v_{\bar\theta_i} \\
            0 &\text{otherwise}.
        \end{cases}
    \end{align*}
    Notice that agent $i$ receives the task with the stated fixed contract.
    This mechanism is both IR and BIC.
    These constraints are trivially satisfied for agents $j \neq i$.
    It is individually rational, as agent $i$ receives the task with contract $\bar\alpha = v_{\bar\theta_i}$ when reporting a type $\theta_i$ with $v_{\theta_i} \le v_{\bar\theta_i}$.
    It is also BIC, as the agent is not incentivized to misreport to attain the task.
    Specifically, if the agent does not receive the task by reporting its true type $\theta_i \in \typeSet_i$, then $\bar\alpha < v_{\theta_i}$.
    Therefore, the agent would receive negative utility by misreporting and performing the task, or zero by misreporting and opting out.
    
    Finally, observe that the profit of $(x,\beta)$ is at least $ \mu_{\bar\theta_i}p_{\bar\theta_i}(1-v_{\bar\theta_i}) > 0$ (otherwise $\OPT=\OPT_C=0$).
    % Let $\Theta'_j \subseteq \Theta_j$ be the set of types of agent $j \in [N]$ that have critical value at most one. 
    Consequently:
    \begin{align*}
        \OPTC &= \sum_{\theta \in \typeSet} \mu_\theta \max_{i \in [N]} [p_{\theta_i}(1-v_{\theta_i})]_+ \\
        &\le \sum_{i \in [N]} \sum_{\theta \in \typeSet}  \mu_\theta [p_{\theta_i}(1-v_{\theta_i})]_+ \\
        &= \sum_{i \in [N]} \sum_{\theta \in \typeSet_i}  \mu_{\theta_i} [p_{\theta_i}(1-v_{\theta_i})]_+ \\
        &\le \sum_{i \in [N]} \sum_{\theta \in \typeSet_i}  \mu_{\bar\theta_i}p_{\bar\theta_i}(1-v_{\bar\theta_i}) \\
        &= KN \mu_{\bar\theta_i}p_{\bar\theta_i}(1-v_{\bar\theta_i}) \le KN \profit(x,\beta) \le KN\OPT,
    \end{align*}
    where the second-to-last inequality holds by definition.
\end{proof}

We complement this result by showing that there is always an instance where $\OPT = \Theta\left(\frac{\OPTC}{NK}\right)$, and, therefore, $\PIC = \Theta(NK)$ in the worst case.

\begin{restatable}{proposition}{PICFinalTh}
    There exists an instance such that:
    \begin{equation*}
        \PIC = \frac{\OPTC}{\OPT} \ge  \frac{NK}{8}.
    \end{equation*}
\end{restatable}
\begin{proof}
    We now build an instance where $\PIC \ge NK/8$.
    Specifically, let the $N$ agents share the same set of types $\Theta_1 = \Theta_2 = \dots =\Theta_N = [K]$.
    Let us also define the following quantities:
    \begin{equation*}
        q \coloneqq 1-\frac{1}{2N}, \quad a \coloneqq \frac{1}{2N(2^{K-1}-1)}.
    \end{equation*}
    Then, we build the instance according to the following table:
    \begin{center}
    \begin{tabular}{l|cccc}
    Type            & $\mu_k$       & $p_k$             & $c_k$         & $v_k$\\ \hline
      $1$           & $q$           & $\nicefrac{a}{q}$ & $0$           & $0$\\
    $k=2,\dots,K$   & $2^{k-2}a$    & $1$               & $1-2^{-(k-1)}$   & $1-2^{-(k-1)}$
    \end{tabular}
    \end{center}
    Notice that the prior probabilities sum to one, and both success probabilities and costs belong to the interval $[0,1]$.
    Thus, the instance is valid.
    Furthermore, the types are correctly ordered in increasing order of critical value.

    \paragraph{Optimal Mechanism } Let $(x,\beta)$ be the optimal mechanism, which w.l.o.g. we assume to be regular (Lemma~\ref{lem:regular_optimal}).
    From the BIC constraint (Equation~\eqref{eq:regular_ic}) for the deviation from type $1 \le k \le K-1$ to $k+1$, for every agent $i \in [N]$ the following holds:
    \begin{align*}
        \beta_{i,k} -v_{i,k}x_{i,k} \ge \beta_{i,k+1} -v_{i,k}x_{i,k+1} = \beta_{i,k+1} -v_{i,k+1}x_{i,k+1} +(v_{i,k+1}-v_{i,k})x_{i,k+1}. 
    \end{align*}
    Recursively applying the inequality above, we have:
    \begin{equation}
        \label{eq:proof_pic_recursive}
        \beta_{i,k} -v_{i,k}x_{i,k} \ge \sum_{h=k+1}^K (v_h -v_{h-1})x_{i,h}.
    \end{equation}
    Let us also observe that $\mu_1 p_1 = a$ and $\mu_k p_k = 2^{k-2} a$ for $k \ge 2$.
    Exploiting Equation~\eqref{eq:proof_pic_recursive}, we upper bound the contribution of agent $i \in [N]$ to the principal's profit as follows:
    \begin{align*}
        \sum_{k=1}^K \mu_k p_k (x_{i,k} -\beta_{i,k}) &= \sum_{k=1}^K \mu_k p_k ((1-v_k)x_{i,k} -(\beta_{i,k}-v_k x_{i,k})) \\
        % convention \sum_{h=K+1}^K ... = 0
        &\le \sum_{k=1}^K \mu_k p_k \left( (1-v_k)x_{i,k} -\sum_{h=k+1}^K (v_h -v_{h-1})x_{i,h} \right) \\ 
        &= \sum_{k=1}^K \mu_k p_k (1-v_k)x_{i,k} -\sum_{k=1}^K \mu_k p_k  \sum_{h=k+1}^K (v_h -v_{h-1})x_{i,h} \\
        &= \sum_{k=1}^K \mu_k p_k (1-v_k)x_{i,k} -\sum_{k=2}^K x_{i,k} (v_{i,k}-v_{i,k-1})\sum_{h=1}^{k-1} \mu_h p_h \\
        &= a x_{i,1} +\sum_{k=2}^K x_{i,k} \left(\mu_k p_k (1-v_k) -(v_{i,k}-v_{i,k-1})\sum_{h=1}^{k-1} \mu_h p_h \right) \\
        &= a x_{i,1} +\sum_{k=2}^K x_{i,k} (2^{k-2}a (1-v_k) - (v_{i,k}-v_{i,k-1})2^{k-2}a) \\
        &= a x_{i,1} +\sum_{k=2}^K x_{i,k} (2^{k-2}a 2^{-(k-1)} - (2^{-(k-2)}-2^{-(k-1)})2^{k-2}a) \\
        &= a x_{i,1} +\sum_{k=2}^K x_{i,k} \left(\frac{a}{2} - a +\frac{a}{2} \right) \\
        &= a x_{i,1} = \frac{a}{q} \mu_1 x_{i,1},
    \end{align*}
    where the inequality holds by Equation~\eqref{eq:proof_pic_recursive}, and the equalities by exploiting the facts that $\mu_1 p_1 = a$, $\mu_k p_k = 2^{k-2} a$ for $k\ge2$, and $\sum_{h=1}^{k-1} \mu_h p_h = 2^{k-2}a$.
    As a result:
    \begin{equation}
        \label{eq:proof_pic_opt}
        \OPT = \profit(x,\beta) \le \frac{a}{q} \sum_{i \in [N]} \mu_1 x_{i,1} 
        \le \frac{a}{q}\sum_{i \in [N]} \sum_{\theta \in \typeSet: \theta_i=1} \mu_\theta  x_{i,\theta}
        \le \frac{a}{q}\sum_{i \in [N]} \sum_{\theta \in \typeSet} \mu_\theta  x_{i,\theta} \le \frac{a}{q}.
    \end{equation}

    \paragraph{Clairvoyant profit and PIC} For every $i \in [N]$, the probability that agent $i$ has type $2 \le k \le K$ and all the other agents have type $1$ is $q^{N-1}2^{k-2}a$.
    Therefore:
    \begin{equation}
        \label{eq:proof_pic_optc}
        % thefirst passage should be to consider the su only on the type profiles described above, but who wants to introduce the notation for that?
        \OPTC \ge \sum_{i=1}^N \sum_{k=2}^K q^{N-1}2^{k-2}a (p_k - c_k) 
        =\sum_{i=1}^N \sum_{k=2}^K q^{N-1}2^{k-2}a 2^{-(k-1)} = N(K-1) q^{N-1} \frac{a}{2}
    \end{equation}
    Combining Equation~\eqref{eq:proof_pic_opt} and Equation~\eqref{eq:proof_pic_optc}, we have:
    \begin{equation*}
        \PIC = \frac{\OPTC}{\OPT} \ge \frac{N(K-1) q^N}{2} \ge \frac{N(K-1)}{4} = \frac{NK}{8},
    \end{equation*}
    where we exploited the facts that $q^N = (1-1/(2N))^N \ge 1/2$ thanks to Bernoulli's inequality, and $K-1 \ge K/2$ since $K \ge 2$.
\end{proof}

    \section{Multiple Tasks}\label{sec:multi_task}
% In this section we analyze the problem of computing an optimal mechanism for the multi-task setting.
% Unlike the single-task scenario, in this case we cannot recover any monotonicity property.
% Indeed, monotonicity property cannot be recover even in simple auctions with multiple items, or reverse auctions with multiple workers (see. e.g., ...).
% Therefore, we cannot use a Myerson mechanism to maximize the profit.
% Furthermore, given a type profile $\theta \in \typeSet$, an optimal mechanism could employ a distribution $y(\theta) \in \Delta_{\tSet}$ with support over the entire set of allocation $\tSet$.
% As there are $\mathcal{O}(M^N)$ possible allocations, we are unable to represent this distribution in space polynomial in $N$.
% As a result, we will focus on the problem of computing a profit-maximizing mechanism when the number of agent $N$ is constant.
% This allows us to represent any mechanism $(y,\sigma)$ explicitly, \emph{i.e.}, by computing the values $y(t|\theta) \in [0,1]$ and $\sigma(t|\theta) \in \mathbb{R}_+^N$ for every possible type profile $\theta \in \typeSet$ and allocation $t \in \typeSet$.
% Such an explicit representation is polynomial in $M$ and $K$ when the number of agents $N$ is constant.

In this section, we investigate the computation of an optimal mechanism in the multi-task setting. 
Unlike the single-task case, agents' multidimensional private information and the interaction between task allocations and outcome-contingent payments prevent us from directly extending our monotonicity-based characterization. 
Although multidimensional implementability admits more general characterizations (see, \emph{e.g.},~\cite{rochet1987necessary,cai2012optimal}), these results do not directly address our multi-task environment, where payments are provided only upon success.
We therefore adopt a direct optimization approach. 
Since an optimal mechanism may randomize over the entire set of allocations $\tSet$, whose size is exponential in the number $N$ of agents, we focus on settings in which this number is constant, allowing for an explicit representation of the mechanism in polynomial space.

% \bollo{Nota: nel paragrafo sopra non abbiamo ancora dimostrato che l'ottimo non esiste. Quindi mi sembra più intuitivo parlare del problema di calcoalre l'ottimo, e solo dopo discutere del fatto che al posto dell'ottimo dobbiamo fare un $\epsilon$-ottimo.}

We first observe that, unlike the single-task setting, in a multi-task problem the optimum may not exist.
The key observation is that in some instances, the principal may benefit from employing large payments (larger than one) with small probability $\epsilon > 0$.
In these instances the profit increases as $\epsilon$ decreases (and the payment increases).
As a result, the profit-maximizing problem admits a supremum, which we will denote by $\SUP$, but not an optimum.

We formally prove this result in the case with $N=M=2$. 
Specifically, we build an instance where both agents share the same set of types $\typeSet_1 = \typeSet_2 = \{\bar \theta, \tilde \theta\}$.
An agent of type $\bar\theta$ can complete both tasks 1 and 2, but only with a small success probability. 
Therefore, the principal can extract only negligible utility from type $\bar\theta$.
By contrast, an agent of type $\tilde\theta$ can complete task 1 with probability one, but cannot complete task 2.
Ideally, we would like to give the task to an agent of type $\tilde\theta$ with payment $c_{\tilde\theta,1}$, but in order to do so we must ensure that no agent misreports $\tilde\theta$ in place of $\bar\theta$.
Task 2 can be used to distinguish between these two types, by assigning it with probability $\epsilon$ and a large contract $L=L(\epsilon) > 1$ to any agent that reports type $\bar\theta$.
While assigning task 2 provides negative utility to the principal due to the large contract $L>1$, this loss decreases as $\epsilon$ tends to zero.
Since $\epsilon$ must be strictly positive, we conclude that the optimum cannot be achieved.
\begin{lemma}
    \label{lem:multi_task_no_opt}
    There is an instance with $N=2$ agents, $M=2$ tasks and at most $K=2$ types per agent which does not admit an optimum.
\end{lemma}
\begin{proof}
    We build the instance as follows.
    The two agents share the same set of types $\typeSet_1 = \typeSet_2 = \{\bar \theta, \tilde \theta\}$.
    For the sake of notation, we will denote by $\bar p_m$ and $\bar c_m$ ($\tilde p_m$ and $\tilde c_m$) the success probability and cost of type $\bar\theta$ ($\tilde \theta$), which are defined as follows (we also report the critical values for completeness):
    \begin{align*}
        \bar p_1 = \frac{1}{10},\quad     &\bar c_1 = \frac{1}{100},\quad     &\bar v_1 = \frac{1}{10}\\
        \bar p_2 = \frac{1}{10},\quad     &\bar c_2 = \frac{1}{5},\quad     &\bar v_2 = 2\\
        \tilde p_1 = 1,\quad   &\tilde c_1 = \frac{1}{4},\quad     &\tilde v_1 = \frac{1}{4}\\
        \tilde p_2 = 0,\quad   &\tilde c_2 = 1,\quad   &\tilde v_2 = \infty\\
    \end{align*}
    We further define the corresponding prior probabilities $\mu_{\bar\theta} \coloneqq \bar w \coloneqq \frac{1}{100}$ and $\mu_{\tilde\theta} \coloneqq \tilde w \coloneqq \frac{99}{100}$.

    We will prove that there exists no optimal mechanism in this instance, \emph{i.e.} the problem of maximizing the profit over the BIC and IR mechanisms admits a supremum, but not a maximum.
    We outline the proof as follows.
    First, we compute the maximum profit that can be achieved when neither agent receives task $2$.
    Then, we show that if the agents also receive task $2$ with some (small) probability, we can increase the profit.
    Finally, we show that any mechanism that assigns task $2$ to one of the first two agents can be improved, and therefore it is not optimal.
    As a supremum must exist, since the profit is upper bounded by one, we conclude that no mechanism is optimal.

    \paragraph{Profit without task $2$} Consider the class of mechanisms that never assign task $2$.
    The problem of assigning task $1$ degenerates to a single-task problem.
    We observe that the agents share the same types, and simple computations show that the ironed virtual utilities are positive for both types.
    Therefore, the optimal mechanism in this class always assigns the first task to agent $1$ (or equivalently, always to agent $2$) with contract $\tilde v_1$.
    Hence, the optimal mechanism in this class achieves the following profit, which we denote by $R$:
    \begin{equation*}
        R \coloneqq \bar w \bar p_1 (1-\tilde v_1) + \tilde w (\tilde p_1 - \tilde c_1) =\frac{2973}{4000} < 0.744.
    \end{equation*}

    \paragraph{Profit with task $2$} We now show that by assigning task $2$, it is possible to achieve larger utility than the one attained by the previous mechanism.
    Consider the mechanism whose allocation rule, parameterized by $\epsilon \in (0,1/2)$, is defined by the following table, where $i \rightarrow m$ indicates that agent $i$ is allocated to task $m$.
    \begin{equation*}
    \begin{array}{|c|c|c|}
        \hline
                                & \theta_2 = \bar\theta             & \theta_2 = \tilde\theta \\ \hline
        \theta_1 = \bar\theta   & \makecell{1\rightarrow 2 \text{ with probability $\epsilon$} \\ 2\rightarrow 2 \text{ with probability $\epsilon$}}                                                & \makecell{1 \rightarrow 2 \text{ with probability $\epsilon$} \\ 2 \rightarrow 1}\\ \hline
        \theta_1 = \tilde\theta & 1 \rightarrow 1 & 1 \rightarrow 1 \\ \hline
    \end{array}
    \end{equation*}
    Agents $1$ and $2$ receive task $1$ with payment $\tilde v_1$, and task $2$ with payment $L$, which we now define to ensure that the mechanism is BIC.
    Notice that if agent $1$ or $2$ has type $\tilde\theta$, they are not incentivized to deviate, as doing so would prevent them from receiving the single task they can complete ($\tilde\theta$ cannot complete task $2$).
    Therefore, we can focus on the sole IC constraints that prevent agents $1$ and $2$ from misreporting $\tilde\theta$ in place of the true type $\bar\theta$.
    Consider agent $1$, who is allocated task $2$ with probability $\epsilon$ and contract $L$ by truthfully reporting $\bar\theta$, and gets task $1$ with contract $\tilde v_1$ by misreporting $\tilde\theta$.
    The corresponding IC constraint imposes:
    \begin{equation*}
        \epsilon (\bar p_2 L - \bar c_2) \ge \bar p_1 \tilde v_1 - \bar c_1 \coloneqq m.
    \end{equation*}
    Furthermore, observe that agent $2$ has the same IC constraint preventing the misreport, which is satisfied for the following choice of $L$:
    % if agent $1$ reports $\tilde\theta$, agent 2 gets always zero utility. Otherwise it gets \epsilon (\bar p_2 L - \bar c_2) by truthful report and \bar p_1 \tidle v_1 - \bar c_1 by misreporting. So BIC gives the same constraint (formall left and right side get multipled by $\bar w$).
    \begin{equation*}
        L = \frac{\bar p_1 \tilde v_1 - \bar c_1}{\epsilon \bar p_2} +\bar v_2.
    \end{equation*}
    Notice that this choice of $L$ also satisfies individual rationality, and thus the mechanism is both IR and BIC.
    We now analyze the profit of this mechanism.
    Whenever agent $1$ or $2$ receives task $2$, the principal suffers a loss of utility (since $L > 1$) equal to $\bar p_2 (L - 1)$.
    Such a task is given to the first two agents with probability $2 \bar w^2 \epsilon + \bar w \tilde w \epsilon$.
    Overall, the expected loss $\ell$ due to task $2$ is:
    \begin{align*}
        \ell &\coloneqq (2 \bar w^2 \epsilon + \bar w \tilde w \epsilon)\bar p_2 (L - 1) \\
        &= (2 \bar w^2 + \bar w \tilde w)\epsilon \bar p_2 (\frac{m}{\epsilon \bar p_2} +\bar v_2 - 1) \\
        &= (2 \bar w^2 +\bar w \tilde w)m +(2 \bar w^2 + \bar w \tilde w)\epsilon\bar p_2(\bar v_2 - 1) \\
        &\le (2 \bar w^2 +\bar w \tilde w)m + 2\epsilon.
    \end{align*}
    Furthermore, task $1$ provides a profit of $\tilde p_1 - \tilde c_1$ whenever it is assigned to an agent (since the contract is $\tilde v_1$ and the agent has type $\tilde\theta$), which happens with probability $\tilde w +\bar w\tilde w$.
    As a result, the profit of this mechanism is at least:
    \begin{align*}
        R(\epsilon) &\ge (\tilde w +\bar w\tilde w)(\tilde p_1 - \tilde c_1) - \ell \\
        &\ge (\tilde w +\bar w\tilde w)(\tilde p_1 - \tilde c_1) -(2 \bar w^2 +\bar w \tilde w)m  -2\epsilon \\
        &= (\tilde w +\bar w\tilde w)(\tilde p_1 - \tilde c_1) -(2 \bar w^2 +\bar w \tilde w)(\bar p_1 \tilde v_1 - \bar c_1) - 2\epsilon \\
        &= \frac{1499547}{2000000} -2\epsilon \ge 0.749 -2\epsilon.
    \end{align*}
    For sufficiently small $\epsilon$ (\emph{e.g.}, $\epsilon=10^{-4}$), this mechanism collects a profit larger than $R$, which is the maximum profit achievable when agents $1$ and $2$ never receive task $2$.

    % \bollo{Decisamente da riscrivere decentemente} \\
    % \bollo{L'orginiale non era del tutto corretta: dato $\theta$, il meccanismo può dare allocazioni $(2,1)$ e $(2,\varnothing)$ con diverse probabilità e diversi contratti}
    % \bollo{Credo he formalmente si debbano fare più sottocasi. Se $y((2,\varnothing)|\theta) = \epsilon$, è facile e basta ridurre quello cambiado il suo contratto.
    % Se $y((2,1)|\theta) = \epsilon$ e $y((2,\varnothing)|\theta)$, allora cambiarlo cambia anche qualcosa per l'agente 2, che però volevo lasciare uguale. Dovrei aumentare $y((\varnothing,1)|\theta)$, ma questo può usare un contratto diverso. 
    % In pratica mi rimane il caso $y((2,1)|\theta)=\epsilon>0$, con $y((2,\varnothing)|\theta')=y((\varnothing,2)|\theta')=0$ per $\theta'$ (anche $\theta'=\theta$)}
    \paragraph{No mechanism is optimal}
    We now show that every other mechanism is also suboptimal, and therefore the supremum cannot be achieved.
    Consider any IR and BIC mechanism $(y,\sigma)$ that, for some type report $\theta \in \typeSet$, assigns task $2$ to agent $1$ (the argument is the same with agent $2$, as they are symmetrical) with positive probability.
    There are two possible allocations that assign task $2$ to agent $1$, namely $(2,\varnothing)$, and $(2,1)$.
    We will now consider the case where $y((2,\varnothing)|\theta) > 0$.
    Formally, assume $\epsilon \coloneqq y((2,\varnothing)|\theta) > 0$ and let $L \coloneqq \sigma_1((2,\varnothing)|\theta)$.
    Observe that $\theta_1 = \bar\theta$, as task $2$ cannot be assigned to an agent of type $\tilde\theta$ without violating the IR constraints.
    Consider a new mechanism $(y^\circ,\sigma^\circ)$ that, given $\theta$, assigns task $2$ to agent $1$ with probability only $\epsilon/2$ and suitably increases the proposed contracts.
    Formally, we let:
    % it is a bit strange because I must change the probability ofgining the taks to agent 1, but nothing must change for agent 1.
    \begin{align*}
        y^\circ((2,\varnothing)|\theta)= \epsilon/2, \quad \sigma^\circ_1((2,\varnothing)|\theta)= L'= 2L - \bar v_2, \quad
        y^\circ((\varnothing,\varnothing)|\theta)= y((\varnothing,\varnothing)|\theta) +\epsilon/2.
    \end{align*}
    We leave every other component unchanged, \emph{i.e.}:
    \begin{align*}
        y^\circ(t'|\theta')= y(t'|\theta'), \quad \sigma^\circ_1(t'|\theta')= \sigma_1(t'|\theta'),
    \end{align*}
    for every pair $\theta' \in \typeSet,t' \in \tSet$ such that $\theta' \neq \theta$ or $t' \notin \{(2,\varnothing),(\varnothing,\varnothing)\}$.
    
    We argue that this new mechanism is BIC and IR.
    The mechanism $(y^\circ,\sigma^\circ)$ is IR since $L'\ge L$ ($L \ge \bar v_2$ in order to satisfy the IR constraints in the original mechanism), and the other contracts are left unchanged.
    The BIC constraints for agent $2$ hold, as they hold for the original mechanism.

    Denote by $U_i(\hat\theta_i|\theta_i)$ and $U^\circ_i(\hat\theta_i|\theta_i)$ the expected utility of agent $i$ when they misreport $\hat\theta_i \in \typeSet_i$ in place of $\theta_i \in \typeSet_i$ in the old and new mechanisms, respectively.
    Now suppose that agent $1$ has type $\tilde\theta$, which cannot complete task $2$.
    The new mechanism does not change the probability that this agent receives task $1$, nor the corresponding contracts when this happens.
    Therefore, the agent is not incentivized to misreport.
    Instead, if the agent has type $\bar\theta$, we observe that:
    \begin{align*}
        U^\circ_1(\bar\theta|\bar\theta) - U_1(\bar\theta|\bar\theta) = \frac{\epsilon}{2}(\bar p_2 L' - \bar c_2) - \epsilon(\bar p_2 L - \bar c_2) =0
        % U_1'(\bar\theta|\bar\theta) - U_1(\bar\theta|\bar\theta) = \frac{\epsilon}{2}(\bar p_2 L' - \bar c_2) - \epsilon(\bar p_2 L - \bar c_2) 
        % = \frac{\epsilon}{2}( 2\bar p_2 L - \bar p_2 \bar v_2 -\bar c_2 ) - \epsilon(\bar p_2 L - \bar c_2)  
        % = \frac{\epsilon}{2}( 2\bar p_2 L - -2\bar c_2 ) - \epsilon(\bar p_2 L - \bar c_2) = 0
    \end{align*}
    At the same time, we have $U_1(\tilde\theta|\bar\theta)=U^\circ_1(\tilde\theta|\bar\theta)$, as the two mechanisms have the same behavior when the agent reports $\tilde\theta$.
    Therefore, when employing mechanism $(y^\circ,\sigma^\circ)$, the agent's utility is unchanged both upon truthful and non-truthful reporting.
    Since the original mechanism is BIC, we conclude that $(y^\circ,\sigma^\circ)$ is also BIC.
    
    Finally, the profit $\profit(y^\circ,\sigma^\circ)$ of the new mechanism is strictly larger than the profit $\profit(y,\sigma)$ of the original mechanism.
    Indeed, by recalling that $L' = 2L - \bar v_2 = 2L -2$, the difference between these profits is:
    \begin{align*}
        \profit(y^\circ,\sigma^\circ) -\profit(y,\sigma) &= \mu_\theta\big(y^\circ((2,\varnothing)|\theta) \bar p_2 (1-L') -y((2,\varnothing)|\theta)\bar p_2 (1-L) \big) \\
        &=\mu_\theta \left(\epsilon\bar p_2 (\frac{3}{2}-L) - \epsilon \bar p_2 (1-L)\right) > 0.
    \end{align*}
    
    A similar argument applies when $y((2,\varnothing)|\theta)=0$ and $y((2,1)|\theta) = \epsilon > 0$. 
    In this case, we move probability mass $\varepsilon/2$ from allocation $(2,1)$ to allocation $(\varnothing,1)$, and set $\sigma_1'((2,1)|\theta)=2L-\bar v_2$, where $L \coloneq \sigma_1((2,1)|\theta)$. 
    As in the previous case, the IR and BIC constraints of agent 1 remain satisfied, while its contribution to the principal's profit strictly increases; the argument is the same, since we moved probability mass from an allocation where agent $1$ receives task $2$ to an allocation where she receives no task.
    % given \theta, original mechanism: with probability \epsilon agent 2 gets 1, agent 1 gets 2
    % given \theta, new mwchanism: with probability \epsilon agent 2 gets 1, agent 1 get sometimes 2 sometimes nothing. 
    % Agent 2 does not care waht agent 1 gets, so from the point of view of agent 1 nothing changes.
    Furthermore, we can assume that $\sigma_2((2,1)|\theta)=\sigma_2((\varnothing,1)|\theta)$, as otherwise the payments of agent 2 can be changed to satisfy this property without altering the revenue (Lemma~\ref{lem:contracts_ind_from_others_tasks}).
    Since agent 2 receives the same task with the same total probability and contract as in the original mechanism, its IR and BIC constraints, as well as its contribution to the principal's profit, remain unchanged. Thus, the modified mechanism is BIC and IR and achieves strictly higher profit.

    We have thus proven that the supremum is not achieved by a mechanism that never gives task $2$ to agents $1$ or $2$.
    Furthermore, any other mechanism is not optimal, since we can build a new mechanism with larger profit.
    We conclude that no mechanism is optimal.
\end{proof}

\subsection{Computing an approximately optimal mechanism}
% Given that the problem of computing a profit-maximizing mechanism in a multi-task problem does not admit an optimum, we focus on the problem of developing an FPTAS.
% Specifically, we now show how to compute an $\epsilon$-optimal mechanism $(y,\sigma)$, for any $\epsilon > 0$ in polynomial time (in $M$,$K$ and $\log(\nicefrac{1}{\epsilon})$) when the number of agents $N$ is constant.
Since the problem of computing a profit-maximizing mechanism in a multi-task setting may not admit an optimal solution, we focus on developing an FPTAS. 
Specifically, we show how to compute, for any $\epsilon>0$, an $\epsilon$-optimal mechanism $(y,\sigma)$ achieving a profit of at least $(1-\epsilon)\SUP$ in time polynomial in the instance size and $\log(\nicefrac{1}{\epsilon})$, provided that the number of agents $N$ is constant.
Our technique is similar to the one proposed by~\cite{cacciamani2024multi} to compute approximately optimal randomized contracts in multi-agent settings.
As in~\cite{cacciamani2024multi}, we derive an upper bound $B$ on the magnitude of the contracts employed by an $\epsilon$-optimal solution; then we show that computing the optimal solution with bounded contracts essentially reduces to solving a linear program with polynomially many variables and constraints.
The main difference with respect to~\cite{cacciamani2024multi} is that we derive a multiplicative FPTAS, rather than an additive one.

As we did in Section~\ref{sec:single_task}, for every mechanism $(y,\sigma)$  we define the expected contracts $\nu_i(t|\theta) = y(t|\theta) \sigma_i(t|\theta)$, for every agent $i \in [N]$, profile $\theta \in \typeSet$ and allocation $t \in \tSet$.
We denote $(y,\nu)$ as the $\nu$-representation of $(y,\sigma)$.
The $(y,\sigma)$ representation can be recovered by defining $\sigma_i(t|\theta) = \frac{\nu_i(t|\theta)}{y(t|\theta)}$, with the convention that $0/0=0$.
With an abuse of notation, we will refer to the tuple $(y,\nu)$ as a mechanism.
Let us remark that a tuple $(y,\nu)$, with $y : \typeSet \rightarrow \Delta_{\tSet}$ and $\nu : \tSet \times \typeSet \rightarrow \mathbb{R}^N_+$, is a mechanism if and only if $\nu_i(t|\theta)=0$ whenever $y(t|\theta)=0$; otherwise, there is no contract $\sigma_i(t|\theta)$ such that $\nu_i(t|\theta)= y(t|\theta) \sigma_i(t|\theta)$.
We now show how to compute, for every $B>0$, the profit-maximizing BIC and IR mechanism that uses payments at most $B$.

When working in the $\nu$-representation, by simple variable substitution the profit of the mechanism can be rewritten as:
\begin{equation*}
    \profit(y,\nu) = \sum_{\theta \in \typeSet} \mu_\theta \sum_{t \in \tSet} \sum_{i \in [N]} p_{\theta_i,t_i}(y(t|\theta) -\nu_i(t|\theta)).
\end{equation*}
Similarly, the IR and IC constraints can be rewritten as:
\begin{align}
    &\nu_i(t|\theta) p_{\theta_i,t_i} -y(t|\theta)c_{\theta_i,t_i} \ge 0 \quad &&\forall i \in [N], \theta \in \typeSet, t \in \tSet \label{eq:constraint_ir_multi_task_nu}\\
    &\sum_{\theta_{-i} \in \typeSet_{-i}} \mu_{\theta_{-i}}\sum_{t \in \tSet} \nu_i(t|\theta_i,\theta_{-i}) p_{\theta_i,t_i} -y(t|\theta_i,\theta_{-i}) c_{\theta_i,t_i} \ge &&\nonumber\\
    &\hspace{2cm}\sum_{\theta_{-i} \in \typeSet_{-i}} \mu_{\theta_{-i}}\sum_{t \in \tSet} [\nu_i(t|\hat\theta_i,\theta_{-i}) p_{\theta_i,t_i} -y(t|\hat\theta_i,\theta_{-i}) c_{\theta_i,t_i}]_+ &&\forall i \in [N], \theta_i,\hat\theta_i \in \typeSet_i
    \label{eq:contraint_ic_multi_task_nu}
\end{align}
The profit and the IR constraints are thus linear in $y$ and $\nu$.
To linearize the IC constraints, we introduce the auxiliary variables $d_i(t,\hat\theta_i|\theta_i,\theta_{-i}) \in \mathbb{R}$, which represent the utility of the agent $i$ when it misreports $\hat\theta_i$ in place of $\theta_i$, and the other agents report $\theta_{-i}$. 
We then construct the space $\Psi$ of tuples $(y,\nu,d)$, consisting of all the tuples satisfying the following linear constraints:
\begin{subequations}
    \label{eq:psi_definition}
\begin{align}
    &\nu_i(t|\theta) p_{\theta_i,t_i} -y(t|\theta)c_{\theta_i,t_i} \ge 0 \quad &\forall i \in [N], \theta \in \typeSet, t \in \tSet  \label{seq:constraint_ir_multi_task_nu}\\
     &\sum_{\theta_{-i} \in \typeSet_{-i}} \mu_{\theta_{-i}}\sum_{t \in \tSet} \nu_i(t|\theta_i,\theta_{-i}) p_{\theta_i,t_i} -y(t|\theta_i,\theta_{-i}) c_{\theta_i,t_i} \ge \nonumber\\ 
    &\hspace{4cm}\sum_{\theta_{-i} \in \typeSet_{-i}} \mu_{\theta_{-i}}\sum_{t \in \tSet} d_i(t,\hat\theta_i|\theta_i,\theta_{-i})\hspace{-1cm} &\forall i \in [N], \theta_i,\hat\theta_i \in \typeSet_i \label{seq:constraint_ic_multi_task_linearized} \\
    &d_i(t,\hat\theta_i|\theta_i,\theta_{-i}) \ge \nu_i(t|\hat\theta_i,\theta_{-i}) p_{\theta_i,t_i} -y(t|\hat\theta_i,\theta_{-i}) c_{\theta_i,t_i} &\forall i \in [N], \theta_i,\hat\theta_i \in \typeSet_i, \theta_{-i} \in \typeSet_{-i}, t \in \tSet  \label{seq:constraint_ic_multi_task_d_greater}\\
    &d_i(t,\hat\theta_i|\theta_i,\theta_{-i}) \ge 0 &\forall i \in [N], \theta_i,\hat\theta_i \in \typeSet_i, \theta_{-i} \in \typeSet_{-i}, t \in \tSet \label{seq:constraint_ic_multi_task_d_positive}
\end{align}
\end{subequations}
Constraint~\eqref{seq:constraint_ic_multi_task_linearized} corresponds to the BIC constraints of Equation~\eqref{eq:contraint_ic_multi_task_nu}, while Constraint~\eqref{seq:constraint_ic_multi_task_d_greater} and Constraint~\eqref{seq:constraint_ic_multi_task_d_positive} ensure that $d_i(t,\hat\theta_i|\theta_i,\theta_{-i})$ is at least the utility attained by agent $i$ when deviating. 
Notice that the smaller the values of the variables $d_i(t,\hat\theta_i|\theta_i,\theta_{-i})$, the larger the feasible space for $y$ and $\nu$. 
Therefore, whenever we maximize $\profit(y,\nu)$ within the space $\Psi$, it is without loss of generality to focus on the solutions where $d_i(t,\hat\theta_i|\theta_i,\theta_{-i})$ is exactly the utility attained by the agent deviating.

Let us remark that if $(y,\nu,d) \in \Psi$, it does not necessarily mean that $(y,\nu)$ is a mechanism.
Indeed, we must ensure that $\nu_i(t|\theta)=0$ if $y(t|\theta) = 0$, as otherwise we cannot recover a suitable contract $\sigma_i(t|\theta)$ such that $\nu_i(t|\theta)=\sigma_i(t|\theta)y(t|\theta)$.
Therefore, when working with expected contracts $\nu$ and unbounded contracts, we face the following optimization problem:
\begin{subequations}
\label{prog:multi_task_with_linearized_ic_and_no_contract_bounds} 
\begin{align}
    \sup_{y,\nu,d}\;\; &\profit(y,\nu) = \sum_{\theta \in \typeSet} \mu_\theta \sum_{t \in \tSet} \sum_{i \in [N]} p_{\theta_i,t_i}(y(t|\theta) -\nu_i(t|\theta)) \\
    \text{s.t.} \;\; &(y,\nu,d) \in \Psi \\
    &y(t|\theta) > 0 \vee \nu_i(t|\theta)=0 \label{seq:constraint_no_critical_pairs} &\forall i \in [N], t \in \tSet, \theta \in \typeSet. 
\end{align}
\end{subequations}
Constraint~\eqref{seq:constraint_no_critical_pairs} ensures that $\nu_i(t|\theta)=0$ whenever $y(t|\theta)=0$, which allows us to recover a mechanism $(y,\sigma)$. 
The following result proves the equivalence between Program~\ref{prog:multi_task_with_linearized_ic_and_no_contract_bounds} and the problem of maximizing the profit $\profit(y,\sigma)$ among the IR and BIC mechanisms $(y,\sigma)$.
\begin{restatable}{lemma}{ProgMultiTask}
    \label{lem:prog_multi_task_lienar_ic}
    The value of Program~\ref{prog:multi_task_with_linearized_ic_and_no_contract_bounds} is $\SUP$. 
    For every feasible solution $(y,\nu,d)$, the pair $(y,\nu)$ is the $\nu$-representation of a BIC and IR mechanism whose profit is the value of the solution $(y,\nu,d)$.
    For every BIC and IR mechanism $(y,\nu)$, there is a suitable assignment of variables $d$ such that $(y,\nu,d)$ is a feasible solution of Program~\ref{prog:multi_task_with_linearized_ic_and_no_contract_bounds}, and its value is $\profit(y,\nu)$. 
\end{restatable}
\begin{proof}
    Let $(y,\sigma)$ be a BIC and IR mechanism.
    Define $(y,\nu,d)$ as:
    \begin{align*}
        % \sigma_i(t|\theta) &= \frac{\nu_i(t|\theta)}{y(t|\theta)} &\forall i \in [N], t \in \tSet, \theta \in \typeSet, \\
        \nu_i(t|\theta) &= y(t|\theta) \sigma_i(t|\theta) &\forall i \in [N], t \in \tSet, \theta \in \typeSet, \\
        d_i(t,\hat\theta_i|\theta_i,\theta_{-i}) &= [\nu_i(t|\hat\theta_i,\theta_{-i}) p_{\theta_i,t_i} -y(t|\hat\theta_i,\theta_{-i}) c_{\theta_i,t_i}]_+ 
    \end{align*}
    The solution $(y,\nu,d)$ satisfies Constraints~\eqref{seq:constraint_ic_multi_task_d_greater}, \eqref{seq:constraint_ic_multi_task_d_positive} and \eqref{seq:constraint_no_critical_pairs} by construction.
    Furthermore, by simple variable substitution one can verify that the remaining constraints correspond to the IR and BIC constraints of mechanism $(y,\sigma)$, which are satisfied.
    Therefore $(y,\nu,d)$ is feasible for Program~\ref{prog:multi_task_with_linearized_ic_and_no_contract_bounds}.
    By noticing that $\profit(y,\sigma)=\profit(y,\nu)$, we conclude that the value of Program~\ref{prog:multi_task_with_linearized_ic_and_no_contract_bounds} is at least $\SUP$.

    Now let $(y,\nu,d)$ be a feasible solution of Program~\ref{prog:multi_task_with_linearized_ic_and_no_contract_bounds}.
    Recover the $(y,\sigma)$ representation of $(y,\nu)$ defined as:
    \begin{equation*}
        \sigma_i(t|\theta) = \frac{\nu_i(t|\theta)}{y(t|\theta)} \quad \forall i \in [N], t \in \tSet, \theta \in \typeSet,
    \end{equation*}
    with the convention that $0/0 \coloneqq 0$.
    Notice that this is well-defined thanks to Constraint~\eqref{seq:constraint_no_critical_pairs}, and it is BIC and IR thanks to the remaining constraints.
    Since $\profit(y,\sigma)=\profit(y,\nu)$, we have that Program~\ref{prog:multi_task_with_linearized_ic_and_no_contract_bounds} has value at most $\SUP$.
    We conclude that the value of the program is $\SUP$.
\end{proof}

By Lemma~\ref{lem:prog_multi_task_lienar_ic}, we have reduced the problem of devising an FPTAS to that of computing an $\epsilon$-optimal solution of Program~\ref{prog:multi_task_with_linearized_ic_and_no_contract_bounds}.
Program~\ref{prog:multi_task_with_linearized_ic_and_no_contract_bounds} has a single non-linear constraint, namely Constraint~\ref{seq:constraint_no_critical_pairs}.
Now suppose that there exists an $\epsilon$-optimal mechanism whose payments are bounded by some $B > 0$.
Then, we argue that an  $\epsilon$-optimal mechanism can be found by solving the following linear program:
\begin{subequations}
\label{prog:multi_task_lp_with_bounded_contracts} 
\begin{align}
    \max_{y,\nu,d}\;\; &\profit(y,\nu) = \sum_{\theta \in \typeSet} \mu_\theta \sum_{t \in \tSet} \sum_{i \in [N]} p_{\theta_i,t_i}(y(t|\theta) -\nu_i(t|\theta)) \\
    \text{s.t.} \;\; &(y,\nu,d) \in \Psi \\
    &\nu_i(t|\theta) \le B y(t|\theta) \label{seq:lp_bounded_contracts_constraint_m} &\forall i \in [N], t \in \tSet, \theta \in \typeSet,
\end{align}
\end{subequations}
which is derived by substituting Constraint~\ref{seq:constraint_no_critical_pairs} with Constraint~\ref{seq:lp_bounded_contracts_constraint_m}.
Notice that Constraint~\ref{seq:lp_bounded_contracts_constraint_m} is both stronger than Constraint~\ref{seq:constraint_no_critical_pairs},
and satisfied by any mechanism whose payments are bounded by $B$.
Therefore, we have the following result.
\begin{restatable}{lemma}{LPMultiTask}
    \label{lem:multi_task_lp}
    Let $(y,\nu,d)$ be an optimal solution to Program~\ref{prog:multi_task_lp_with_bounded_contracts}.
    Suppose that there exists an $\epsilon$-optimal valid mechanism whose contracts are bounded by $B > 0$.
    Then $(y,\nu)$ is the $\nu$-representation of an $\epsilon$-optimal valid mechanism.
\end{restatable}
\begin{proof}
    Let $(y^\epsilon,\sigma^\epsilon)$ be an $\epsilon$-optimal valid mechanism whose contracts are bounded by $B > 0$.
    By Lemma~\ref{lem:prog_multi_task_lienar_ic}, we can build a feasible solution $(y^\epsilon,\nu^\epsilon,d^\epsilon)$ to Program~\ref{prog:multi_task_with_linearized_ic_and_no_contract_bounds}, such that $\profit(y^\epsilon,\nu^\epsilon) = \profit(y^\epsilon,\sigma^\epsilon)$.
    Such a solution is feasible for Program~\ref{prog:multi_task_lp_with_bounded_contracts}, as we have:
    \begin{equation*}
        \nu^\epsilon_i(t|\theta) = y^\epsilon(t|\theta)\sigma^\epsilon_i(t|\theta) \le y^\epsilon(t|\theta)B,
    \end{equation*}
    where the equality holds by construction, and the inequality by assumption.
    Let $(y,\nu,d)$ be an optimal solution to Program~\ref{prog:multi_task_lp_with_bounded_contracts}, so that:
    \begin{equation*}
        \profit(y,\nu) \ge \profit(y^\epsilon,\nu^\epsilon) \ge (1-\epsilon)\SUP.
    \end{equation*}
    Observe that any feasible solution of Program~\ref{prog:multi_task_lp_with_bounded_contracts} is also feasible for Program~\ref{prog:multi_task_with_linearized_ic_and_no_contract_bounds}, as Constraint~\eqref{seq:lp_bounded_contracts_constraint_m} implies Constraint~\eqref{seq:constraint_no_critical_pairs}.
    By Lemma~\ref{lem:prog_multi_task_lienar_ic}, $(y,\nu)$ corresponds to a valid mechanism whose profit is at least $(1-\epsilon)\SUP$, concluding the proof.
\end{proof}

% \bollo{Da qualche parte servirà un teorema finale "For any constant number of agents $N$ and every $\epsilon >0$, it is possible to compute an $\epsilon$-optimal mechanism in time polynomial in the instance size and $\log(1/\epsilon)$.}
% \bollo{Overall mi serve un modo formale per stabilire quell'instance size con $N$ costante, perchè non è semplicemente polinomiale in $M$ e $K$ ma anche nel numero di bit. Però senza N costante non è polinomiale nel size dell'istanza, perchè le allocazioni e $B$ vanno come $M^N$ e l'istanza solo $K^N$ (se il prior è correlato e esplicito).}

\subsection{An upper bound on the contracts}

In the remainder of this Section, we show how to compute the upper bound $B$ to be employed in Program~\ref{prog:multi_task_lp_with_bounded_contracts}.
Specifically, we show that there always exists an $\epsilon$-optimal mechanism whose contracts are bounded by some $B$ with a binary representation of length polynomial in the instance size and $\log(\nicefrac{1}{\epsilon})$, provided that $N$ is constant.
This allows us to solve Program~\ref{prog:multi_task_lp_with_bounded_contracts} and recover this $\epsilon$-optimal solution in polynomial time.
W.l.o.g. we will assume that $\SUP > 0$, as when $\SUP=0$ any $B>0$ is a valid bound.

At a high level, our proof works as follows.
We build an LP, call it LP1, by removing the non-linear constraint from Program~\ref{prog:multi_task_with_linearized_ic_and_no_contract_bounds}.
The optimal solution $(y,\nu,d)$ of LP1 has value at least $\SUP$, albeit $(y,\nu)$ may \emph{not} be the $\nu$-representation of a mechanism.
We then build a new LP, call it $\text{LP}(\bar t, \bar \theta)$, for each $\bar t \in \tSet$ and $\bar\theta \in \typeSet$, by changing the objective function to maximize $y(\bar t | \bar\theta)$.
Suppose for the moment that each $\text{LP}(\bar t, \bar \theta)$ has an optimal solution $(y^{\bar t, \bar \theta},\nu^{\bar t, \bar \theta},d^{\bar t, \bar \theta})$ with $y(\bar t | \bar\theta)>0$ and positive profit.
Then, we build $(y^\epsilon,\nu^\epsilon,d^\epsilon)$ as the convex combination of $(y,\nu,d)$, with weight $(1-\epsilon)$, and the solutions $\{(y^{\bar t, \bar \theta},\nu^{\bar t, \bar \theta},d^{\bar t, \bar \theta})\}_{\bar t, \bar \theta}$.
The new solution satisfies Constraint~\eqref{seq:constraint_no_critical_pairs}, since all variables $y^\epsilon(\bar t | \bar\theta)$ are positive thanks to the contributions from the solution $y^{\bar t, \bar \theta}(\bar t | \bar\theta)$.
It also belongs to $\Psi$ as it is a convex combination of solutions in this set.
Therefore, it is feasible for Program~\ref{prog:multi_task_with_linearized_ic_and_no_contract_bounds}.
Noticing that the variables in the solutions to LP1 and $\text{LP}(\bar t, \bar \theta)$ can be represented by polynomially many bits (see, \emph{e.g.},~\cite{bertsimas1997introduction}), we recover an upper bound $B$ on the contracts of the mechanism  $(y^\epsilon,\nu^\epsilon)$.
Its profit must also be at least $(1-\epsilon)\profit(y,\nu) \ge (1-\epsilon)\SUP$, as the profit is linear in $y$ and $\nu$.
We have therefore built an $\epsilon$-optimal mechanism whose contracts have polynomially many bits.

We now discuss the formal argument.
The main difference with respect to the outline above is that we cannot assume the existence of $(y^{\bar t, \bar \theta},\nu^{\bar t, \bar \theta},d^{\bar t, \bar \theta})$ for every $(\bar t,\bar\theta) \in \tSet \times \typeSet$.
Therefore, let us denote by $\Splus \subseteq \tSet \times \typeSet$ the set of pairs $(\bar t,\bar\theta)$ such that there exists a BIC and IR mechanism $(y,\sigma)$ with $y(\bar t|\bar\theta) >0$ and non-negative profit $\profit(y,\sigma) \ge 0$.
Denote by $\Sminus = \tSet \times \typeSet \setminus \Splus$ the other pairs.

Now we define the following linear program, which plays the role of LP1:
\begin{subequations}
\label{prog:multi_task_lp_construction} 
\begin{align}
    \max_{y,\nu, d}\;\; &\profit(y,\nu) \\
    \text{s.t.} \;\; &(y, \nu, d) \in \Psi \\
    &y(\bar t|\bar \theta)=\nu_i(\bar t| \bar \theta) = 0 &\forall i \in [N], (\bar t, \bar\theta) \in \Sminus \label{seq:constr_lp_construction}. 
\end{align}
\end{subequations}
Program~\ref{prog:multi_task_lp_construction} is built from Program~\ref{prog:multi_task_with_linearized_ic_and_no_contract_bounds} in two steps.
First, we impose that, for every $(\bar t, \bar\theta) \in \Sminus$, the mechanism does not assign allocation $\bar t$ given type report $\bar\theta$.
This does not reduce the value of the program, as the $\nu$-representation of every mechanism with positive profit satisfies the new constraint.
Second, we remove the non-linear constraint, so that Program~\ref{prog:multi_task_lp_construction} is an LP and its value is at least $\SUP$.
We formalize this property in the following lemma.
\begin{restatable}{lemma}{LpOneAtLeastSup}
    \label{lem:lp_one_at_least_sup}
    If $\SUP > 0$, then the value of Program~\ref{prog:multi_task_lp_construction} is at least $\SUP$.
\end{restatable}
\begin{proof}
    By Lemma~\ref{lem:prog_multi_task_lienar_ic}, Program~\ref{prog:multi_task_with_linearized_ic_and_no_contract_bounds} has value at least $\SUP$.
    Consider any feasible solution $(y,\nu,d)$ for such a program with value strictly larger than zero (which exists since $\SUP > 0$).
    We will prove it is also feasible for Program~\ref{prog:multi_task_lp_construction}.
    Since the objective function of the two programs is the same, this proves the statement.

    In order to prove that $(y,\nu,d)$ is a feasible solution of Program~\ref{prog:multi_task_lp_construction}, we observe that it belongs to $\Psi$, as it is feasible for Program~\ref{prog:multi_task_with_linearized_ic_and_no_contract_bounds}.
    Furthermore, by Lemma~\ref{lem:prog_multi_task_lienar_ic}, $(y,\nu)$ is the $\nu$-representation of a valid mechanism with positive profit.
    Consider any pair $(\bar t,\bar \theta) \in \Sminus$.
    By definition of $\Sminus$, no valid mechanism  $(\tilde y,\tilde \nu)$ can have $\tilde y(\bar t|\bar \theta) > 0$ and positive profit.
    Hence, it must be the case that $y(\bar t|\bar \theta) = 0$, and since it satisfies Constraint~\eqref{seq:constraint_no_critical_pairs}, we also have $\nu_i(\bar t | \bar \theta)=0$ for every agent $i \in [N]$.
    We conclude that $(y,\nu,d)$ satisfies Constraint~\eqref{seq:constr_lp_construction} and is a feasible solution of Program~\ref{prog:multi_task_lp_construction}.
    % \bollo{It actually holds for $\SUP=0$ too, since in that case $\Sminus$ consists of all the pairs, the optimal solution of the LP is the mechanisms that never gives the task.}
\end{proof}

We now define the programs $\text{LP}(\bar t, \bar \theta)$, for every pair $(\bar t,\bar\theta) \in \Splus$, as follows:
\begin{subequations}
\label{prog:multi_task_lp_t_bar_y_bar} 
\begin{align}
    \max_{y^{\bar t, \bar \theta},\nu^{\bar t, \bar \theta},d^{\bar t, \bar \theta}}\;\; &y^{\bar t, \bar \theta}(\bar t | \bar \theta) \\
    \text{s.t.} \;\; &(y^{\bar t, \bar \theta},\nu^{\bar t, \bar \theta},d^{\bar t, \bar \theta}) \in \Psi \\
    &\profit(y^{\bar t, \bar \theta},\nu^{\bar t, \bar \theta}) \ge 0 \label{seq:lp_bar_non_negative_rev} \\
    &y^{\bar t, \bar \theta}(t|\theta)=\nu^{\bar t, \bar \theta}_i(t| \theta) = 0 &\forall i \in [N], (t, \theta) \in \Sminus
\end{align}
\end{subequations}

We are now ready to provide an upper bound $B$ on the magnitude of the contracts of an $\epsilon$-optimal valid mechanism.
The proof is constructive, \emph{i.e.}, it provides the value of $B$.
Furthermore, the number of bits required to represent $B$ is polynomial in the instance size and $\log(1/\epsilon)$ when $N$ is constant.
Thanks to Lemma~\ref{lem:multi_task_lp}, we can thus build an FPTAS when $N$ is constant.
% \bollo{Non è esattamente $M$ poinomiale, è la sua rappresentazione ad esserlo}
\begin{restatable}{theorem}{BoundOnContracts}
    For every $\epsilon>0$ and every constant number of agents $N$, there exists an $\epsilon$-optimal mechanism whose contracts are upper bounded by some $B>0$, whose size is polynomial in the instance size and $\nicefrac{1}{\epsilon}$.
    Furthermore, an $\epsilon$-optimal mechanism can be computed in time polynomial in the instance size and $\log(1/\epsilon)$.
\end{restatable}
\begin{proof}
    We will provide the upper bound $B$ on the magnitude of the contracts. 
    Then, the existence of an FPTAS for constant $N$ follows from Lemma~\ref{lem:multi_task_lp} and the fact that Program~\ref{prog:multi_task_lp_with_bounded_contracts} has polynomially many variables and constraints, and all the coefficients (including $B$) can be represented by polynomially many bits. 
    
    In the following we assume $\text{SUP}>0$, as otherwise there exists an optimal mechanism that never assigns any task.
    Let $(y^{\bar t, \bar \theta},\nu^{\bar t, \bar \theta},d^{\bar t, \bar \theta})$ be the optimal solution to $\text{LP}(\bar t, \bar\theta)$, for each $(\bar t,\bar \theta) \in \Splus$.
    Denote by $S$ the number of such solutions, and notice that $S \ge 1$, as otherwise there exists no valid mechanism that collects positive profit.
    Furthermore, by the definition of $\Splus$ and Lemma~\ref{lem:prog_multi_task_lienar_ic}, it must be the case that $y^{\bar t, \bar \theta}(\bar t|\bar\theta) > 0$ strictly.
    % Indeed, from any valid mechanism $(y',\sigma')$ we can build $(y',\sigma',d') \in \Psi$ (Lemma~\ref{lem:prog_multi_task_lienar_ic}).
    % Taking a mechanism with positive profit (which must have $y'$)
    Let also $(y,\nu,d)$ be the optimal solution to Program~\ref{prog:multi_task_lp_construction}.

    We build an $\epsilon$-optimal solution $(y^\epsilon,\nu^\epsilon,d^\epsilon)$ to Program~\ref{prog:multi_task_with_linearized_ic_and_no_contract_bounds} as a convex combination between $(y,\nu,d)$, with weight $(1-\epsilon$), and $\{(y^{\bar t, \bar \theta},\nu^{\bar t, \bar \theta},d^{\bar t, \bar \theta})\}_{(\bar t,\bar \theta) \in \Splus}$, with weight $\nicefrac{\epsilon}{S}$ each. 
    Specifically:
    \begin{align*}
        y^\epsilon(t|\theta) &= (1-\epsilon) y(t|\theta) +\frac{\epsilon}{S} \sum_{(\bar t,\bar \theta) \in \Splus} y^{\bar t,\bar \theta}(t|\theta)    &\forall t \in \tSet, \theta \in \typeSet, \\
        \nu_i^\epsilon(t|\theta) &= (1-\epsilon) \nu_i(t|\theta) +\frac{\epsilon}{S} \sum_{(\bar t,\bar \theta) \in \Splus} \nu_i^{\bar t,\bar \theta}(t|\theta)    &\forall t \in \tSet, \theta \in \typeSet, i \in [N], \\
        d_i^\epsilon(t,\hat\theta_i|\theta) &= (1-\epsilon) d_i(t,\hat\theta_i|\theta) +\frac{\epsilon}{S} \sum_{(\bar t,\bar \theta) \in \Splus} d_i^{\bar t,\bar \theta}(t,\hat\theta_i|\theta)    &\forall t \in \tSet, \theta \in \typeSet, i \in [N], \hat\theta_i \in \typeSet_i.
    \end{align*}
    We will now prove that it is feasible, $\epsilon$-optimal, and its contracts have polynomial encoding length.
    
    \paragraph{Feasibility} This solution belongs to the convex set $\Psi$, as it is a convex combination of solutions within it.
    It also satisfies Constraint~\ref{seq:constraint_no_critical_pairs}, since for every $(\bar t,\bar \theta) \in \Splus$ the contribution from $y^{\bar t,\bar \theta}(\bar t|\bar \theta)$ ensures that $y^\epsilon(\bar t|\bar \theta) > 0$.
    Formally:
    \begin{align*}
        y^\epsilon(\bar t|\bar \theta) &= (1-\epsilon) y(\bar t|\bar \theta) +\frac{\epsilon}{S} \sum_{(\tilde t,\tilde \theta) \in \Splus} y^{\tilde t,\tilde \theta}(\bar t|\bar \theta) \ge \frac{\epsilon}{S} y^{\bar t,\bar \theta}(\bar t|\bar \theta) >0,
    \end{align*}
    where the first inequality holds since the allocation probabilities are non-negative, and the strict inequality by definition of $\Splus$.
    For the remaining pairs $(\bar t,\bar \theta) \in \Sminus$, we have $y^\epsilon(\bar t|\bar\theta) = \nu^\epsilon_i(\bar t|\bar \theta)=0$ as imposed in both Programs~\ref{prog:multi_task_lp_construction} and~\ref{prog:multi_task_lp_t_bar_y_bar}.
    % As a result, $(y^\epsilon,\nu^\epsilon,d^\epsilon)$ is a feasible solution to Program~\ref{prog:multi_task_with_linearized_ic_and_no_contract_bounds}.
    As a result, $(y^\epsilon,\nu^\epsilon,d^\epsilon)$ is feasible for Program~\ref{prog:multi_task_with_linearized_ic_and_no_contract_bounds}.

    \paragraph{Optimality} By linearity of the objective function $\profit(y^\epsilon,\nu^\epsilon)$, we have:
    \begin{equation*}
        \profit(y^\epsilon,\nu^\epsilon) = (1-\epsilon)\profit(y,\nu) +\frac{\epsilon}{S}\sum_{(\bar t,\bar \theta) \in \Splus}\profit(y^{\bar t,\bar \theta},\nu^{\bar t,\bar \theta}) \ge (1-\epsilon)\profit(y,\nu) \ge (1-\epsilon)\SUP,
    \end{equation*}
    where the first inequality holds since Program~\ref{prog:multi_task_lp_t_bar_y_bar} imposes non-negative profit, and the second as $\profit(y,\nu) \ge \SUP$ by Lemma~\ref{lem:lp_one_at_least_sup}.

    \paragraph{Bound on the contracts}
    Both Program~\ref{prog:multi_task_lp_construction} and Program~\ref{prog:multi_task_lp_t_bar_y_bar} have polynomially many variables and constraints when $N$ is constant.
    Hence, there exists some $G > 0$, represented by polynomially many bits, such that every variable of their optimal solution is either zero or in the range $[1/G,G]$ (see, e.g.~\cite{bertsimas1997introduction}).
    By Lemma~\ref{lem:prog_multi_task_lienar_ic} and the arguments provided above, $(y^\epsilon,\nu^\epsilon)$ is the $\nu$-representation of a BIC and IR $\epsilon$-optimal mechanism.
    For every pair $(t,\theta) \in \Splus$ and agent $i \in [N]$, the contract $\sigma^\epsilon_i(t|\theta)$ can be upper bounded as:
    \begin{align*}
        \sigma^\epsilon_i(t|\theta) &= \frac{\nu^\epsilon_i(t|\theta)}{y^\epsilon(t|\theta)} \\
        &=\frac{(1-\epsilon) \nu_i(t|\theta) +\frac{\epsilon}{S} \sum_{(\bar t,\bar \theta) \in \Splus} \nu_i^{\bar t,\bar \theta}(t|\theta)}{(1-\epsilon) y(t|\theta) +\frac{\epsilon}{S} \sum_{(\bar t,\bar \theta) \in \Splus} y^{\bar t,\bar \theta}(t|\theta)} \\
        &\le \frac{(1-\epsilon)G +\frac{\epsilon}{S}SG}{(1-\epsilon) y(t|\theta) +\frac{\epsilon}{S} \sum_{(\bar t,\bar \theta) \in \Splus} y^{\bar t,\bar \theta}(t|\theta)} \\
        &= \frac{G}{(1-\epsilon) y(t|\theta) +\frac{\epsilon}{S} \sum_{(\bar t,\bar \theta) \in \Splus} y^{\bar t,\bar \theta}(t|\theta)} \\
        &\le \frac{G}{\frac{\epsilon}{SG}} = \frac{SG^2}{\epsilon} \\
        &\le \frac{(M+1)^N K^N G^2}{\epsilon} \eqqcolon B,
    \end{align*}
    where the first inequality holds since every expected contract is at most $G$, the second as $y^{t,\theta}(t|\theta) \ge 1/G$ while all other allocation probabilities are non-negative, and the last one since $S \le (M+1)^N K^N$.
    For the remaining pairs $(t,\theta) \in \Sminus$, we have $\sigma^\epsilon_i(t|\theta)=0$, concluding the proof.
\end{proof}
    
    %\section*{ACKNOWLEDGMENTS}
    
    \newpage
    \printbibliography
    
    \newpage
    \appendix
    
    \section*{Appendix}
    \section{Omitted Proofs from Section~\ref{sec:single_task}}
\label{app:single_task}
\MonotoneNecessaryLemma*
\begin{proof}    
    % \bollo{Basically: IC with max => IC without max => monotone in v}
    From Equation~\eqref{eq:ic_with_critical_value}, the BIC mechanism $(x,\beta)$ satisfies:
    \begin{align*}
        \beta_{i,\theta_i} -x_{i,\theta_i}v_{\theta_i} &\ge \hspace{-0.2cm}\sum_{\theta_{-i} \in \typeSet_{-i}}\hspace{-0.2cm} \mu_{\theta_{-i}} [\beta_{i,\hat\theta_i, \theta_{-i}}  -x_{i,\hat\theta_i, \theta_{-i}}v_{\theta_i}]_+ \ge \hspace{-0.4cm}\sum_{\theta_{-i} \in \typeSet_{-i}}\hspace{-0.2cm} \mu_{\theta_{-i}} \beta_{i,\hat\theta_i, \theta_{-i}}  -x_{i,\hat\theta_i, \theta_{-i}}v_{\theta_i} = \beta_{i,\hat\theta_i} -x_{i,\hat\theta_i}v_{\theta_i}, \\
        \beta_{i,\hat\theta_i} -x_{i,\hat\theta_i}v_{\hat\theta_i} &\ge \hspace{-0.2cm}\sum_{\theta_{-i} \in \typeSet_{-i}}\hspace{-0.2cm} \mu_{\theta_{-i}} [\beta_{i,\theta_i, \theta_{-i}}  -x_{i,\theta_i, \theta_{-i}}v_{\hat\theta_i}]_+ \ge \hspace{-0.4cm}\sum_{\theta_{-i} \in \typeSet_{-i}}\hspace{-0.2cm} \mu_{\theta_{-i}} \beta_{i,\theta_i, \theta_{-i}}  -x_{i,\theta_i, \theta_{-i}}v_{\hat\theta_i} = \beta_{i,\theta_i} -x_{i,\theta_i}v_{\hat\theta_i}.
    \end{align*}
    Summing up the two inequalities above, we have:
    \begin{equation*}
        \beta_{i,\theta_i} -x_{i,\theta_i}v_{\theta_i} +\beta_{i,\hat\theta_i} -x_{i,\hat\theta_i}v_{\hat\theta_i} \ge \beta_{i,\hat\theta_i} -x_{i,\hat\theta_i}v_{\theta_i} +\beta_{i,\theta_i} -x_{i,\theta_i}v_{\hat\theta_i}.
    \end{equation*}
    Rearranging and simplifying:
    \begin{equation*}
        x_{i,\theta_i}(v_{\hat\theta_i}-v_{\theta_i}) \ge x_{i,\hat\theta_i}(v_{\hat\theta_i}-v_{\theta_i}).
    \end{equation*}
    Dividing both sides by $v_{\hat\theta_i}-v_{\theta_i}$ (notice that $v_{\hat\theta_i}-v_{\theta_i} > 0$), we have $x_{i,\theta_i} \ge x_{i,\hat\theta_i}$.
\end{proof}

\BetaForSingleTask*
\begin{proof}
    As the mechanism is regular, we can employ Lemma~\ref{lem:regular_ic_and_ir} to prove that it is IR and BIC by looking at the ex-interim values only. In the following we let $i \in [N]$.

    First, we consider the IR constraints of agent $i$ (Equation~\eqref{eq:regular_ir}).
    The IR constraint of type $\lastType_i$ is satisfied by definition, see Equation~\eqref{eq:beta_expected_profit_max}.
    Proceeding by induction over the types $\theta_i < \lastType_i$, we have:
    \begin{align*}
        \beta_{i,\theta_i} &= \beta_{i,\theta_i+1} +(x_{i,\theta_i}-x_{i,\theta_{i+1}})v_{\theta_i} \\& \ge x_{i,\theta_i+1} v_{\theta_{i+1}} +(x_{i,\theta_i}-x_{i,\theta_{i+1}})v_{\theta_i} \\&= x_{i,\theta_i+1}(v_{\theta_i+1}-v_{\theta_i}) +x_{i,\theta_i} v_{\theta_i} \ge x_{i,\theta_i}v_{\theta_i},
    \end{align*}
    where the first inequality holds by induction and the last one since $v_{\theta_i+1} \ge v_{\theta_i}$. 
    As a result, all IR constraints are satisfied for this choice of $\beta$.
    
    We now focus on the BIC constraints (Equation~\eqref{eq:regular_ic}).
    First, we observe that an agent  of type $\theta_i < \lastType_i$ is not incentivized to misreport $\theta_i+1$.
    Indeed, Equation~\eqref{eq:regular_ic} imposes:
    \begin{equation*}
        \beta_{i,\theta_i} \ge \beta_{i,\theta_i+1} +(x_{i,\theta_i} -x_{i,\theta_i+1}) v_{\theta_i},
    \end{equation*}
    that is clearly satisfied by Equation~\ref{eq:beta_expected_profit_max} by definition.
    Consider now a type $\theta_i <\lastType_i$ and a successive type $\theta_i+k$ for some positive integer $k \ge 1$.
    We proceed by induction.
    Suppose by induction that $$\beta_{i,\theta_i} -x_{i,\theta_i}v_{\theta_i} \ge \beta_{i,\theta_i+k-1} -x_{i,\theta_i+k-1} v_{\theta_i},$$ which we proved for the base case $k=1$.
    Observe that:
    \begin{align*}
        D &\coloneqq (\beta_{i,\theta_i+k-1} -x_{i,\theta_i+k-1} v_{\theta_i}) -(\beta_{i,\theta_i+k} -x_{i,\theta_i+k} v_{\theta_i})  \\
        &= (\beta_{i,\theta_i+k-1} -x_{i,\theta_i+k-1} v_{\theta_i+k-1}) - (\beta_{i,\theta_i+k} -x_{i,\theta_i+k} v_{\theta_i+k-1}) \\
        &\hspace{3cm}+x_{i,\theta_i+k-1} v_{\theta_i+k-1} -x_{i,\theta_i+k} v_{\theta_i+k-1} -x_{i,\theta_i+k-1} v_{\theta_i} +x_{i,\theta_i+k} v_{\theta_i} \\
        &\ge x_{i,\theta_i+k-1} v_{\theta_i+k-1} -x_{i,\theta_i+k} v_{\theta_i+k-1} -x_{i,\theta_i+k-1} v_{\theta_i} +x_{i,\theta_i+k} v_{\theta_i} \\
        &= x_{i,\theta_i+k-1} (v_{\theta_i+k-1} -v_{\theta_i}) -x_{i,\theta_i+k}(v_{\theta_i+k-1} -v_{\theta_i}) \\
        &= (x_{i,\theta_i+k-1}-x_{i,\theta_i+k}) (v_{\theta_i+k-1} -v_{\theta_i}) \ge 0,
    \end{align*} 
    where the first inequality holds since $\theta_i+k-1$ is not incentivized to misreport its successor, and the last one holds because $x$ is monotone ($x_{i,\theta_i+k-1} \ge x_{i,\theta_i+k}$) and types are ordered by critical value ($v_{\theta_i+k-1} \ge v_{\theta_i}$).
    As a result:
    \begin{equation*}
        \beta_{i,\theta_i} -x_{i,\theta_i}v_{\theta_i} \ge \beta_{i,\theta_i+k-1}-x_{i,\theta_i+k-1}v_{\theta_i} \ge \beta_{i,\theta_i+k} -x_{i,\theta_i+k} v_{\theta_i},
    \end{equation*}
    where the first step is due to the inductive hypothesis and the second holds since $D \ge 0$.
    Consequently, an agent of type $\theta_i$ is not incentivized to misreport $\theta_i+k$.
    
    We now focus on the remaining BIC constraints, namely those preventing a type $\theta_i> \firstType_i$ from misreporting a previous type $\theta_i -k$. 
    Observe that $\beta_{i,\firstType_i} \ge \beta_{i,\firstType_i+1} \ge \dots \ge \beta_{i,\lastType_i} = x_{i,\lastType_i}v_{\lastType_i} \ge 0$.
    Consider the case $k=1$.
    The IC constraint imposes:
    \begin{equation*}
        \beta_{i,\theta_i} \ge \beta_{i,\theta_i-1} +(x_{i,\theta_i} -x_{i,\theta_i-1}) v_{\theta_i},
    \end{equation*}
    and we observe that:
    \begin{align*}
        \beta_{i,\theta_i-1} +(x_{i,\theta_i} -x_{i,\theta_i-1}) v_{\theta_i} &= \beta_{i,\theta_i} +(x_{i,\theta_i-1}-x_{i,\theta_i})v_{\theta_i-1} +(x_{i,\theta_i} -x_{i,\theta_i-1}) v_{\theta_i} \\
        &\le \beta_{i,\theta_i} +(x_{i,\theta_i-1}-x_{i,\theta_i})v_{\theta_i-1} +(x_{i,\theta_i} -x_{i,\theta_i-1}) v_{\theta_i-1} = \beta_{i,\theta_i},
    \end{align*}
    hence the BIC constraint is satisfied.
    For $k>1$ we proceed by induction: if the constraint is satisfied for $k$, then for $k+1$ we have:
    \begin{align*}
        \beta_{i,\theta_i-k-1} +(x_{i,\theta_i} -x_{i,\theta_i-k-1}) v_{\theta_i} &= \beta_{i,\theta_i-k} +(x_{i,\theta_i-k-1} -x_{i,\theta_i-k})v_{\theta_i-k-1} +(x_{i,\theta_i} -x_{i,\theta_i-k-1}) v_{\theta_i} \\
        &\le \beta_{i,\theta_i-k} +(x_{i,\theta_i-k-1} -x_{i,\theta_i-k})v_{\theta_i} +(x_{i,\theta_i} -x_{i,\theta_i-k-1}) v_{\theta_i} \\
        &= \beta_{i,\theta_i-k} +(x_{i,\theta_i} -x_{i,\theta_i-k})v_{\theta_i} \\
        &\le \beta_{i,\theta_i},
    \end{align*}
    where the first equality is due to Equation~\ref{eq:beta_expected_profit_max}, the first inequality since $v_{\theta_i-k-1} \le v_{\theta_i}$, and the last one by the inductive hypothesis.
    As a result, all the BIC constraints are satisfied, and the mechanism is both BIC and IR.

    Finally, we notice that the IR constraint for type $\lastType_i$ is binding.
    Therefore, one cannot design an IR mechanism with a smaller value of $\beta_{i,\lastType_i}$.
    Similarly, all the BIC constraints between types $\theta_i < \lastType_i$ and $\theta_i+1$ are binding, preventing us from reducing the ex-interim payments $\beta_{i,\theta_i}$ for these types.
    We conclude that this choice of $\beta$ minimizes the payments among all the BIC and IR mechanisms.
    By Equation~\eqref{eq:profit_in_beta_and_exinterim}, it also maximizes the profit, concluding the proof.    
\end{proof}
    \section{Ironing}
\label{app:ironing}

Thanks to Lemma~\ref{lem:regular_optimal}, Lemma~\ref{lem:total_monotonicity}, Lemma~\ref{lem:beta_given_x}, and Lemma~\ref{lem:virtual_utilities}, the allocation rule of an optimal mechanism is the solution to Program~\ref{prog:alloc_virtual_utilites}.
The value of the solution corresponds to the profit of the mechanism, and its contracts are computed 
according to Equation~\eqref{eq:beta_expected_profit_max}. 
This problem has been widely studied in the literature (see, e.g.,~\cite{myerson1981optimal,hartline2008optimal,cai2016duality}), and is solved by substituting the virtual utilities with suitable ironed virtual utilities.
For completeness, we provide the construction here for our specific setting with finite types, following the order $>$ defined in Section~\ref{sec:single_task}. 

The ironing procedure allows us to compute the ironed virtual utilities $\bar\varphi_i(\theta_i)$, so that Program~\eqref{prog:alloc_virtual_utilites} has the same optimal solution and value as the following program:
\begin{equation}\label{prog:alloc_ironed_virtual_utilites}
\begin{aligned}
    \max_{x} &\sum_{i \in [N]} \sum_{\theta_i \in \Theta_i} \mu_{\theta_i} x_{i,\theta_i} \bar\varphi_i(\theta_i) \\
    \text{s.t.} \;&\text{$x$ is an allocation rule}
\end{aligned}
\end{equation}
For every agent  $i \in [N]$, the ironed virtual utilities are constructed as follows.
Define:
\begin{align*}
    F(\theta_i) \coloneqq \sum_{\theta'_i=\theta_i}^{\lastType_i} \mu_{\theta'_i}, \quad H(\theta_i) \coloneqq \sum_{\theta'_i=\theta_i}^{\lastType_i} \mu_{\theta'_i} \varphi_i(\theta'_i),
\end{align*}
for every type $\theta_i \in \typeSet_i$.
For ease of notation, let also $F(\lastType_i+1)=H(\lastType_i+1)=0$, and observe that $F$ is monotonically decreasing.
% Then, let $G:[0,F(\firstType_i)] \rightarrow \mathbb{R}$ be the lower convex hull of $\{(F(\theta_i),H(\theta_i)\}_{\theta_i=\firstType}^{\lastType_i+1}$.
Then, let $G:[0,F(\firstType_i)] \rightarrow \mathbb{R}$ be defined such that $\{(z,G(z)): z \in [0,F(\firstType_i)]\}$ is the lower convex hull of $\{(F(\theta_i),H(\theta_i))\}_{\theta_i=\firstType_i}^{\lastType_i+1}$.
Define the ironed utilities as the corresponding slopes of $G$ (observe that $F(\lastType_i+1) < F(\lastType_i) < \dots <F(\firstType_i)$):
\begin{equation}
    \label{eq:ironed_utilities}
    \bar\varphi_i(\theta_i) \coloneqq \frac{G(F(\theta_i))-G(F(\theta_i+1))}{F(\theta_i)-F(\theta_i+1)} = \frac{G(F(\theta_i))-G(F(\theta_i+1))}{\mu_{\theta_i}},
\end{equation}
for every $\theta_i \in \typeSet_i$.
Notice that the virtual utilities are (weakly) monotonically decreasing. 
%Basically the slops of a convex hull from left to right are (weakly) monotone increasing when taken from left to right.
%Notie howoever that we have F(\theta_1)>F(\theta_2)>..., so we are actually taking them from right to left
%But then we make everything symmetric to get the reverse order, so we go monotone decreasing. 
%Notice that the symmetry is equivalent to multiply the slopes by -1 (because you take the slope from right point to left and not from left to right), and then enumerate them in reverse order.

\begin{lemma}
    The allocation rule $\bar x^\star$ defined by Equation~\eqref{eq:optimal_allocation_ironing} is an optimal feasible solution of Program~\eqref{prog:alloc_virtual_utilites}. Therefore, the mechanism $(\bar x^\star,\beta^\star)$, where $\beta^\star$ satisfies Equation~\eqref{eq:beta_expected_profit_max}, is optimal.
\end{lemma}
\begin{proof}
    Let $x^\star$ be an optimal solution to Program~\ref{prog:alloc_virtual_utilites} and $\bar x^\star$ be computed according to Equation~\eqref{eq:optimal_allocation_ironing}.
    As a first step, we observe the following.
    $x^\star$ is also a feasible solution to Program~\ref{prog:alloc_ironed_virtual_utilites}, since this problem relaxes the constraints of Program~\ref{prog:alloc_virtual_utilites}.    
    Furthermore, $\bar x^\star$ is both an optimal solution to Program~\ref{prog:alloc_ironed_virtual_utilites} and a monotone allocation rule, and thus a feasible solution to Program~\ref{prog:alloc_virtual_utilites}.
    Indeed, it maximizes the objective function of Program~\ref{prog:alloc_ironed_virtual_utilites} and it is monotone.
    We will show that $\bar x^\star$ is also optimal for Program~\ref{prog:alloc_virtual_utilites}.
    
    % It may make sense to write the summation from "the last type to the first", a the summation by parts is used in that way: basically everything works by considering the types in reverse order. Namely, when using Abel's summation, we extract the term of the first type, not of the last.
    Now let $x$ be any monotone allocation rule (feasible solution to Program~\ref{prog:alloc_virtual_utilites}), and define $V(x)$ and $\bar V(x)$ as the value of the solution in Program~\eqref{prog:alloc_virtual_utilites} and Program~\eqref{prog:alloc_ironed_virtual_utilites}, respectively.
    We now study the difference between $V(x)$ and $\bar V(x)$.
    In order to do so, we rewrite $V(x)$ as follows:
    \begin{align*}
        V(x) &= \sum_{i \in [N]} \sum_{\theta_i=\firstType_i}^{\lastType_i} \mu_{\theta_i} x_{i,\theta_i} \varphi_i(\theta_i) \\
        &= \sum_{i \in [N]}\sum_{\theta_i=\firstType_i}^{\lastType_i} (H(\theta_i)- H(\theta_i+1))x_{i,\theta_i} \\
        % &=\sum_{i \in [N]} \left( (H(\firstType_i) -H(\lastType_i+1)) x_{i,\firstType_i} -\sum_{\theta_i=\firstType_i+1}^{\lastType_i} H(\theta_i) (x_{i,\theta_i-1}-x_{i,\theta_i}) \right),
        &=\sum_{i \in [N]} \left( H(\firstType_i) x_{i,\firstType_i} -\sum_{\theta_i=\firstType_i+1}^{\lastType_i} H(\theta_i) (x_{i,\theta_i-1}-x_{i,\theta_i}) \right),
    \end{align*}
    where the last equality holds by Abel's summation formula (summation by parts).
    Similarly:
    \begin{align*}
        \bar V(x) &= \sum_{i \in [N]}\sum_{\theta_i=\firstType_i}^{\lastType_i} \mu_{\theta_i} x_{i,\theta_i} \bar\varphi_i(\theta_i) \\
        &=\sum_{i \in [N]} \sum_{\theta_i=\firstType_i}^{\lastType_i} (G(F(\theta_i)) -G(F(\theta_i+1)))x_{i,\theta_i} \\
        &=\sum_{i \in [N]} \left( G(F(\firstType_i)) x_{i,\firstType_i} -\sum_{\theta_i=\firstType_i+1}^{\lastType_i} G(F(\theta_i)) (x_{i,\theta_i-1}-x_{i,\theta_i}) \right).
    \end{align*}
    Furthermore, observe that $G(F(\firstType_i))=H(\firstType_i)$, since $G$ is the lower convex hull of a set whose rightmost point is $(F(\firstType_i),H(\firstType_i))$.
    Therefore the difference between the values of $x$ in the objective functions is:
    \begin{equation}
        \label{eq:proof_ironing_difference_values}
        % V(x) - \bar V(x) = H(\firstType_i) x_{i,\firstType_i} -\sum_{\theta_i=\lastType_i}^{\firstType_i+1} H(\theta_i) (x_{i,\theta_i-1}-x_{i,\theta_i}) - G(F(\firstType_i)) x_{i,\firstType_i} +\sum_{\theta_i=\lastType_i}^{\firstType_i+1} G(F(\theta_i)) (x_{i,\theta_i-1}-x_{i,\theta_i}) \\
        %= \sum_{\theta_i=\lastType_i}^{\firstType_i+1} G(F(\theta_i)) (x_{i,\theta_i-1}-x_{i,\theta_i}) -\sum_{\theta_i=\lastType_i}^{\firstType_i+1} H(\theta_i) (x_{i,\theta_i-1}-x_{i,\theta_i})
        \bar V(x) - V(x) =\sum_{i \in [N]} \sum_{\theta_i=\firstType_i+1}^{\lastType_i} (H(\theta_i)-G(F(\theta_i))) (x_{i,\theta_i-1}-x_{i,\theta_i}) \ge 0,
    \end{equation}
    where the inequality holds since $H(\theta_i) \ge G(F(\theta_i))$ by definition of $G$ and $x_{i,\theta_i-1} \ge x_{i,\theta_i}$.
    Therefore the value $V(x)$ of any feasible solution to Program~\ref{prog:alloc_virtual_utilites} is upper bounded by $\bar V(x)$.
    Specifically, $V(x^\star) \le \bar V(x^\star) \le \bar V(\bar x^\star)$, where the last inequality holds since $\bar x^\star$ is optimal for Program~\ref{prog:alloc_ironed_virtual_utilites}.

    We are now ready to prove that $\bar x^\star$ is an optimal solution for Program~\ref{prog:alloc_virtual_utilites}.
    Specifically, we argue that $V(x^\star) \le \bar V(x^\star) \le \bar V(\bar x^\star) = V(\bar x^\star)$, which implies that $\bar x^\star$ is optimal.
    The chain of inequalities holds thanks to Equation~\eqref{eq:proof_ironing_difference_values} ($V(x^\star) \le \bar V(x^\star)$) and the fact that $\bar x^\star$ is an optimal solution of Program~\ref{prog:alloc_ironed_virtual_utilites} ($\bar V(x^\star) \le \bar V(\bar x^\star)$).
    We thus have to prove the equality $\bar V(\bar x^\star) = V(\bar x^\star)$, \emph{i.e.}, that the difference $\bar V(\bar x^\star) - V(\bar x^\star)$ is zero.
    Observe that this difference can be computed by Equation~\eqref{eq:proof_ironing_difference_values}, with $x=\bar x^\star$, and it is the sum of one term $d(\theta_i) \coloneqq (H(\theta_i)-G(F(\theta_i))) (\bar x^\star_{i,\theta_i-1}-\bar x^\star_{i,\theta_i})$ for each agent $i \in [N]$ and type $\theta_i > \firstType_i$.
    We prove that $d(\theta_i)=0$ by considering two cases.
    If $G(F(\theta_i))=H(\theta_i)$, then $d(\theta_i)$ is zero.
    Otherwise, when $H(\theta_i) > G(F(\theta_i))$ (the opposite is not possible by construction of $G$), it must be the case that $\bar\varphi_i(\theta_i) = \bar\varphi_i(\theta_i-1)$, and therefore $\bar x^\star_{i,\theta_i-1}-\bar x^\star_{i,\theta_i} = 0$.
    As a result, $d(\theta_i)$ is always zero, and employing Equation~\eqref{eq:proof_ironing_difference_values} we have $V(\bar x^\star) = \bar V(\bar x^\star)$, concluding the proof.
\end{proof}

    \section{Deterministic Payment Schemes Are Optimal}
\label{app:no_rand}
In a multi-project problem, a general mechanism could randomize over both allocations and contracts.
Specifically, it could be defined by a randomized allocation rule $y : \Theta \rightarrow \Delta_{\tSet}$ and a randomized payment scheme $\eta :  \Theta \times \tSet \rightarrow \Delta(\mathbb{R}_+^N)$.
% Notice that wether contracts for different agents are correlated or not is irrelevant, so one can equivalently consider a different $\eta_i$ for each agent.
% Then the problem becomes (for brevity we write $E_{\gamma}$ in place of $E_{\alpha \sim \gamma}$):
Then the problem of computing an optimal BIC and IR mechanism becomes:
\begin{align*}
    \sup_{y,\eta} \;&\revenue(y,\eta)=\sum_{\theta \in \typeSet} \mu_\theta \sum_{t \in \tSet} y(t|\theta)\sum_{i \in [N]} p_{\theta_i,t_i} (1 -\mathbb{E}_{\alpha \sim \eta(t,\theta)}[\alpha_i]) \\
    \text{s.t. } &y(t|\theta)(\alpha_i p_{\theta_i,t_i} -c_{\theta_i,t_i}) \ge 0 && \hspace{-1.2cm}\forall i \in [N], t \in \tSet,\theta \in \typeSet, \alpha \in \text{supp}(\eta(t,\theta))\\
    &\sum_{\theta_{-i}} \mu_{\theta_{-i}} \sum_{t \in \tSet}y(t|\theta_i,\theta_{-i}) (\mathbb{E}_{\alpha \sim \eta(t,\theta_i,\theta_{-i})}[\alpha_i p_{\theta_i,t_i} -c_{\theta_i,t_i}]) \ge \\
    &\hspace{1cm}\sum_{\theta_{-i}} \mu_{\theta_{-i}} \sum_{t \in \tSet}y(t|\hat\theta_i,\theta_{-i}) \mathbb{E}_{\alpha \sim \eta(t,\hat\theta_i,\theta_{-i})}[[\alpha_i p_{\theta_i,t_i} -c_{\theta_i,t_i}]_+] &&\forall i \in [N], \theta_i,\hat\theta_i \in \typeSet_i.
\end{align*}
We now show that this is not necessary, \emph{i.e.}, randomized payments do not increase the profit.

\begin{proposition}
    \label{prop:multi_task_no_rand_contr}
    For every BIC and IR mechanism $(y,\eta)$ with a randomized payment scheme, there exists a BIC and IR mechanism $(y,\sigma)$ with a deterministic payment scheme that attains the same profit.
\end{proposition}
\begin{proof}
    Let $(y,\eta)$ be any BIC and IR randomized mechanism.
    Now consider the mechanism $(y,\sigma)$ that uses the same allocation function, but given an allocation $t \in \tSet$ and a type profile $\theta \in \typeSet$, uses a single contract $\sigma(t|\theta) \in \mathbb{R}_+^N$ defined as:
    \begin{equation*}
        \sigma(t|\theta) \coloneqq \mathbb{E}_{\alpha \sim \eta(t,\theta)}[\alpha].
    \end{equation*}
    By linearity of expectation, it is easy to see that $\revenue(y,\eta)=\revenue(y,\sigma)$.
    Furthermore, the mechanism $(y,\sigma)$ is IR.
    Regarding the BIC constraints, the revenue of an agent that reports truthfully is the same in both mechanisms.
    When an agent misreports, we observe:
    \begin{align*}
        \mathbb{E}_{\alpha \sim \eta(t,\hat\theta_i,\theta_{-i})}[[\alpha_i p_{\theta_i,t_i} -c_{\theta_i,t_i}]_+] &\ge [\mathbb{E}_{\alpha \sim \eta(t,\hat\theta_i,\theta_{-i})}[\alpha_i] p_{\theta_i,t_i} -c_{\theta_i,t_i}]_+  \\
        &=[\sigma_i(t|\hat\theta_i,\theta_{-i}) p_{\theta_i,t_i} -c_{\theta_i,t_i}]_+.
    \end{align*}
    Hence, the agent obtains at least as much utility from misreporting under $(y,\eta)$ as under $(y,\sigma)$.
    As $(y,\eta)$ is BIC, $(y,\sigma)$ is also BIC.
    We conclude that there is no need to use randomized payment schemes.
\end{proof}

    \section{Details Omitted from the Proof of Lemma~\ref{lem:multi_task_no_opt}}
\label{app:no_opt_multi_task}

% \bollo{Note: this is not enough to remove the max form the BIC constraints, as the task can change depending on the report of the others}
% I keep this file beacuse this lemma could be interesting, but we do nto use it anymore.
\begin{restatable}{lemma}{prova}
    \label{lem:contracts_ind_from_others_tasks}
    In the multi-task problem, for every BIC and IR mechanism $(y,\sigma)$ there is another mechanism $(y,\sigma')$ attaining at least the same profit and such that the payment to agent $i \in [N]$ given report $\theta \in \typeSet$ depends on the task $t_i \in [M]$ assigned to the agent, but not on the tasks $t_{-i}$ assigned to the others. 
    Formally:
    \begin{equation*}
        \sigma'_i(t|\theta) = \sigma'_i(t'|\theta) \quad \forall i \in [N], \theta \in \typeSet, t,t' \in \tSet : t_i = t'_i.
    \end{equation*}
\end{restatable}
\begin{proof}
    Suppose that $(y,\sigma)$ violates the condition above, \emph{i.e.}, there exists some $\bar t,\bar t' \in \tSet : t_i = t'_i$ and type report $\bar \theta \in \typeSet$ such that $\sigma_i(\bar t|\bar \theta) \neq \sigma_i(\bar t'|\bar \theta)$.
    Suppose w.l.o.g. that $i=1$ and $\bar t_i=\bar t'_i=1$, and let $C \subseteq \tSet$ be the set of allocations such that $t_1 = 1$.
    Suppose also that $y(\bar t|\bar \theta)>0$ (payments for allocations that are not employed can be set arbitrarily).
    Let also $W = \sum_{t \in C} y(t|\bar \theta)$.
    Define a new payment $\sigma'$ as follows:
    \begin{align*}
        \sigma'_1(t|\bar \theta) &= \mathbb{E}_{t \sim y(\bar\theta)}[\sigma(t|\bar \theta)|t_1=1] = \frac{1}{W}\sum_{t \in C} y(t|\bar \theta)\sigma_1(t|\bar \theta) &\forall t \in C,
    \end{align*}
    and $\sigma'_i(t|\theta) = \sigma_i(t|\theta)$ in the other cases.
    We now show that $(y,\sigma')$ is BIC, IR, and attains at least the same profit as $(y,\sigma)$.

    We observe that all IR and IC constraints for agents different from 1 are satisfied, since their payments were not modified.
    Regarding agent 1, the new mechanism is IR, as the payments when the agent is assigned task 1 are the average of those employed in the original IR mechanism.
    We thus focus on the IC constraints.
    When the agent $i=1$ reports $\bar\theta_i$ truthfully, the contribution to its expected utility from receiving an allocation $t \in C$ given report $\bar\theta$ is:
    \begin{align*}
        &\mu_{\bar \theta_{-i}} \sum_{t \in C} y(t|\bar\theta) (p_{\bar\theta_i,1}\sigma'_i(t|\bar\theta_i,\theta_{-i}) -c_{\bar\theta_i,1}) \\
        &=\mu_{\bar \theta_{-i}} \sum_{t \in C} y(t|\bar\theta) (p_{\bar\theta_i,1}\mathbb{E}_{t \sim y(\bar\theta)}[\sigma(t|\bar \theta)|t \in C] -c_{\bar\theta_i,1}) \\
        &=\mu_{\bar \theta_{-i}} \mathbb{E}_{t \sim y(\bar\theta)} [\mathbf{1}\{t \in C\}(p_{\bar\theta_i,1}\mathbb{E}_{t \sim y(\bar\theta)}[\sigma(t|\bar \theta)|t \in C] -c_{\bar\theta_i,1})] \\
        &=\mu_{\bar \theta_{-i}} \mathbb{E}_{t \sim y(\bar\theta)} [\mathbf{1}\{t \in C\}\mathbb{E}_{t \sim y(\bar\theta)}[p_{\bar\theta_i,1}\sigma(t|\bar \theta) -c_{\bar\theta_i,1} |t \in C]] \\
        &=\mu_{\bar \theta_{-i}} \mathbb{E}_{t \sim y(\bar\theta)} [\mathbf{1}\{t \in C\}(p_{\bar\theta_i,1}\sigma(t|\bar \theta) -c_{\bar\theta_i,1})] \\
        &=\mu_{\bar \theta_{-i}} \sum_{t \in C} y(t|\bar\theta) (p_{\bar\theta_i,1}\sigma_i(t|\bar\theta_i,\theta_{-i}) -c_{\bar\theta_i,1}).
    \end{align*}
    Hence, the expected utility of agent $i$ reporting $\bar\theta_i$ truthfully is the same in both mechanisms $(y,\sigma)$ and $(y,\sigma')$.
    Formally, $U_i(\bar\theta_i|\bar\theta)=U'_i(\bar\theta_i|\bar\theta)$.
    Furthermore, if the agent misreports $\bar\theta_i$ in place of $\hat\theta_i$, we have:
    \begin{align*}
        &\mu_{\bar \theta_{-i}} \sum_{t \in C} y(t|\bar\theta) [p_{\hat\theta_i,1}\sigma'_i(t|\bar\theta_i,\theta_{-i}) -c_{\hat\theta_i,1}]_+ \\
        &=\mu_{\bar \theta_{-i}} \sum_{t \in C} y(t|\bar\theta) [p_{\hat\theta_i,1}\mathbb{E}_{t \sim y(\bar\theta)}[\sigma(t|\bar \theta)|t \in C] -c_{\hat\theta_i,1}]_+ \\
        &=\mu_{\bar \theta_{-i}} \mathbb{E}_{t \sim y(\bar\theta)} [\mathbf{1}\{t \in C\}[p_{\hat\theta_i,1}\mathbb{E}_{t \sim y(\bar\theta)}[\sigma(t|\bar \theta)|t \in C] -c_{\hat\theta_i,1}]_+] \\
        &=\mu_{\bar \theta_{-i}} \mathbb{E}_{t \sim y(\bar\theta)} [\mathbf{1}\{t \in C\}[\mathbb{E}_{t \sim y(\bar\theta)}[p_{\hat\theta_i,1}\sigma(t|\bar \theta) -c_{\hat\theta_i,1} |t \in C]]_+] \\
        &\le\mu_{\bar \theta_{-i}} \mathbb{E}_{t \sim y(\bar\theta)} [\mathbf{1}\{t \in C\}[p_{\hat\theta_i,1}\sigma(t|\bar \theta) -c_{\hat\theta_i,1}]_+] \\
        &=\mu_{\bar \theta_{-i}} \sum_{t \in C} y(t|\bar\theta) [p_{\hat\theta_i,1}\sigma_i(t|\bar\theta_i,\theta_{-i}) -c_{\hat\theta_i,1}]_+,
    \end{align*}
    where the inequality holds by Jensen inequality.
    Therefore, we have $U'_i(\hat\theta_i|\bar\theta) \le U_i(\hat\theta_i|\bar\theta)$, and as a result the BIC constraints are satisfied and the new payment scheme is BIC and IR.
    % Here U(deviation) is computed allowing to opt out if convininent

    In order to conclude the proof, we compare the profit of $(y,\sigma)$ and $(y,\sigma')$.
    The contribution to $\profit(y,\sigma')$ from agent $i=1$ when the type report is $\bar\theta$ and the allocation is in $C$ is:
    \begin{align*}
        \sum_{t \in C} y(t|\bar\theta) p_{\bar\theta_i,i} (1-\sigma'_i(t|\bar\theta)) &=\mathbb{E}_{t \sim y(\bar\theta)} [\mathbf{1}\{t \in C\} p_{\bar\theta_i,i} (1-\mathbb{E}_{t \sim y(\bar\theta)}[\sigma(t|\bar \theta)|t_1=1])] \\
        &=\mathbb{E}_{t \sim y(\bar\theta)} [\mathbf{1}\{t \in C\}p_{\bar\theta_i,i} (1-\sigma_i(t|\bar\theta))] \\
        &=\sum_{t \in C} y(t|\bar\theta) p_{\bar\theta_i,i}(1-\sigma_i(t|\bar\theta)),
    \end{align*}
    which is the same contribution provided to $\profit(y,\sigma)$.
    We conclude that $\profit(y,\sigma)=\profit(y,\sigma')$.
\end{proof}

% \begin{restatable}{claim}{CompleteProofNoOpt}
%     \label{cl:final_no_opt_proof}
%     In the instance considered in the proof of Lemma~\ref{lem:multi_task_no_opt}, any BIC and IR mechanism $(y,\sigma)$ such that $y((2,1)|\theta)>0$ and $y((\varnothing,1)|\theta)=0$ is not optimal.
% \end{restatable}
% \begin{proof}
%     Without loss of generality, we can assume $\sigma_2((2,1)|\theta) = \sigma_2((\varnothing,1)|\theta)$ by Lemma~\ref{lem:contracts_ind_from_others_tasks}.
%     Now build a new mechanism $(y^\circ,\sigma^\circ)$ similar to the one used in the proof of Lemma~\ref{lem:multi_task_no_opt}.
%     Specifically, move a portion $\epsilon/2$ of probability mass from $y((2,1)|\theta)$ to $y((\varnothing,1)|\theta)$, and change the contract $\sigma_1((2,1)|\theta)$.
%     Then, the argument follows the same steps as the one provided in the proof of Lemma~\ref{lem:multi_task_no_opt}.
% \end{proof}

    % \input{appendix/stuff}

\end{document}